\documentclass{article}[11pt]
\usepackage[utf8]{inputenc}

\usepackage{xcolor}
\definecolor{blueviolet}{rgb}{0.2, 0.2, 0.6}
\definecolor{webgreen}{rgb}{0,.5,0}
\definecolor{webbrown}{rgb}{.6,0,0}
\usepackage[pdftex,
  bookmarks=false,
  colorlinks=true, 
  urlcolor=webbrown,
  linkcolor=blueviolet, 
  citecolor=webgreen,
  pdfstartpage=1,
  pdfstartview={FitH},  
  bookmarksopen=false,
  pagebackref
  ]{hyperref}
\renewcommand*{\backref}[1]{}
\renewcommand*{\backrefalt}[4]{%
    \ifcase #1%
          \or Cited on page~#2.%
          \else Cited on pages~#2.%
    \fi%
    }
\usepackage{amsmath}
\usepackage{amssymb}
\usepackage{amsfonts}
\usepackage{amsthm}
\usepackage{zref-clever}
\zcsetup{cap = true}
\newcommand{\Cref}{\zcref}
\allowdisplaybreaks 

\usepackage{enumerate}
\usepackage{enumitem}
\usepackage{comment}
\usepackage{xpatch}
\usepackage{graphicx}
\usepackage{tabularx}
\usepackage{braket}
\usepackage{amsthm}
\usepackage{tikz}
\usepackage{commath}
\usepackage{mathtools}
\usepackage{qcircuit}
\usepackage{algorithm}
\usepackage{algpseudocodex}[indLines = true,italicComments = false, ]
\usepackage{tabto}
\usepackage{pgf-umlsd}
\usepackage{bm}
\usepackage{bbm}
\usepackage{multirow, multicol}
\usepackage{float}
\usepackage{pdfpages}
\usepackage{caption}
\usepackage{subcaption}
\usepackage{makecell}
\usepackage{cite}
\usepackage{tablefootnote}

\let\originalleft\left
\let\originalright\right
\renewcommand{\left}{\mathopen{}\mathclose\bgroup\originalleft}
\renewcommand{\right}{\aftergroup\egroup\originalright}

\newcommand{\mc}[1]{\mathcal{#1}}
\newcommand{\mbb}[1]{\mathbb{#1}}

\newcommand{\mrm}[1]{\mathrm{#1}}

\newcommand{\msf}[1]{\mathsf{#1}}

\newcommand{\divchi}{\mathrm{d}_{\chi^2}}

\newcommand{\bigo}{\mathcal{O}}

\newcommand{\mmstate}{\rho_{\mathrm{mm}}}

\newcommand{\diag}{\mathrm{diag}}

\newcommand{\stail}{S_{\mathrm{tail}}}

\newcommand{\bftheta}{\boldsymbol{\theta}}

\newcommand{\eps}{\epsilon}

\newcommand{\tr}{\mathrm{tr}}

\newcommand{\symsum}{\sum_{\mathrm{sym}}}
\newcommand{\swap}{\mathrm{SWAP}}
\newcommand{\purify}{\Phi_{\mathrm{purify}}}
\newcommand{\bfrho}{\boldsymbol{\rho}}
\newcommand{\bfz}{\boldsymbol{z}}
\newcommand{\op}{\mathrm{op}}
\newcommand{\Lower}{\mathrm{Lower}}
\newcommand{\rank}{\mathrm{rank}}

\newtheorem{theorem}{Theorem}[section]
\newtheorem{proposition}[theorem]{Proposition}
\newtheorem{lemma}[theorem]{Lemma}
\newtheorem{corollary}[theorem]{Corollary}
\newtheorem{claim}[theorem]{Claim}

\newtheorem{fact}[theorem]{Fact}

\AddToHook{env/theorem/begin}{%
\zcsetup{countertype={theorem=theorem}}}
\AddToHook{env/proposition/begin}{%
\zcsetup{countertype={theorem=proposition}}}
\AddToHook{env/lemma/begin}{%
\zcsetup{countertype={theorem=lemma}}}
\AddToHook{env/corollary/begin}{%
\zcsetup{countertype={theorem=corollary}}}
\AddToHook{env/claim/begin}{%
\zcsetup{countertype={theorem=claim}}}
\AddToHook{env/problem/begin}{%
\zcsetup{countertype={theorem=problem}}}
\AddToHook{env/conjecture/begin}{%
\zcsetup{countertype={theorem=conjecture}}}
\AddToHook{env/fact/begin}{%
\zcsetup{countertype={theorem=fact}}}
\zcRefTypeSetup{fact}{
Name-sg = Fact ,
name-sg = fact ,
Name-pl = Facts ,
name-pl = facts ,
}
\AddToHook{env/question/begin}{%
\zcsetup{countertype={theorem=question}}}
\AddToHook{env/claim/begin}{%
\zcsetup{countertype={theorem=claim}}}
\zcRefTypeSetup{claim}{
Name-sg = Claim ,
name-sg = claim ,
Name-pl = Claims ,
name-pl = claims ,
}

\theoremstyle{definition}
\newtheorem{definition}[theorem]{Definition}

\newtheorem{remark}[theorem]{Remark}
\AddToHook{env/definition/begin}{%
\zcsetup{countertype={theorem=definition}}}
\AddToHook{env/notation/begin}{%
\zcsetup{countertype={theorem=notation}}}
\AddToHook{env/example/begin}{%
\zcsetup{countertype={theorem=example}}}
\AddToHook{env/assumption/begin}{%
\zcsetup{countertype={theorem=assumption}}}
\AddToHook{env/remark/begin}{%
\zcsetup{countertype={theorem=remark}}}

\numberwithin{equation}{section}

\definecolor{chiragcolor}{RGB}{66,175,210}
\definecolor{sitancolor}{rgb}{1,0,0} 

\makeatletter
\newif\if@comments
\newcommand{\chirag}[1]{\if@comments\textcolor{chiragcolor}{[CW: #1]}\fi}
\newcommand{\sitan}[1]{\if@comments\textcolor{sitancolor}{[SC: #1]}\fi}

\newcommand\enablecomments{\@commentstrue}
\newcommand\disablecomments{\@commentsfalse}
\enablecomments

\makeatother
\usepackage[a4paper, left=1in, right=1in, top=1in, bottom=1in]{geometry}

\title{Optimal Quantum State Testing Even with Limited Entanglement}

\author{
Chirag Wadhwa\thanks{University of Edinburgh. Email: \href{mailto:chirag.wadhwa@ed.ac.uk}{chirag.wadhwa@ed.ac.uk}.}
\and 
Sitan Chen\thanks{Harvard University. Email: \href{mailto:sitan@seas.harvard.edu}{sitan@seas.harvard.edu}.}
}
\date{}

\makeatletter
\renewcommand\paragraph{\@startsection{paragraph}{4}{\z@}%
                                    {1.8ex \@plus1ex \@minus.2ex}%
                                    {-1em}%
                                    {\normalfont\normalsize\bfseries}}
\makeatother

\begin{document}

\maketitle

\disablecomments
\begin{abstract}

In this work, we consider the fundamental task of quantum state certification: given copies of an unknown quantum state $\rho$, test whether it matches some hypothesis state $\sigma$ or is $\epsilon$-far from it. For certifying $d$-dimensional states, $\Theta(d/\eps^2)$ copies of $\rho$ are known to be necessary and sufficient \cite{o2015quantum,buadescu2019quantum}. However, the algorithm achieving this copy complexity makes fully entangled measurements over all $\bigo(d/\eps^2)$ copies of $\rho$. Often, one is interested in certifying states to a high precision; this makes performing such joint measurements intractable even for certifying low-dimensional states. Thus, we aim to study whether one can obtain optimal rates for quantum state certification and related testing problems while only performing measurements on $t$ copies at once, for some $1 < t \ll d/\eps^2$. While it is well-understood how to use intermediate entanglement to achieve optimal quantum state \emph{learning}, to date, the only protocol known to achieve the optimal rate for testing is the one using fully entangled measurements.

Our main result is a near-optimal protocol for state certification that only performs measurements on $t = d^2$ copies at once. In the high-precision regime, i.e., for $\eps < \frac{1}{\sqrt{d}}$, this is a strict improvement over the algorithm of \cite{buadescu2019quantum}, which performs a joint measurement on all $\bigo(d/\eps^2)$ copies. More generally, our algorithm achieves a smooth copy complexity tradeoff as a function of $t$, which interpolates between the single-copy rate at $t = 1$ and the optimal rate at $t = d^2$. We also extend our techniques to develop new algorithms for the related tasks of mixedness testing and purity estimation, and show tradeoffs achieving the optimal rates for these problems at $t = d^2$ as well. Our algorithms are based on novel reductions from testing to learning and leverage recent advances in quantum state tomography in a non-black-box fashion. Typically, testing is a strictly easier problem than learning, so it is somewhat surprising that we achieve optimal rates for testing using algorithms for learning quantum states. 

We complement our testing upper bounds with a (nearly) matching lower bound that holds in the high-precision regime, implying that one must take $t \geq \Omega(d^2)$ to achieve the optimal rate for certification. We also prove a high-precision lower bound for a stronger class of measurements; this result shows that even with these stronger measurements, the same entanglement threshold $t \geq \Omega(d^2)$ is necessary for statistically-optimal certification.

\end{abstract}

\newpage

\tableofcontents

\newpage

\section{Introduction}\label{sec:intro}

A fundamental task within quantum information is to determine whether the quantum data one has access to is coming from the correct reference state. This is an essential primitive in quantum device characterization and calibration, and more broadly in testing hypotheses about quantum systems through experiments. In this work, we consider the following well-studied formalization of this problem of \emph{quantum state certification}. 
Given the description of a fixed $d$-dimensional hypothesis state $\sigma$ and $n$ copies of an unknown state $\rho$, one aims to determine whether the two states are identical or $\eps$-far in trace distance.

In the most general setting, the copy complexity of certifying $d$-dimensional states is $n = \Theta(d/\eps^2)$ \cite{o2015quantum,buadescu2019quantum}.
While this guarantee is compelling because it shows that certification can be done with quadratically fewer copies than are needed to \emph{learn} the full state $\rho$, unfortunately the algorithms used to achieve this copy complexity perform fully entangled measurements on all $\bigo(d/\eps^2)$ copies at once.
This presents significant hurdles for practical deployment.
For instance, the hypothesis state might be a specific resource state in the context of high-fidelity magic-state distillation or quantum metrology, or a thermal state that a quantum simulator purports to prepare which must be verified to high accuracy.
In these applications, it is imperative that state certification be done to sufficiently high precision, but even for $\eps = 0.001$, this requires operating on $\ge 10^6 d$ copies of $\rho$ at once, rendering the $\bigo(d/\epsilon^2)$ guarantee impractical \emph{even for low-dimensional quantum states}.
This motivates our central question: 
\begin{center}
    \emph{Can we achieve the optimal rate for state certification without fully entangled measurements?}
\end{center}

In the restricted setting of \emph{single-copy} measurements, the optimal rate is known to be unachievable~\cite{bubeck2020entanglement,chen2022toward,chen2022tightStateCertification,liu2024role}. In fact, any protocol which performs a fixed sequence of single-copy measurements requires $\Theta(d^2 / \eps^2)$ copies. While this rules out the possibility of achieving the $\bigo(d/\eps^2)$ copy complexity with single-copy measurements, in principle, there could be an intermediate $1 < t \ll d/\eps^2$, such that the desired copy complexity is achievable using only \emph{$t$-copy measurements}. In this setting, the tester receives $n$ copies of the state $\rho^{\otimes t}$, and can only operate on one copy of $\rho^{\otimes t}$ at a time. More generally, for the $i$th measurement, the tester may choose to measure only $t_i$ copies at once, for any $1 \leq t_i \leq t$, and the overall number of copies used would be the sum of these $t_i$'s.

Unfortunately, the landscape of state certification with such $t$-copy measurements, and even that of generic testing problems more broadly, is largely uncharacterized. The only known bounds are due to Chen, Cotler, Huang, and Li \cite{chen2021hierarchy}, who only showed \emph{lower} bounds for this problem. Their results hold  for $t$ not too large, and for measurements chosen either adaptively or only non-adaptively; see \Cref{sec:related-work} for a detailed discussion of their results. On the achievability side, we are only aware of a folklore algorithm that succeeds using $\bigo\bigl(\max\bigl\{\frac{d^2}{t\eps^4}, \frac{d}{\eps^2}\bigr\}\bigr)$ copies by batching the \cite{buadescu2019quantum} tester\footnote{As we haven't seen this tester appear in print, we present the algorithm and prove this upper bound in \Cref{sec:batching-BOW-tester}.}. While this bound does present a smooth tradeoff with $t$, it only achieves the optimal rate when $t = \Theta(d/\eps^2)$, providing no resource savings at all. Further, for smaller $t$, this upper bound has an expensive $1/\eps^4$-dependence, and it does not even recover the optimal $\bigo(d^2/\eps^2)$ rate for $t = 1$. Prior to our work, no other non-trivial upper bounds were known for any intermediate $t > 1$. Thus, it was completely unclear whether one could achieve the optimal rate for state certification with limited entanglement. 

The dearth of results for quantum state certification in the limited entanglement regime stands in stark contrast to the situation for quantum state \emph{tomography}, for which essentially optimal tradeoffs have been established. Here, the goal is to
learn the full description of a state to $\eps$ trace distance. In the general setting, $\Theta(d^2/\eps^2)$ copies are known to be optimal for tomography \cite{o2016efficient,haah2016sample}. With unentangled measurements, the optimal complexity is $\Theta(d^3/\eps^2)$ \cite{kueng2017low,chen2023does}. The complexity of tomography with $t$-copy measurements is also well-understood, and recent work \cite{chen2024optimalTradeoffsTomography,pelecanos2025debiased,pelecanos2025mixed} has shown upper bounds that interpolate smoothly between these extremal rates, achieving the optimal rate at $t = d^2$. This provides a concrete resource advantage over the algorithms of \cite{o2016efficient,haah2016sample} that measure all $\bigo(d^2/\eps^2)$ copies coherently. Crucially, this entanglement threshold, $t = d^2$, is $\eps$-\emph{independent}, i.e., the entanglement cost of optimal tomography does not grow with higher precision.

While we discuss prior results on $t$-copy tomography in more detail in \Cref{sec:related-work}, we would like to note that the work of Pelecanos, Spilecki, Tang, and Wright \cite{pelecanos2025mixed} uses a recent \emph{random purification channel} \cite{tang2025conjugate} to achieve their upper bounds. Essentially, this reduces mixed-state tomography to pure-state tomography, allowing one to design intuitive algorithms possessing many desirable properties. It is unclear whether random purification can also be used to similarly simplify testing; we discuss the barriers to such an approach in \Cref{sec:discussion}. Nevertheless, we leverage the state tomography algorithm of \cite{pelecanos2025mixed} as a key component of our testing algorithms; we do this not for any direct simplification offered by random purification, but because their estimator has a succinct second-moment expression.

\subsection{Our results} \label{sec:results}
We formally define state certification in \Cref{def:state-certification}. We state our results for constant success probability, say, at least $.99$; for an arbitrary $\delta > 0$, this can be boosted to $1-\delta$ using $\bigo(\log(1/\delta))$ repetitions.

Before stating our results, let us define the kinds of measurements we consider in this work. Throughout this work, we only allow algorithms to perform \emph{non-adaptive} measurements. The most general measurements here may be drawn at random, 
and several different measurements may be drawn from a shared source of randomness. When such shared randomness is unavailable, we say that these are \emph{private} measurements. Finally, if even private randomness is not used, we say that the measurements are \emph{fixed}. Now, let us state our main upper and lower bounds for state certification.
\begin{theorem}[Near-optimal fixed/private-measurement complexity]
    \label{thm:intro-state-certification-main}
    Let $d \geq 2, t \geq 1, \eps > 0$. There exists an algorithm for $\eps$-trace norm certification of $d$-dimensional states using only fixed $t$-copy measurements with copy complexity at most\footnote{Throughout, we will use $\Tilde{\bigo},\Tilde{\Omega}$ to suppress polylogarithmic factors in $d, \eps$.}
    \begin{equation}
        \tilde{\bigo}\Bigl(\max\Bigl\{\frac{d^2}{\sqrt{t}\epsilon^2}, \frac{d}{\eps^2}\Bigr\}\Bigr). 
    \end{equation}
    Moreover, for $d$ larger than an absolute constant and $\epsilon \leq \bigo(1/d^2 t^{3/2})$, the copy complexity of certification with private $t$-copy measurements is at least
    \begin{equation}
        \Omega\Bigl(\max\Bigl\{\frac{d^2}{\sqrt{t}\epsilon^2}, \frac{d}{\eps^2}\Bigr\}\Bigr).
    \end{equation}
\end{theorem}

We have thus obtained a near-optimal characterization of state certification with fixed and private $t$-copy measurements in the high-precision regime.

We note that for $t = 1$, our algorithm achieves a copy complexity of $\Tilde{\bigo}(d^2/\eps^2)$, which is optimal (up to log factors) for certification with fixed single-copy measurements. In fact, we also apply the approach used in our upper bound to directly provide a new tester that recovers the tight $\bigo(d^2/\eps^2)$ bound of \cite{yu2021sample,liu2024role} in this setting (see \Cref{sec:closeness-testing-uniform-povm}).

Let us also compare our results in \Cref{thm:intro-state-certification-main} to the folklore upper bound mentioned previously. Note that $\frac{d^2}{\sqrt{t}\eps^2}$ and $\frac{d^2}{t\eps^4}$ are incomparable, and trade-off at $t = 1/\eps^4$. Thus, \Cref{thm:intro-state-certification-main} and the folklore upper bound imply that the complexity of state certification with fixed $t$-copy measurements is at most 
\begin{equation}
    \Tilde{\bigo}\Bigl(\max\Bigl\{\frac{d}{\eps^2}, \min\Bigl\{\frac{d^2}{t\eps^4}, \frac{d^2}{\sqrt{t}\eps^2}\Bigr\}\Bigr\}\Bigr).
\end{equation}

We also note that our lower bound does not contradict the $\bigo(d^2/t\epsilon^4)$ upper bound, as the former only holds in a high-precision regime where the algorithm from \Cref{thm:intro-state-certification-main} already outperforms the folklore one. Our main result and this folklore upper bound thus demonstrate a phase transition in the complexity of state certification: in the high-precision regime, a $d^2/\sqrt{t}$ complexity is necessary (and almost sufficient); at the other extreme, with constant precision, the complexity scales as at most $d^2/t$.

For $t = d^2$, our algorithm succeeds using $\Tilde{\bigo}(d/\eps^2)$ copies, recovering the optimal rate up to log factors. For $\eps = o\left(\frac{1}{\sqrt{d}}\right)$, we can thus obtain near-optimal rates for certification using $o(d/\eps^2)$-copy measurements, answering our central question in the affirmative. Our high-precision lower bound also implies that $t = \Theta(d^2)$ is the optimal entanglement cost for state certification. This demonstrates that, somewhat surprisingly, optimal certification and tomography have the same entanglement cost, at least when only fixed measurements are available. 

We show that this is also the case with \emph{general} non-adaptive measurements, by proving the following lower bound when such stronger measurements are available.

\begin{theorem}[Non-adaptive lower bound]
\label{thm:intro-lower-shared}
    For $d$ larger than some absolute constant, $t \geq 1,$ and $ 0 < \eps \leq \bigo\left(\frac{1}{d^{3/2}t^{7/4}}\right)$, the copy complexity of $\eps$-trace norm certification of $d$-dimensional states with \emph{non-adaptive} $t$-copy measurements, even ones that use shared randomness, is at least 
    \begin{equation}
        \Omega\Bigl(\max\Bigl\{\frac{d^{3/2}}{t^{1/4}\eps^2}, \frac{d}{\eps^2}\Bigr\}\Bigr).
    \end{equation}
\end{theorem}

Indeed, this high-precision lower bound shows that one can only achieve the optimal rate with these stronger non-adaptive measurements for $t \geq \Omega(d^2)$, demonstrating that the use of shared randomness across measurements does not allow one to reduce the entanglement cost of statistically-optimal certification. We also note that this lower bound improves on that of \cite{chen2021hierarchy}, who showed $\Omega(d^{3/2}/t\epsilon^2)$ copies are necessary in this setting. 

We note that we have stated the \Cref{thm:intro-state-certification-main} upper bound in the worst case; however, we actually show a more general instance-dependent upper bound (see \Cref{thm:instance-dependent-tcopy-certification}), i.e., we provide an upper bound that depends on the hypothesis state $\sigma$. The worst-case upper bound is then obtained by setting $\sigma$ to be the maximally mixed state, $\mmstate \triangleq \frac{I}{d}$, which maximizes all $\sigma$-dependent quantities in \Cref{thm:instance-dependent-tcopy-certification}.

In fact, we show that for this special case of \emph{mixedness testing}, i.e., certifying the maximally mixed state, we can do away with the log factors of \Cref{thm:intro-state-certification-main}, obtaining an optimal complexity.

\begin{theorem}
    \label{thm:intro-mixedness-upper}
     Let $d \geq 2, t \geq 1, \eps > 0$. The copy complexity of $\eps$-trace norm mixedness testing of $d$-dimensional states using only fixed $t$-copy measurements is at most
    \begin{equation}
        \bigo\Bigl(\max\Bigl\{\frac{d^2}{\sqrt{t}\epsilon^2}, \frac{d}{\eps^2}\Bigr\}\Bigr).
    \end{equation}
    Moreover, for $d$ larger than an absolute constant and $\epsilon \leq \bigo(1/d^2 t^{3/2})$, the above upper bound is optimal.
\end{theorem}

Mixedness testing is the hardest instance of state certification, and also has complexity $\Theta(d/\eps^2)$ in the general setting \cite{o2015quantum}. Thus, our algorithm again achieves the optimal rate for this problem at $t = d^2$. Our algorithm exploits the fact that the maximally mixed state has the lowest purity among all quantum states by estimating the purity of the unknown state and comparing it to a suitable threshold. Our analysis of this estimator immediately yields an upper bound for multiplicative-error purity estimation:

\begin{theorem}
    \label{thm:intro-purity-upper}
     Let $d \geq 2, t \geq 1, \eps > 0$. Given copies of an unknown state $\rho$, there exists an algorithm to estimate $\tr(\rho^2)$ to multiplicative error at most $\eps$ using fixed $t$-copy measurements with copy complexity
    \begin{equation}
        \bigo\Bigl(\max\Bigl\{
        \frac{d^2}{\sqrt{t}\epsilon}, \frac{d}{\sqrt{t}\eps^2}, \frac{d}{\eps}, \frac{\sqrt{d}}{\eps^2}\Bigr\}\Bigr).
    \end{equation}
\end{theorem}

Note that with fully entangled measurements, $\Theta(\max\{\frac{d}{\eps},\frac{\sqrt{d}}{\eps^2}\})$ copies are known to be necessary and sufficient for this problem \cite{acharya2020estimating}. At the other extreme, with unentangled measurements, the best known complexity for multiplicative-error purity estimation is $\bigo(\max\{\frac{d^2}{\eps},\frac{d}{\eps^2}\})$ \cite{pelecanos2025beating}. Our upper bound interpolates precisely between these regimes, recovering the original bounds at $t \geq d^2$ and $t = 1$.

\subsection{Technical overview} \label{sec:tech-overview}

We will now present an overview of the techniques used in this work. However, before moving on to our techniques for state certification, let us first recall the analogous classical problem of distribution \emph{identity testing}. Here, one is given the description of some distribution $q$ over $\{1,\dots,d\}$, and samples from an unknown distribution $p$, and aims to test whether $p = q$ or whether $p,q$ are $\eps$-far according to some distance. By setting $q = u_d$, the uniform distribution over $\{1,\dots,d\}$, one recovers the problem of \emph{uniformity testing}. State certification and mixedness testing are the quantum analogues of these problems, and often, results for these quantum problems can be obtained by suitably generalizing techniques used to solve the classical ones, as will also be the case for some of our results.

In this overview, we start by presenting the ideas used to prove our mixedness testing upper bound. In fact, we present two distinct algorithms for mixedness testing. Both our protocols are based on a reduction from testing to learning: they repeatedly apply a tomography algorithm to copies of the unknown state and then use the outputs to estimate certain quantities of interest. 

The first of these estimators performs deterministic computations on the outputs obtained above. Interestingly, while we use an existing protocol for tomography, known bounds on the variance achieved by the tomography protocol do not suffice for us to control the variance of our final estimator. Instead, as we explain, we will need to develop refined estimates in order to achieve our claimed rate for mixedness testing. We overview our algorithm and analysis for mixedness testing and their application to purity estimation in \Cref{sec:upper-bound-mixedness-overview}. 

Our second tester actually solves the problem of state certification, and the result for mixedness testing follows as a corollary. Interestingly, unlike the previous tester, the estimator used here applies tomography not just to the unknown state, but also to the known \emph{hypothesis state} $\sigma$. Note that as $\sigma$ is known, one can also simulate these random measurement outcomes without direct access to copies of $\sigma$, resulting in a \emph{randomized} estimator that only consumes copies of $\rho$. While all this might appear to introduce unnecessary additional fluctuations to our estimator, surprisingly, it turns out that this randomization makes the variance of this estimator easier to bound than that of the previous one, but using it to get the upper bound in \Cref{thm:intro-state-certification-main} still requires additional work. We discuss this algorithm for state certification in detail in \Cref{sec:upper-bound-certification-overview}. 

Lastly, we provide an overview of our lower bound techniques in \Cref{sec:lower-bound-overview}.

\subsubsection{Mixedness testing via purity estimation} \label{sec:upper-bound-mixedness-overview}

While results for both distribution testing and quantum state testing are typically stated with respect to the $\ell_1$ and the trace norms, it is often easier to prove upper bounds for testing with respect to the $\ell_2$ and the Hilbert--Schmidt norms. Fortunately, standard norm inequalities imply a straightforward reduction from the former case to the latter (see \Cref{fact:hs-to-trace-norm-testing}). Thus, going forward, we will focus almost entirely on testing with respect to the Hilbert--Schmidt norm. Now, before discussing our algorithm for mixedness testing, let us recall the collision-based uniformity tester.

\paragraph{Uniformity testing via collision estimation:} The canonical algorithm for uniformity testing makes use of the observation that the uniform distribution has the lowest collision probability of all distributions over $[d]$. Recall that the collision probability of a distribution is given by $\sum_{i \in [d]} p_i^2$. In fact, this quantity can be represented conveniently as a function of how far the distribution is from being uniform:
\begin{equation}
    \sum_i p_i^2 = \frac{1}{d} + \|p - u_d\|_2^2.
\end{equation}
Thus, the uniform distribution has collision probability $\frac1d$, and for any distribution $\eps$-far from it in $\ell_2$ distance, this is at least $\frac1d + \eps^2$. Consequently, one can test for uniformity by estimating the collision probability and thresholding. For a full analysis of the uniformity testing upper bound using these arguments, we refer to \cite[Section 2.1.2]{canonne2022topics}.

\paragraph{Purity estimation and mixedness testing:} For a quantum state $\rho$, the quantity analogous to the collision probability is its purity, i.e., $\tr(\rho^2)$. Indeed, this quantity satisfies a similar property:
\begin{equation}
    \tr(\rho^2) = \frac1d + \|\rho - \mmstate\|_2^2,
\end{equation}
where $\|\cdot\|_2$ denotes the Schatten $2$-norm, i.e., the Hilbert--Schmidt norm. Thus, one can solve the mixedness testing problem by estimating the purity of a state to a suitable precision and then thresholding. 

Now, for the task of purity estimation, if one could perform $2$-copy measurements, an obvious idea would be to measure the $\swap$ observable on the state $\rho^{\otimes 2}$, as the expectation value of this observable is $\tr(\swap \cdot \rho^{\otimes 2}) = \tr(\rho^2)$, i.e., one obtains an unbiased purity estimator. When one can perform fully entangled measurements, this can be generalized to measuring the \emph{symmetrized swap observable}, i.e. 
\begin{equation}
    \swap^{(n)} \triangleq \frac{1}{\binom{n}{2}}\sum_{i,j \in [n], i < j} \swap_{i,j}.
\end{equation} The expectation value of this observable on $\rho^{\otimes n}$ is indeed $\tr(\rho^2)$, and it was shown by B\u adescu, O'Donnell, and Wright \cite{buadescu2019quantum} that this is the minimum variance observable for purity estimation given $n$ copies of $\rho$. They further showed that one can achieve the optimal rate for mixedness testing by estimating the purity in this manner (see their Proposition 5.2). Then, one approach to get an upper bound for $t$-copy mixedness testing would be to interpolate between these two algorithms. Such an algorithm would involve measuring $\swap^{(t)}$ on $t$ copies at once, and averaging over the outcomes obtained. This is essentially the idea behind the aforementioned folklore tester, and unfortunately only leads to an $\bigo\bigl(\max\bigl\{\frac{d^2}{t\eps^4}, \frac{d}{\eps^2}\bigr\}\bigr)$ upper bound (see \Cref{sec:batching-BOW-tester}). The reason for this $1/\eps^4$-dependence is the presence of an undesirable term in the variance of the above purity estimator. To improve on this rate, we will consider an entirely different way of estimating the purity which achieves a smaller variance. 

\paragraph{Purity estimation via state estimation:}  Our algorithm for purity estimation will make use of state estimation algorithms. For ease of exposition, we will assume that these estimators are unbiased, but our arguments can be modified, without much additional effort, to work with estimators with state-independent biases as well. Now, suppose we have two independent unbiased estimators $\hat{\rho}_1, \hat{\rho}_2$ for an unknown state $\rho$. Simply looking at the expectation of their inner product, we have
\begin{equation}
    \mathbb{E}[\tr(\hat{\rho}_1 \hat{\rho}_2)] = \tr(\mbb{E}[\hat{\rho}_1] \mbb{E}[\hat{\rho}_2]) = \tr(\rho^2),
\end{equation}
as $\hat{\rho}_1, \hat{\rho}_2$ are independent and unbiased. Thus, one can take any unbiased state estimator, apply it repeatedly, and compute pairwise inner products; this procedure gives us an unbiased estimator for the purity. More concretely, we will repeatedly apply an estimator $\mc{A}$ to $\rho^{\otimes t}$, and obtain estimates $\hat{\rho}_1, \dots, \hat{\rho}_n$. Then, to minimize the variance, we compute inner products of \emph{all} distinct pairs and average over them:
\begin{equation}
    \Bar{X} = \frac{1}{\binom{n}{2}}\sum_{i, j \in [n], i < j} \tr(\hat{\rho}_i \hat{\rho}_j). 
\end{equation}
Clearly, we have $\mathbb{E}[\Bar{X}] = \tr(\rho^2)$, and this algorithm only performs fixed $t$-copy measurements. To show that this estimator is sufficiently precise with high probability, we will bound the variance of this estimator. As we will show in \Cref{lem:collision-variance-generic}, this can be written in terms of the first two moments of the underlying state estimator $\mc{A}$. In general, when $\mc{A}$ is unbiased, we will show that
\begin{equation}
    \mbb{V}[\Bar{X}] = \bigo(n^{-1}) \cdot (\tr(\mbb{E}[\hat{\rho}^{\otimes 2}] \cdot \rho^{\otimes 2}) - \tr(\rho^2)^2) + \bigo(n^{-2}) \cdot (\tr(\mbb{E}[\hat{\rho}^{\otimes 2}]^2) - \tr(\rho^2)^2),
    \label{eq:overview-generic-variance}
\end{equation}
where the expectations above are over the random output $\hat{\rho} \gets \mc{A}(\rho^{\otimes t})$.

As before, our tester does not require the state estimator to be unbiased, but only that its bias is constant. We also note that our approach for purity estimation is similar to that of \cite{pelecanos2025beating}, who gave similar algorithms for estimating $\tr(\rho^k)$ with \emph{single-copy} measurements and integer $k \geq 2$. We believe one can also extend our approach to estimate generic moments $\tr(\rho^k)$ with $t$-copy measurements.

\paragraph{Estimators with second-moment bounds:} Now, to solve mixedness testing, the above discussion directs us towards state estimators that perform $t$-copy measurements and have succinct second-moment expressions. Two such estimators have been developed very recently by Pelecanos, Spilecki, Tang, and Wright~\cite{pelecanos2025debiased,pelecanos2025mixed}. When applied to $\rho^{\otimes t}$, the outputs $\hat{\rho}$ of both estimators satisfy
\begin{equation}
    \mathbb{E}[\hat{\rho}^{\otimes 2}] = \frac{t-1}{t}\cdot \rho^{\otimes 2} + \frac1t(\rho \otimes I + I \otimes \rho)\cdot \swap + \frac{\mathbb{E}[\ell(\lambda)]}{t^2} \swap - \Lower_\rho. \label{eq:overview-second-moment}
\end{equation}
While we discuss this expression in more detail in \Cref{sec:PTSW-moments}, let us note for now that $\Lower_\rho$ is an unspecified lower-order term that can be ignored for almost all known applications of these estimators. Unfortunately, our mixedness tester is the only exception we are aware of to the above norm. It turns out that neglecting the lower-order term and computing the variance via our \Cref{eq:overview-generic-variance}, one obtains an additional undesirable term yielding a $1/\eps^4$-dependence in the final copy complexity for mixedness testing. In particular, this extra term arises only in the calculation of $\tr(\mbb{E}[\hat{\rho}^{\otimes 2}] \cdot \rho^{\otimes 2})$ in \Cref{eq:overview-generic-variance}.

\paragraph{The \cite{pelecanos2025mixed} estimator:} To get rid of this undesirable term, we need a refined version of \Cref{eq:overview-second-moment} that explicitly incorporates the lower-order term. Of the two estimators, we will choose to refine the second-moment bound of the estimator of \cite{pelecanos2025mixed}, as this estimator and its analysis are significantly easier to understand than that of \cite{pelecanos2025debiased}. First, let us provide a brief description of the \cite{pelecanos2025mixed} estimator. We refer the interested reader to \Cref{sec:state-estimation} for a detailed description of all state estimators considered in this work, along with their first two moments.  

Now, let us discuss the \cite{pelecanos2025mixed} estimator. By Schur-Weyl duality, one can apply the unitary Schur transform to block-diagonalize the state $\rho^{\otimes t}$ into blocks indexed by partitions $\lambda \vdash t$. Applying a projective measurement on this block-diagonal state collapses it into one of the blocks $\lambda$ at random. Conditioned on outcome $\lambda$, the resulting state $\rho_{|\lambda}$ can be treated as if it has rank at most $\ell(\lambda)$, where $\ell(\lambda)$ is the length of the resulting partition. The \cite{pelecanos2025mixed} estimator then applies a rank-$\ell(\lambda)$ random purification channel to $\rho_{|\lambda}$. At this stage, one can essentially treat the resulting state as $t$-copies of 
a rank-$\ell(\lambda)$ purification of $\rho_{|\lambda}$. Then, the estimator applies a standard pure-state tomography algorithm to this purified state. Finally, the \cite{pelecanos2025mixed} estimator takes the state obtained from this tomography algorithm and traces out the $\ell(\lambda)$-dimensional purification register to obtain the final output.

\paragraph{Refining the second moment:} To analyze the second moment of the above estimator, Pelecanos, Spilecki, Tang, and Wright first compute the second moment conditioned on an intermediate random outcome $\lambda$ and then average over $\lambda$.  They already truncate the lower-order terms in the former conditional second moment expression, which leads to the final truncated form in \Cref{eq:overview-second-moment}. As we will show in \Cref{sec:PTSW-moments}, their proof techniques can be extended without much additional effort to \emph{exactly} compute the conditional second moment without  this truncation. In particular, letting $\hat{\rho}_\lambda$ be the conditional outcome of the estimator, we will provide an exact expression for $\mathbb{E}[\hat{\rho}_\lambda^{\otimes 2} | \lambda]$ in \Cref{lem:quasi-purification-conditioned-second-moment}. However, looking at this expression, it is entirely unclear how one can average over $\lambda$ to get an exact second-moment bound.

Fortunately, this does not pose a major obstacle for our variance bounds. Recall that we aimed to refine our bound on $\tr(\mbb{E}[\hat{\rho}^{\otimes 2}] \cdot \rho^{\otimes 2})$ to obtain a tighter variance bound. We will instead compute the \emph{conditional inner product}, and then average over $\lambda$ to bound this quantity. In particular, we write
\begin{equation}
    \tr(\mbb{E}[\hat{\rho}^{\otimes 2}] \cdot \rho^{\otimes 2}) = \mbb{E}_\lambda[\tr(\mbb{E}[\hat{\rho}_\lambda^{\otimes 2} | \lambda] \cdot \rho^{\otimes 2})];
\end{equation}
we will thus exactly compute the inner term as a function of $\lambda$, and then average over $\lambda$. This average ends up being much easier to calculate, and allows us to remove the undesirable term mentioned previously; see the proof of \Cref{lem:quasi-purification-variance-case-2} for the detailed calculation.

The above refinement ends up being enough to get a strong upper bound on the variance of our purity estimator $\Bar{X}$, allowing us to finally obtain \Cref{thm:intro-mixedness-upper}, our mixedness testing upper bound. Further, we can immediately use this variance bound to prove \Cref{thm:intro-purity-upper}, our upper bound for multiplicative-error purity estimation with $t$-copy measurements. To conclude, we emphasize again that the reason we used the estimator of \cite{pelecanos2025mixed} was due to the ease of understanding (and refining) its second moment, not because of its unbiasedness or any other advantages offered by the use of random purification. 

\subsubsection{Upper bounds for state certification} \label{sec:upper-bound-certification-overview}

We now extend techniques from the previous section to obtain our algorithm for state certification.

\paragraph{Hilbert--Schmidt squared estimation:} We will again start with the case of testing for the Hilbert--Schmidt distance. We want to distinguish between the cases $\rho = \sigma$ and $\|\rho - \sigma\|_2 \geq \epsilon$. Letting $\Delta = \rho - \sigma$, we can write $\|\rho - \sigma\|_2^2 = \tr(\Delta^2)$. Thus, if we could construct an estimator for $\tr(\Delta^2)$, we could simply compare our estimate to an appropriate threshold to distinguish between the two cases. A naive idea for this would be to use state estimates $\{\hat{\rho}_i\}$, and compute quantities of the form $\tr((\hat{\rho}_i - \sigma)(\hat{\rho}_j - \sigma))$. However, setting $\sigma = \mmstate$, this precisely recovers the purity estimator used in the previous section (up to a deterministic shift), and analyzing its variance required refining prior tomography results. Handling such refined second-moment expressions for the problem of certification appears to be much harder.

Instead, we will construct a randomized estimator whose variance can be bounded using only the original second moment bound of \cite{pelecanos2025mixed}. In particular, we will obtain estimates $\hat{\rho}$, $\hat{\sigma}$ by individually applying a state estimator $\mc{A}$ to $\rho^{\otimes t}, \sigma^{\otimes t}$. We then use these to get an estimate $\hat{\Delta} = \hat{\rho} - \hat{\sigma}$. Clearly, $
    \mbb{E}[\hat{\Delta}] = \mbb{E}[\hat{\rho}] - \mbb{E}[\hat{\sigma}] = \rho - \sigma = \Delta.$
Moreover, we again only require the state estimator to have a fixed bias; by subtracting the two estimates, the bias term gets cancelled out, and we always have $\mbb{E}[\hat{\Delta}] = c \cdot \Delta$, for some scalar $c$. Now, as in the previous section, we can use $n$ estimators $\{\hat{\Delta}_i\}_{i = 1}^n$ to estimate $\tr(\Delta^2)$:
\begin{equation}
    \Bar{X} = \frac{1}{\binom{n}{2}} \sum_{i,j \in [n], i < j} \tr(\hat{\Delta}_i \hat{\Delta}_j).
\end{equation}
Our construction of $\Bar{X}$ only used fixed, $t$-copy measurements, and we have $\mbb{E}[\Bar{X}] = \tr(\Delta^2)$. We can also state the variance of $\Bar{X}$ in terms of the first two moments of the underlying unbiased state estimators: 
\begin{equation}
    \mbb{V}[\bar{X}] = \bigo(n^{-1}) \cdot (\tr(\mbb{E}[\hat{\Delta}^{\otimes 2}] \cdot \Delta^{\otimes 2}) - \tr(\Delta^2)^2) + \bigo(n^{-2}) \cdot (\tr(\mbb{E}[\hat{\Delta}^{\otimes 2}]^2) - \tr(\Delta^2)^2);
    \label{eq:overview-variance-generic-closeness}
\end{equation}
see \Cref{lem:hs-estimation-variance-generic} for the general statement when $\mc{A}$ is biased. Note that in the above equation,
\begin{equation}
    \mbb{E}[\hat{\Delta}^{\otimes 2}] = \mbb{E}[(\hat{\rho} - \hat{\sigma})^{\otimes 2}] = \mbb{E}[\hat{\rho}^{\otimes 2}] + \mbb{E}[\hat{\sigma}^{\otimes 2}] - \rho \otimes \sigma - \sigma \otimes \rho.
\end{equation}
To demonstrate that biased estimators suffice, we instantiate the above estimator with the uniform POVM and recover the tight $\bigo(d^2/\eps^2)$ upper bound of \cite{yu2021sample,liu2024role} for certification with fixed single-copy measurements in \Cref{sec:closeness-testing-uniform-povm}. Now, for our $t$-copy upper bound, we again use the \cite{pelecanos2025mixed} estimator. As we show in \Cref{sec:closeness-HS-upper-proof}, the variance of our randomized estimator can already be bounded using their truncated second-moment bound (\Cref{eq:overview-second-moment}), and we do not require our refinement.

However, the variance obtained here does pose another problem: it includes a term that scales with $\tr(\rho \sigma)$, i.e., it depends on the unknown state. In case this inner product is $\approx \frac1d$, the corresponding term is small enough to be neglected. Indeed, this is the case when $\sigma = \mmstate$, yielding an alternate proof of our \Cref{thm:intro-mixedness-upper}. However, in the worst case, the inner product can be $\Omega(1)$, resulting in an $\bigo\bigl(\max\bigr\{\frac{d^2}{\sqrt{t}\eps^2}, \frac{d^{3/2}}{\eps^2}\bigr\}\bigr)$ upper bound for $t$-copy certification and preventing us from ever reaching the optimal $\bigo\bigl(\frac{d}{\eps^2}\bigr)$ rate for any $t$. Thus, we must ensure that we only apply this tester when $\tr(\rho\sigma) \approx \frac1d$. 

\paragraph{Bucketing for instance-dependent state certification:} Note that if we could ensure that one of the states used in the above test is nearly maximally mixed, i.e., has operator norm at most $\frac{c}{d}$, for some constant $c > 0$, then we would have $\tr(\rho\sigma) \leq \frac{c}{d}$. We achieve this by using the bucketing-based reduction of \cite{chen2022toward,o2025instance} for instance-optimal state certification. In particular, they show that trace-norm state certification of any hypothesis state reduces to performing a series of Hilbert--Schmidt tests on low-dimensional states with small operator norm.

More concretely, given $\sigma$'s description, we can block-diagonalize it as $\sigma = \bigoplus_{j} \sigma_j$, where $j$ indexes over each block. Here, we choose the blocks such that for all $j$, $\lambda_{\min}(\sigma_j)$ and $\lambda_{\max}(\sigma_j)$ are within constant factors. Given such a decomposition, one can also write the associated diagonal blocks of the unknown state $\rho$ as $\{\rho_j\}_j$. Note that $\rho$ itself need not be block-diagonal in the same basis, but it is possible to test for such cases too. Now, as we summarize in \Cref{lem:bucketing-reduction}, it essentially suffices to perform tests of the form $\rho_j = \sigma_j$ or $\|\rho_j - \sigma_j\|_2 \geq \eps_j$, for some appropriate $\eps_j$. To perform such a test, we can project the states $\rho, \sigma$ into the $j$th block to get copies of states $\hat{\rho}_j, \hat{\sigma}_j$, and then apply our Hilbert--Schmidt testing algorithm to these copies.  Let $d_j$ be the dimension of the $j$th block. As we required $\sigma_j$'s eigenvalues to be within constant factors of each other, we will necessarily have $\|\hat{\sigma}_j\|_{\infty} \leq \frac{c}{d_j}$. Thus, we will have $\tr(\hat{\rho}_j \hat{\sigma}_j) \leq \frac{c}{d_j}$. Consequently, the undesirable `$\tr(\rho\sigma)$ term' can be neglected in the complexity of each Hilbert--Schmidt sub-test. Summing over all sub-tests, we get an instance-dependent upper bound for trace-norm certification in \Cref{thm:instance-dependent-tcopy-certification}, without any contribution from the undesirable terms, as desired.

The instance-dependent complexity we obtain from the above methods depends on Schatten quasinorms of the form $\|\sigma\|_p$, for $0 < p < 1$. It is not hard to see that all such quasinorms are maximized by the maximally mixed state. In such cases, we have $\|\mmstate\|_p = d^{\frac1p-1}$. Plugging this back into \Cref{thm:instance-dependent-tcopy-certification}, we get our upper bound stated in \Cref{thm:intro-state-certification-main}. The polylog factors in our worst-case upper bound arise from our use of bucketing. When we group eigenvalues by orders of magnitude, we can end up with $m = \bigo(\log(d/\eps))$ buckets, and the precise upper bound we obtain includes $\mathrm{poly}(m)$ factors. To remove these polylog factors from our main upper bound, one would have to get around the undesirable variance term discussed previously without using bucketing.

\subsubsection{Lower bounds} \label{sec:lower-bound-overview}
We will now provide a brief overview of our lower bound proofs and refer to \Cref{sec:lower-bound-prelims} for the detailed approach. Recall that our main lower bound is against testers that can perform private $t$-copy measurements. We first recall arguments from prior work for \emph{single}-copy state certification in such settings.

\paragraph{Prior lower bounds for single-copy measurements:}

In the single-copy setting, one can assume that the tester sequentially applies POVMs $\mc{M}_1, \dots, \mc{M}_n$ to individual copies of $\rho$. In the general non-adaptive setting, we assume that all measurements are selected before the copies of the unknown state are received and that these measurements can be drawn using shared randomness. For instance, the tight non-adaptive algorithm of \cite{chen2022toward} samples a random basis and then measures each copy in this same basis, achieving an $\bigo(d^{3/2}/\eps^2)$ upper bound for state certification. In the more restricted setting of measurements drawn with private randomness, one loses out on the ability to draw measurements from a common distribution, and $\Omega(d^2/\eps^2)$ copies end up being necessary in this setting \cite{liu2024role}. 
Note that there is no such distinction between private and public randomness in the setting of \emph{fully entangled} measurements, as here, one only performs a single measurement on all copies at once.

To prove lower bounds for state certification, one typically considers the hardness of distinguishing between the maximally mixed state $\mmstate$ and a random state $\rho \sim \mc{D}$, where $\mc{D}$ is a mixture of states that are $\eps$-far from $\mmstate$. Clearly, if an algorithm could certify the maximally mixed state, it could also solve the above point-versus-mixture task. Then, to show the hardness of this latter task, one typically shows that the distributions over measurement outcomes in the two cases are statistically indistinguishable. In particular, let $p_\rho$ be the distribution over outcomes from applying the POVMs $\mc{M}_1, \dots, \mc{M}_n$ to individual copies of $\rho$. Then, one aims to show that $p_{\mmstate}$ and $\mbb{E}_{\rho \sim \mc{D}}[p_{\rho}]$ are close in statistical distance.

In the settings of non-adaptive and private measurements, it is not hard to see that the outcome distributions $p_\rho$ can be written as product distributions. In particular, we can write $p_\rho = p_\rho^{(1)} \otimes \dots \otimes p_\rho^{(n)}$, where $p_\rho^{(i)}$ is the distribution over outcomes when $\mc{M}_i$ is applied to $\rho$. Now, to show that a product distribution and a mixture of product distributions are statistically close, one can conveniently upper bound the $\chi^2$-divergence between the two as a function of $n$ using the Ingster--Suslina method (see \Cref{lem:ingster-suslina}). A lower bound for the copy complexity can then be obtained by noting that $n$ must be large enough for the $\chi^2$-divergence to exceed some absolute constant. Such arguments are already sufficient to prove the optimal $\Omega(d^{3/2}/\eps^2)$ non-adaptive lower bound \cite{bubeck2020entanglement}. 

However, in the setting of private randomness, the stronger $\Omega(d^2/\eps^2)$ lower bound of \cite{liu2024role} makes use of an additional observation: this setting is equivalent to one where the mixture $\mc{D}$ is chosen adversarially \emph{after} the measurements $\mc{M}_1, \dots, \mc{M}_n$ have been chosen. This adversarial freedom allows one to show much tighter bounds on the $\chi^2$-divergence, enabling stronger lower bounds in this setting.

Let us now recall the mixture of alternatives considered by \cite{liu2024role}, as this will also be the hard instance used in our work. Let $\ell \approx d^2$ be an integer. Given some orthonormal basis $V_1, \dots, V_{d^2}$ of $\mbb{C}^{d \times d}$, and uniformly random $\bfz \sim \{-1,+1\}^\ell$, the induced mixture of alternatives consists of states of the form
\begin{equation}
    \rho_{\bfz} = \mmstate + \Delta_{\bfz}, \quad \text{ where} \quad \Delta_{\bfz} \propto \frac{\eps}{\sqrt{d\ell}} \sum_{i = 1}^\ell V_i \bfz_i. \label{eq:overview-hard-instance}
\end{equation}

Then, using the Ingster--Suslina method, one can upper bound the $\chi^2$-divergence between outcome distributions as a function of the chosen POVMs and the perturbation basis $V_1, \dots, V_{d^2}$. In the non-adaptive case, one can only prove an upper bound on this quantity for a worst-case choice of $V_1, \dots, V_{d^2}$. However, in the case of private measurements, Liu and Acharya show an adversarial choice of this basis that allows for a much tighter bound on the $\chi^2$-divergence, yielding their $\Omega(d^2/\eps^2)$ lower bound.

\paragraph{Extending to $t$-copy measurements:}

In our setting of $t$-copy measurements, we instead imagine that the tester performs $n$ measurements $\mc{M}_1, \dots, \mc{M}_n$ on $t$ copies of $\rho$ at a time, and then aim to show that the outcome distributions are statistically indistinguishable for small $n$. When these measurements are chosen non-adaptively or with private randomness, the resulting outcome distributions can still be written as product distributions. Thus, we will also pass to the $\chi^2$-divergence and upper bound it using the Ingster--Suslina method. Using the \cite{liu2024role} mixture of alternatives, let us first write the likelihood-ratio deviation for some measurement outcome $x$ associated with a POVM operator $M_x$:
\begin{equation}
    \delta_x(\bfz) = \frac{\tr(M_x \rho_{\bfz}^{\otimes t})}{\tr(M_x \mmstate^{\otimes t})} - 1 = d^t \frac{\langle M_x, (\mmstate + \Delta_{\bfz})^{\otimes t} - \mmstate^{\otimes t} \rangle}{\tr(M_x)}.
\end{equation}
Given the form of $\Delta_{\bfz}$ in \Cref{eq:overview-hard-instance}, it is not hard to see that the likelihood-ratio deviation above is a symmetric degree-$t$ polynomial in the entries of the random string $\bfz$. Then, using the Ingster--Suslina method, one can show that upper-bounding the $\chi^2$-divergence boils down to bounding the moment generating function of a structured degree-$2t$ polynomial in Rademacher random variables $\bfz_1, \dots, \bfz_\ell, \bfz_1^\prime, \dots, \bfz_\ell^\prime$. Unfortunately, bounding the mgf of generic polynomials of Rademacher random variables appears intractable, and it is unclear how one can exploit the structure of our polynomials to simplify this task.

\paragraph{High-precision lower bounds via linearization:} To simplify the polynomials appearing above, we will \emph{linearize} the $t$-fold tensor product of our hard instance, i.e., write its first order expansion around the maximally mixed state. In particular, we will write
\begin{equation}
    \rho_{\bfz}^{\otimes t} = (\mmstate + \Delta_{\bfz})^{\otimes t} \approx \mmstate^{\otimes t} + \sum_{i \in [t]} \Delta_{\bfz}^{(i)} \otimes \mmstate^{\otimes [t] \setminus \{i\}}.
\end{equation}
The linear approximation used here is accurate in the high-precision regime that we focus on in this work. This idea has appeared previously in the high-precision lower bound of \cite{chen2024optimalTradeoffsTomography} for $t$-copy tomography.

Now, when the above approximation is sufficiently accurate, the likelihood ratio deviation can be written as a \emph{linear} polynomial in the entries of the string $\bfz$. Consequently, upper-bounding the $\chi^2$-divergence reduces to bounding the mgf of a \emph{quadratic} polynomial in some Rademacher random variables. MGF bounds for such quadratic polynomials are well-understood, allowing us to upper bound the $\chi^2$-divergence as a function of the POVMs $\mc{M}_1, \dots, \mc{M}_n$ and the perturbation basis $V_1, \dots, V_{d^2}$. 

In the general non-adaptive setting, we prove a worst-case upper bound on the above function, allowing us to obtain \Cref{thm:intro-lower-shared}. Moreover, in the private randomness setting, we can adversarially select the perturbation basis to minimize the $\chi^2$-divergence, giving us our main lower bound in \Cref{thm:intro-state-certification-main,thm:intro-mixedness-upper}. We emphasize that the adversarial perturbation basis in this setting is distinct from the one in \cite{liu2024role}, and we require some additional work to bound the quantities arising in our setting.

\subsection{Discussion} \label{sec:discussion}

In the high-precision regime, we have shown that the copy complexity of state certification with fixed $t$-copy measurements 
is $\Tilde{\Theta}\bigl(\max\bigl\{
\frac{d^2}{\sqrt{t}\eps^2},
\frac{d}{\eps^2}\bigr\}\bigr)$. Our result shows that near-optimal rates for certification are attainable at $t = \Theta(d^2)$, which is strictly better than performing fully entangled measurements whenever $\eps = o(1/\sqrt{d})$. Our \Cref{thm:intro-lower-shared} shows that this entanglement cost is also necessary in the non-adaptive setting. Despite these results, our work raises several important open questions and future directions.

\paragraph{A complete characterization of $t$-copy certification?} While we have provided a near-optimal characterization in the high-precision regime, only the folklore $\bigo(d^2/t\epsilon^4)$ upper bound is known for testing with lower precision. It is unclear whether this complexity is optimal in the latter regime, and how one can prove lower bounds here.
We leave the remaining characterization of $t$-copy certification and identification of the precise transition point between the two phases to future work.

\paragraph{Improved tradeoffs with shared randomness?} With fixed and private measurements, one requires $t \geq d$ to even reach the $\bigo(d^{3/2}/\eps^2)$ non-adaptive single-copy rate. Can one design a $t$-copy tester using shared randomness that starts at $\bigo(d^{3/2}/\eps^2)$ for $t = 1$? It is unclear how to exploit shared randomness in the $t$-copy setting; all generalizations of the single-copy tester we attempted led to undesirable rates.

\paragraph{Beyond linearization for $t$-copy lower bounds?} Our lower bound proof splits up the $\chi^2$-divergence into the contributions due to linear perturbations and higher-order terms separately. We are then able to show that the higher-order terms are uniformly bounded for small $\eps$, leading to lower bounds in the high-precision regime. Similarly, \cite{chen2024optimalTradeoffsTomography} also prove high-precision lower bounds for tomography via linearization. The $t$-copy shadow tomography lower bounds of \cite[Theorems 5 and 6]{chen2024optimalTradeoffsShadowTomography} have a tradeoff between contributions arising from linear and non-linear terms in the $\chi^2$-divergence; this also leads to a main term that only dominates the lower bound in the high-precision setting. Is this recurrence of high-precision lower bounds in the $t$-copy setting an artifact of the linearization-based approach, or is there an intrinsic property of this setting that prevents us from ever proving these bounds across all regimes of $\eps$? It is likely that proving such lower bounds across all $\eps$ would require developing new lower bound techniques. The techniques of \cite{chen2021hierarchy} do not seem suited for this task, as their lower bounds do not hold for the full range of $t$. We also note that the symmetrized swap observable mentioned in \Cref{sec:upper-bound-mixedness-overview} cannot distinguish between $\mmstate^{\otimes t}$ and a linearized approximation of $\rho^{\otimes t}$. In other words, the optimal algorithm for mixedness testing no longer works when linear approximations hold well; this phenomenon may explain the phase transition found here.

\paragraph{From state certification to closeness testing?}
Prior worst-case upper bounds for certification in the fully entangled setting \cite{buadescu2019quantum} and the single-copy measurement settings \cite{chen2022toward,liu2024role} also hold for the problem of \emph{closeness testing}, i.e., when we are only given access to copies of $\sigma$ instead of its description. As our upper bound uses the bucketing-based reduction for instance-optimal state certification from \cite{chen2022toward,o2025instance}, and then specializes to the case of the maximally mixed state, our algorithm does not succeed at this harder problem. Concretely, the main roadblock is that bucketing itself requires the description of $\sigma$. If one could obtain the necessary information to carry out bucketing from copies of $\sigma$ alone, this could indeed be overcome, lifting our upper bounds to closeness testing. Going further, can one also show instance-optimal bounds for quantum closeness testing? We note that such bounds have already been shown for the classical problem of distribution closeness testing \cite{diakonikolas2016new}. 

\paragraph{Are random purification and state certification compatible?} 
Pelecanos, Spilecki, Tang, and Wright \cite{pelecanos2025mixed} use the random purification channel to map their estimator's copies $\rho^{\otimes n}$ to copies of a uniformly random purification $\ket{\bfrho}$, i.e., to $\ket{\bfrho}\bra{\bfrho}^{\otimes n}$. This allows them to map mixed-state tomography to pure-state tomography, which ends up being a (conceptually) simpler problem. Similarly, can one map mixed-state certification to a simpler problem for pure states? It is unclear if this approach can help, as under random purification, even the purification register of the hypothesis state becomes unknown.

More simply, can the above approach even work for mixedness testing? Under the random purification channel, one can transform mixedness testing of $d$-dimensional states to a testing problem of the entanglement spectrum of bipartite $d^2$-dimensional pure states. However, it is unclear whether there are any non-trivial intuitive tests for this property of the entanglement spectrum. Further, `local testing', a dual perspective to random purification, has previously led to tight bounds for such testing tasks \cite{chen2026local}. Nevertheless, we cannot rule out the possibility of an alternate approach via random purification leading to new algorithms for state certification. A useful starting point in this direction might be to recast the symmetrized swap observable of \cite{buadescu2019quantum} in the framework of the random purification channel.

\subsection{Related work} \label{sec:related-work}

\paragraph{Quantum state certification:}
Unless specified otherwise, all copy complexities stated here are for certification with respect to the trace norm. As mentioned previously, the worst-case complexity of quantum state certification with fully entangled measurements is $\Theta(d/\eps^2)$ \cite{o2015quantum,buadescu2019quantum}. With single-copy measurements, this complexity was shown to be $\Theta(d^{3/2}/\eps^2)$\cite{bubeck2020entanglement,chen2022toward,chen2022tightStateCertification}. Further, for certification with single-copy measurements, \cite{chen2022toward,chen2022tightStateCertification} showed \emph{instance-optimal} bounds, i.e., bounds with an optimal dependence on the state $\sigma$ to be certified. Their bounds essentially show that the complexity of certification depends on the fidelity between $\mmstate$ and the hypothesis state $\sigma$. Such instance-optimal bounds were also shown recently in the setting of fully entangled measurements \cite{o2025instance}.

Returning to the single-copy setting, when measurements are restricted to be drawn using private randomness, the complexity of certification is $\Theta(d^2/\eps^2)$ \cite{yu2021sample,liu2024role}. A recent single-copy tomography algorithm \cite{grewal2026pauli} develops a new Hilbert--Schmidt distance estimation protocol as a subroutine. This algorithm uses $\bigo(d/\eps^2)$ copies to estimate the Hilbert--Schmidt distance between two states to accuracy $\bigo(\eps)$, implying an $\bigo(d^2/\eps^2)$-copy algorithm for trace-norm state certification. However, their algorithm also employs \emph{shared randomness} across measurements. While this is not an optimal algorithm for general non-adaptive measurements, the algorithm only performs Pauli measurements, and it could be the case that this complexity is optimal for certification with such restricted measurements.

State certification has also been studied in various other resource-constrained settings. The most relevant of these to our work are the lower bounds of Chen, Cotler, Huang, and Li~\cite{chen2021hierarchy} for mixedness testing with $t$-copy measurements. In particular, \cite{chen2021hierarchy} showed that for $t \leq \bigo(\min\{d^{2/3},\sqrt{d}/\eps, 1/\eps^2\})$, $\Omega\left(\frac{d^{4/3}}{t\eps^2}\right)$ copies are necessary using adaptively chosen $t$-copy measurements, and $\Omega\left(\frac{d^{3/2}}{t\eps^2}\right)$ copies are necessary with non-adaptive $t$-copy measurements\footnote{Looking at the proofs of \cite[Theorems 3 and 15]{chen2021hierarchy}, they proved that for $t \leq \bigo(\min\{d^{2/3},\sqrt{d}/\eps\})$, the copy complexities are at least $\Omega\left(\frac{d^{4/3}}{t\eps^2(1+\eps^2)^{(t-1)}}\right)$ and $\Omega\left(\frac{d^{3/2}}{t\eps^2(1+\eps^2)^{(t-1)}}\right)$ respectively. This yields the bounds we have stated here. However, these results were inaccurately reported in their work, with the linear $t$ dependence omitted in the denominator. }. Other restrictions include single-copy measurements with few outcomes \cite{liu2024quantum}, adversarial access \cite{aliakbarpour2025adversarially}, limited communication in a distributed setting \cite{doosti2026distributed}, and non-iid copies \cite{de2025non}. The complexity of certifying \emph{pure} states using only \emph{single-qubit} \cite{huang2024certifying,gupta2025few} and \emph{few-qubit} measurements \cite{coladangelo2026power} has also been studied recently.

\paragraph{Quantum state tomography:}
For the problem of tomography with respect to the trace norm, the optimal copy complexity is $\bigo(d^2/\eps^2)$ with general measurements\cite{o2016efficient,haah2016sample}. With unentangled measurements, the complexity of this problem becomes $\Theta(d^3/\eps^2)$ \cite{kueng2017low,haah2016sample,chen2023does}. We note that the optimal rate for tomography with single-copy measurements can be achieved by repeatedly applying the same measurement, implying that, unlike the case of certification, there is no separation between non-adaptive and deterministic measurements for tomography. For $t$-copy measurements, the first upper and lower bounds were shown by Chen, Li, and Liu~\cite{chen2024optimalTradeoffsTomography}. Recently, Pelecanos, Spilecki, Tang, and Wright \cite{pelecanos2025debiased,pelecanos2025mixed} developed unbiased estimators for quantum states that achieve the optimal rate in the fully entangled setting. Further, by applying their estimator to batches of $t$ copies, they improve on the upper bounds of \cite{chen2024optimalTradeoffsTomography} for $t$-copy tomography, showing that $\bigo\left(\max\bigl\{\frac{d^3}{\sqrt{t}\eps^2}, \frac{d^2}{\eps^2}\bigr\}\right)$ copies are sufficient. As noted previously, these algorithms achieve the optimal rate for tomography at $t = d^2$. We also note that Chen, Li, and Liu \cite{chen2024optimalTradeoffsTomography} showed lower bounds matching this complexity in the low-entanglement regime, i.e., for $t \leq 1/\eps^{0.1}$. The complexity of tomography with other distance measures is also well-understood, e.g., fidelity \cite{haah2016sample, o2017efficient, chen2023does} and quantum $\chi^2$-divergence \cite{flammia2024quantum}. Tomography has also recently been studied under other restrictions, including with Pauli measurements \cite{acharya2025pauliSingleQubit,acharya2025pauliSingleCopy,grewal2026pauli} and in adversarial settings \cite{aliakbarpour2025adversarially}.

\paragraph{Random purification:}
As discussed previously, the estimator of \cite{pelecanos2025mixed} uses a recently introduced \emph{random purification channel} \cite{tang2025conjugate}. This channel allows one to map $\rho^{\otimes n}$ to $\mbb{E}_{\ket{\bfrho}}\ket{\bfrho}\bra{\bfrho}^{\otimes n}$, where the expectation is over all random purifications $\ket{\bfrho}$ of $\rho$. A simpler implementation of this channel was shown in \cite{girardi2025randomPurificationSimplified}. There has been a flurry of recent work exploring such operations, such as for fermionic and bosonic states \cite{mele2025randomPurificationGaussianBosons,walter2025randomPurificationGeneric} and quantum channels \cite{yoshida2025randomDilation,girardi2025randomStinespring}, as well as applications of such operations for other problems, like channel tomography \cite{mele2025optimal,girardi2025randomStinespring}. The dual local-testing framework has also been used to obtain improved rates for channel tomography \cite{chen2026quantum}.

\paragraph{Other quantum inference problems:}
There is a rich literature studying the effects of restricted resources on the complexity of various quantum learning and testing problems, including for shadow tomography \cite{aaronson2018shadow,huang2021information,chen2022exponential,chen2024optimalTradeoffsShadowTomography,king2025triply,chen2024optimalHighPrecisionShadow}, spectrum estimation \cite{o2015quantum,pelecanos2025beating}, inner product estimation \cite{anshu2022distributed,gong2024sample,arunachalam2025generalized}, and moment estimation~\cite{acharya2020estimating,pelecanos2025beating}. In particular, we note that Chen, Gong, and Ye \cite[Theorem 5]{chen2024optimalTradeoffsShadowTomography} also show lower bounds for shadow tomography with $t$-copy measurements in a high-precision regime via an implicit linearization argument. Additionally, they also consider the effect of some qubits of persistent memory across measurements. Recent works \cite{ye2025exponential,noller2025infinite} have also shown exponential separations between algorithms restricted to $t$-copy and $t+1$-copy measurements. We refer to the excellent surveys \cite{montanaro2016survey,anshu2024survey} for a more thorough overview of quantum state learning and testing.

\subsection{Organization}

We start by presenting technical preliminaries and notation in \Cref{sec:prelim}. We provide a technical exposition of relevant state estimators in \Cref{sec:state-estimation}. We prove our upper bounds for state certification in \Cref{sec:main-upper-bound}, and also obtain the mixedness testing upper bound (\Cref{thm:intro-mixedness-upper}) along the way. We provide an alternate proof of \Cref{thm:intro-mixedness-upper} and also prove our purity estimation upper bound in \Cref{sec:mixedness-and-purity-upper}.  Finally, we prove our lower bounds in \Cref{sec:lower}. Additionally, we describe and prove the folklore upper bound in \Cref{sec:batching-BOW-tester}.

\subsection*{Acknowledgments}
CW was supported by the UK EPSRC through the Quantum Advantage Pathfinder project with grant reference EP/X026167/1 and thanks Angelos Pelecanos and Jack Spilecki for insightful discussions on state tomography. SC was supported by NSF CCF-2430375 and acknowledges helpful past discussions with Jerry Li, Ryan O'Donnell, and John Wright on mixedness testing with intermediate entanglement.

\subsection*{Use of LLMs}

The first arXiv version of this work was prepared entirely by humans but had weaker lower bounds. In the second (current) version, a frontier model generated proofs of the improved intermediate \Cref{lem:reduced-channel-operator,lem:reduced-channel-trace}, which led to our new lower bounds. The proofs of these lemmas have been digested and rewritten by the authors, who claim full responsibility for the updated results.

\section{Preliminaries}
\label{sec:prelim}

A $d$-dimensional quantum state is represented by a density operator in $\mbb{C}^{d \times d}$, i.e., a positive semidefinite operator with unit trace. The maximally mixed state of dimension $d$ is given by $\mmstate = I/d$, where $I \in \mathbb{C}^{d \times d}$ is the identity operator. A pure state is a quantum state with rank $1$, and can instead be represented by its sole eigenvector with eigenvalue $1$. For a bipartite quantum state $\rho$ over registers $\mathsf{A},\mathsf{B}$, we will denote its reduced state on each register by $(\rho)_{\mathsf{A}}, (\rho)_{\mathsf{B}}$ respectively, and also extend this notation to more general multipartite systems. In such multipartite systems, suppose we have $k$ registers numbered $1, \dots, k$. For an operator $M \in \mbb{C}^{d \times d}$, we will denote $M_i$ to be the operator $M$ on the $i$th register. More generally, for a subset of registers $S \subseteq \{1, \dots, k\}$, we denote $M^{\otimes S} = \bigotimes_{i \in S} M_i$. Further, for any operator $N \in (\mbb{C}^{d \times d})^{\otimes k},$ we define $N_S$ to be $N$ across the registers in $S$.
Finally, for $1 \leq t \leq k$, and any operators $M,N \in \mbb{C}^{d \times d}$, we define the symmetric sum:
\begin{equation}
    \symsum M^{\otimes t} \otimes N^{\otimes k - t} = \sum_{S \subseteq [k], |S| = t} M^{\otimes S} \otimes N^{\otimes [k] \setminus S}.
\end{equation}

One can act on quantum states by performing positive-operator valued measurements (POVMs). A POVM consists of a set of measurement operators $\{M_x\}_x$, where each operator belongs to $\mbb{C}^{d \times d}$ and is positive semidefinite, and $\sum_x M_x = I$. Applying this POVM to a state $\rho \in \mbb{C}^{d \times d}$ yields outcome $x$ with probability $\tr(M_x \rho)$. 
More generally, we may consider applying POVMs to multiple copies of a state at once. Then for, say, $t$ copies, the measurement operators belong to $\mbb{C}^{d^t \times d^t}$, and are applied to the joint state $\rho^{\otimes t}$. One can also map quantum states to other quantum states by applying quantum channels, which are completely positive and trace-preserving (CPTP) linear maps.

It will also be helpful to represent quantum states and channels via vectorization. More generally, for any operator $M \in \mbb{C}^{d_A \times d_B}$ with $M = \sum_{i,j} M_{i,j} \ket{i}\bra{j}$, we define its vectorization, $\mathrm{vec}(M) \triangleq \sum_{i,j} M_{i,j} \ket{i} \otimes \ket{j}$. The action of quantum channels in the vectorized picture can be understood using the Liouville matrix representation:

\begin{definition}
    Let $\Phi : \mbb{C}^{d_A \times d_A} \mapsto \mbb{C}^{d_B \times d_B}$ be a quantum channel. Then, the associated \emph{Liouville matrix} $S_\Phi \in \mbb{C}^{d_B^2 \times d_A^2}$ is the unique matrix satisfying 
    \begin{equation}
        \mathrm{vec}(\Phi(X)) = S_\Phi \mathrm{vec}(X),
    \end{equation}
    for all $X \in \mbb{C}^{d_A \times d_A}$.
\end{definition}

Let us now talk about the kind of problems we will consider in this work. We will primarily be interested in three testing problems for quantum states. Let us first define the problem of quantum state certification.
\begin{definition}[Quantum State Certification]
\label{def:state-certification}
    Let $\rho,\sigma$ be two $d$-dimensional states, where one receives the complete description of $\sigma$ and multiple copies of the unknown state $\rho$. The task of \emph{$\eps$-trace norm certification} of $\sigma$ refers to the problem of determining whether $\rho = \sigma$ or $\|\rho - \sigma\|_1 \geq \eps$ with probability at least $.99$. 
\end{definition}
We also similarly define \emph{$\eps$-Hilbert--Schmidt certification}, where we instead test whether $\rho = \sigma$ or $\|\rho - \sigma\|_2 \geq \eps$. A special case of state certification with key relevance is \emph{mixedness testing}, where $\sigma = \mmstate$. Lastly, we will also extend the above definition to the more general problem of \emph{closeness testing}; here, the tester is given copies of two unknown states $\rho,\sigma$, and tasked with determining whether they are equal or far. While we don't focus on closeness testing in this work, some of our protocols for state certification are based on solving this more challenging task.

As mentioned above, we will consider testing quantum states with respect to the Hilbert--Schmidt norm (i.e., the Schatten $2$-norm), and the trace norm (i.e., the Schatten $1$-norm). In general, for any $p > 0$, we will use $\|\cdot\|_p$ to denote the Schatten $p$-norm. For $p < 1$, these are actually quasinorms, but will still be relevant in our analysis. We will primarily aim to prove results for testing with respect to the trace norm. However, it is often easier to prove upper bounds for the Hilbert--Schmidt norm and then extend them to the trace norm. The following elementary fact explains how such a conversion can be carried out in general:

\begin{fact}
\label{fact:hs-to-trace-norm-testing}
    Let $d \in \mbb{N}, d > 1, 0 < \eps < 1$. Let $\rho,\sigma$ be $d$-dimensional quantum states. Suppose that the copy complexity of testing whether $\rho = \sigma$ or $\|\rho - \sigma\|_2 \geq \eps$ is at most $f(d,\eps)$. Then, the complexity of testing whether $\rho = \sigma$ or $\|\rho - \sigma\|_1 \geq \eps$ is at most $f\bigl(d,\frac{\eps}{\sqrt{d}}\bigr)$.
\end{fact}

\subsection{Symmetric subspace and random states}

In this section, we recall some standard results about the symmetric subspace, which will help us compute moments of random states in \Cref{sec:state-estimation}. Let the symmetric group over $n$ elements be denoted by $\mc{S}_n$. For any permutation $\pi \in \mc{S}_n$, its representation on $(\mbb{C}^{d})^{\otimes n}$ is denoted by the unitary $V_d(\pi)$ satisfying 
\begin{equation}
    V_d(\pi) \ket{i_1, \dots, i_n} = \ket{i_{\pi^{-1}(1)}, \dots, i_{\pi^{-1}(n)}}.    
\end{equation}
When clear from context, we may use $\pi$ to denote $V_d(\pi)$. We will use $e \in \mc{S}_n$ to denote the identity permutation, and $\swap_{i,j} \in \mc{S}_n$ to denote the permutation that swaps the $i$th and $j$th elements, where $i,j \in [n]$ and $i \neq j$. Now, we will define the symmetric subspace and discuss some of its useful properties. We refer to \cite{harrow2013church} for more details.

\begin{definition}[Symmetric subspace and symmetric projector]
    The symmetric subspace of $(\mbb{C}^{d})^{\otimes n}$, denoted by $\vee^n \mbb{C}^d$, is defined as the set of $\mc{S}_n$-invariant pure states. 
    \begin{equation}
        \vee^n\mbb{C}^d = \{\ket{\psi} \in (\mbb{C}^{d})^{\otimes n} : V_d(\pi) \ket{\psi} = \ket{\psi} \forall \pi \in \mc{S}_n\}.
    \end{equation}
   The dimension of $\vee^n\mbb{C}^{d}$ is $d[n] = \binom{d+n-1}{n}$. The symmetric projector, denoted by $\Pi_{\mrm{sym}}^n$, is the orthogonal projector onto $\vee^n\mbb{C}^d$, and satisfies
    \begin{equation}
        \Pi_{\mrm{sym}}^n = \frac{1}{n!} \sum_{\pi \in \mc{S}_n} V_d(\pi).
    \end{equation}
\end{definition}

Next, we state a standard lemma relating moments of random pure states to the symmetric projector.

\begin{lemma}[$k$-th moment of Haar-random pure states]
\label{lem:kth-moment-pure-states}
    Let $k$ be a positive integer. Then,
    \begin{equation}
        \mbb{E}_{\ket{\psi}} [\ket{\psi}\bra{\psi}^{\otimes k}] = \frac{1}{d[k]} \Pi_{\mrm{sym}}^k,
    \end{equation}
    where the expectation is taken over uniformly random pure states in $\mbb{C}^d$.
\end{lemma}

Alternatively, it will also be convenient to use the following inductive form of the moments, from \cite{pelecanos2025mixed}:
\begin{lemma}[{\cite[Lemma 5.4]{pelecanos2025mixed}}]
    \label{lem:kth-moment-alternate}
    Let $n,k$ be positive integers. Then,
    \begin{equation}
        \mbb{E}[\ket{\psi}\bra{\psi}^{\otimes n + k}] = \frac{1}{d[n] (d+n)^{\uparrow k}} (e + X_{n+k}) \dots (e + X_{n+1}) \cdot (\Pi_{\mrm{sym}}^n \otimes I^{\otimes k}),
    \end{equation}
    where $a^{\uparrow b} = a(a+1)\dots(a+b-1)$ is the rising factorial and for $i \in [n+k]$, $X_i \triangleq \sum_{j = 1}^{i-1} \swap_{j,i}$ are the Jucys-Murphys elements of the symmetric group algebra.
\end{lemma}
\section{Estimators for quantum states}
\label{sec:state-estimation}

Most of our testing algorithms use state estimation algorithms as key building blocks. To analyze our testers, we will need closed-form expressions for the first two moments of such state estimators. In \Cref{sec:uniform-povm-moments}, we present the first two moments of the uniform POVM, which will be used in our unentangled testing upper bounds. Next, in \Cref{sec:hayashi-povm-moments}, we present the first two moments of Hayashi's POVM \cite{hayashi1998asymptotic} and its debiased version from \cite{grier2024sample}. Finally, in \Cref{sec:unbiased-state-estimation-moments}, we state the second moment of a recent unbiased estimator for mixed states 
from \cite{pelecanos2025mixed}. For our mixedness testing upper bound (\Cref{thm:intro-mixedness-upper}), we need a refined version of the second moment bound from \cite{pelecanos2025mixed} conditioned on the outcome of an intermediate step in their algorithm, which we show in \Cref{lem:quasi-purification-conditioned-second-moment}. 
\subsection{Uniform POVM}
\label{sec:uniform-povm-moments}
First, let us define the uniform POVM.

\begin{definition}
    The uniform POVM over states in $\mbb{C}^{d \times d}$ is the continuous POVM with measurement operators $\{d \ket{\psi}\bra{\psi} \mrm{d}\psi\}$, where $\mrm{d}\psi$ is the uniform measure over pure states in $\mbb{C}^d$.
\end{definition}

We start by stating the first two moments of the uniform POVM, omitting the proofs as these are standard.

\begin{lemma}
\label{lem:uniform-povm-first-moment}
    Let $\psi$ be the outcome obtained when applying the uniform POVM to some state $\rho \in \mbb{C}^{d \times d}$. Then,
    \begin{equation}
        \mbb{E}[\ket{\psi}\bra{\psi}] = \frac{\rho + I}{d+1}\,.
    \end{equation}
\end{lemma}

\begin{lemma}
\label{lem:uniform-povm-second-moment}
    Let $\psi$ be the outcome obtained when applying the uniform POVM to some state $\rho \in \mbb{C}^{d \times d}$. Then,
    \begin{equation}
        \mbb{E}[\ket{\psi}\bra{\psi}^{\otimes 2}] = \frac{1}{(d+1)(d+2)}\left(I^{\otimes 2} + \swap + (\rho \otimes I + I \otimes \rho) + (\rho \otimes I + I \otimes \rho)\cdot \swap\right).
    \end{equation}
\end{lemma}

\subsection{Entangled measurements over the symmetric subspace}
\label{sec:hayashi-povm-moments}
In this section, we present moments of Hayashi's POVM and its debiased version. Typically, Hayashi's POVM is applied to $n$ copies of some pure state $\ket{\psi}$, i.e., $\ket{\psi}\bra{\psi}^{\otimes n}$, but we will consider its general application to any mixed state in the symmetric subspace. The exposition here follows that of \cite{pelecanos2025mixed}; however, unlike \cite{pelecanos2025mixed}, we state exact second-moment expressions instead of truncating the lower-order terms. This will be important when we prove our refined bound in \Cref{sec:unbiased-state-estimation-moments}.

\begin{definition}[Hayashi's POVM \cite{hayashi1998asymptotic}]
\label{def:hayashi-povm}
    Let $d[n] = \binom{d+n-1}{n}$ be the dimension of $\vee^n \mbb{C}^d$, the symmetric subspace of $(\mbb{C}^d)^{\otimes n}$. Then, \emph{Hayashi's POVM} has elements $\{d[n] \cdot \ket{\psi}\bra{\psi}^{\otimes n} \cdot \mrm{d}\psi\}$.
\end{definition}

\begin{lemma}[Moments of Hayashi's POVM]
 \label{lem:hayashi-moments}
     Let $\psi_{\mrm{sym}}$ be a mixed state on $\vee^n\mbb{C}^d$. Consider a random outcome $\ket{\psi}$ obtained by applying Hayashi's POVM to $\psi_{\mrm{sym}}$. Then,
     \begin{equation}
         \mbb{E}[\ket{\psi}\bra{\psi}] = \frac{I}{d+n} + \frac{n}{d+n} (\psi_{\mrm{sym}})_1,
     \end{equation}
     where $(\psi_{\mrm{sym}})_1$ is the reduced density matrix of $\psi_{\mrm{sym}}$ on the first subsystem. Further,
     \begin{equation}
         \mbb{E}[\ket{\psi}\bra{\psi}^{\otimes 2}] = \frac{I^{\otimes 2} +  \swap + n((\psi_{\mrm{sym}})_1\otimes I + I \otimes (\psi_{\mrm{sym}})_1) (I^{\otimes 2} + \swap) + n(n-1)(\psi_{\mrm{sym}})_{1,2}}{(d+n)(d+n+1)}.
     \end{equation}
 \end{lemma}

 \begin{proof}
     Note that outcome $\psi$ has associated density $d[n]\tr(\ket{\psi}\bra{\psi}^{\otimes n} \psi_{\mrm{sym}}) \mrm{d}\psi$. Thus, for any $k \geq 1$,
     \begin{align}
         \mbb{E}[\ket{\psi}\bra{\psi}^{\otimes k}] &= d[n] \int \ket{\psi}\bra{\psi}^{\otimes k} \tr(\ket{\psi}\bra{\psi}^{\otimes n} \psi_{\mrm{sym}}) \mrm{d}\psi
         \\&= d[n] \int \tr_{1,\dots,n}(\ket{\psi}\bra{\psi}^{\otimes n+k} \cdot \psi_{\mrm{sym}} \otimes I^{\otimes k}) \mrm{d}\psi
         \\&= \frac{1}{(d+n)\dots(d+n+k-1)} \tr_{1,\dots,n}((e + X_{n+k}) \dots (e+X_{n+1}) \cdot \Pi_{\mrm{sym}}^n \otimes I^{\otimes k} \cdot \psi_{\mrm{sym}} \otimes I^{\otimes k})
         \\&= \frac{1}{(d+n) \dots (d+n+k-1)} \tr_{1, \dots, n}((e + X_{n+k}) \dots (e + X_{n+1}) \cdot \psi_{\mrm{sym}} \otimes I^{\otimes k}), 
     \end{align}
     where we used \Cref{lem:kth-moment-alternate} in the third line, and in the last line we used that $\psi_{\mrm{sym}}$ is a mixed state on $\vee^n\mbb{C}^d$.

     We first consider the case of $k = 1$.
     \begin{align}
         \mbb{E}[\ket{\psi}\bra{\psi}] &= \frac{1}{d+n} \tr_{1,\dots,n}((e+X_{n+1})\cdot \psi_{\mrm{sym}} \otimes I)
         \\&= \frac{1}{d+n} \tr_{1,\dots,n}(\psi_{\mrm{sym}} \otimes I) + \frac{1}{d+n}\sum_{i = 1}^n \tr_{1, \dots, n}(\swap_{i,n+1} \cdot \psi_{\mrm{sym}} \otimes I)
         \\&= \frac{I}{d+n} + \frac{1}{d+n}\sum_{i = 1}^n (\psi_{\mrm{sym}})_i
         \\&= \frac{I}{d+n} + \frac{n}{d+n} (\psi_{\mrm{sym}})_1,
     \end{align}
     where the last line used the permutation invariance of $\psi_{\mrm{sym}}$.

     Now, for $k = 2$, we have
     \begin{align}
         \mbb{E}[\ket{\psi}\bra{\psi}^{\otimes 2}] &= \frac{1}{(d+n)(d+n+1)}\tr_{1,\dots,n}((e+X_{n+2})(e+X_{n+1}) \cdot \psi_{\mrm{sym}} \otimes I^{\otimes 2})
         \\&= \frac{1}{(d+n)(d+n+1)} \tr_{1,\dots,n}((e + X_{n+1} + X_{n+2} + X_{n+2}X_{n+1}) \cdot \psi_{\mrm{sym}} \otimes I^{\otimes 2}).
         \label{eq:hayashi-second-moment-1}
     \end{align}
     The first term can be obtained by noting $\tr_{1, \dots, n}(e\cdot \psi_{\mrm{sym}} \otimes I^{\otimes 2}) = I^{\otimes 2}$. We simplify the second term next:
     \begin{align}
         \tr_{1, \dots, n}(X_{n+1} \cdot \psi_{\mrm{sym}} \otimes I^{\otimes 2}) &= \sum_{i = 1}^n \tr_{1, \dots, n}(\swap_{i,n+1} \cdot \psi_{\mrm{sym}} \otimes I^{\otimes 2})
         \\&= \sum_{i = 1}^n (\psi_{\mrm{sym}})_i \otimes I
         \\&= n\cdot(\psi_{\mrm{sym}})_1 \otimes I,
     \end{align}
     where we again used the permutation invariance of $\psi_{\mrm{sym}}$. Similarly, the third term can be computed as follows:
     \begin{align}
         \tr_{1, \dots, n}(X_{n+2} \cdot \psi_{\mrm{sym}} \otimes I^{\otimes 2}) &= \sum_{i = 1}^{n+1} \tr_{1, \dots, n}(\swap_{i,n+2} \cdot \psi_{\mrm{sym}} \otimes I^{\otimes 2})
         \\&= \tr_{1, \dots, n}(\swap_{n+1,n+2} \cdot \psi_{\mrm{sym}} \otimes I^{\otimes 2}) + \sum_{i = 1}^n  I \otimes (\psi_{\mrm{sym}})_i 
         \\&= \swap + n\cdot I \otimes (\psi_{\mrm{sym}})_1.
     \end{align}
     Finally, from the proof of Lemma 5.8 in \cite{pelecanos2025mixed}, we obtain the fourth term:
     \begin{equation}
         \tr_{1, \dots, n}(X_{n+2}X_{n+1} \cdot \psi_{\mrm{sym}} \otimes I^{\otimes 2}) = n(n-1) \cdot (\psi_{\mrm{sym}})_{1,2} + n\cdot((\psi_{\mrm{sym}})_{1} \otimes I + I \otimes (\psi_{\mrm{sym}})_{1}) \cdot \swap.
     \end{equation}
     The proof of the lemma is completed by substituting the four terms back into \Cref{eq:hayashi-second-moment-1}.
 \end{proof}

We will now consider the debiased version of the above estimator.

\begin{definition}[Grier-Pashayan-Schaeffer's algorithm \cite{grier2024sample}]
    $\mc{A}_{GPS}$ is the \emph{debiased} version of Hayashi's POVM, i.e., given outcome $\psi$ from Hayashi's POVM, the output of $\mc{A}_{GPS}$ is $\hat{\sigma}_{\psi} = \frac{d+n}{n}\ket{\psi}\bra{\psi} - \frac{1}{n}I$.
\end{definition}

By \Cref{lem:hayashi-moments}, $\mc{A}_{GPS}$ is an unbiased estimator, i.e., $\mbb{E}[\hat{\sigma}_\psi] = (\psi_{\mrm{sym}})_1$. In the following lemma, \cite{pelecanos2025mixed} compute the second moment of $\mc{A}_{GPS}$ up to a truncated lower-order term. 

\begin{lemma}[Truncated second moment of $\mc{A}_{GPS}$ {\cite[Lemma 5.8]{pelecanos2025mixed}}]
    Let $\psi_{\mrm{sym}}$ be a mixed state on $\vee^n\mbb{C}^d$. If $\hat{\sigma}$ is the output of $\mc{A}_{GPS}$ when applied to $\psi_{\mrm{sym}}$, then
    \begin{align}
        \mbb{E}[\hat{\sigma} \otimes \hat{\sigma}] &= \frac{n-1}{n}\cdot (\psi_{\mrm{sym}})_{1,2} + \frac{1}{n}((\psi_{\mrm{sym}})_{1} \otimes I + I \otimes (\psi_{\mrm{sym}})_{1})\cdot \swap + \frac{1}{n^2}\cdot\swap - \mathrm{Lower}_{\psi_{\mrm{sym}}},
    \end{align}
    where $\mathrm{Lower}_{\psi_{\mrm{sym}}} \in \mathrm{SoS}(d)$.
\end{lemma}

In the above, $\mathrm{SoS}(d) \subseteq \mbb{C}^{d^2 \times d^2}$ denotes the set of matrices that can be represented as non-negative linear combinations of matrices of the form $X \otimes X$, where each $X \in \mbb{C}^{d \times d}$ is Hermitian. Now, we state an exact version of the above lemma by explicitly incorporating the lower-order term.

\begin{lemma}[Second moment of $\mc{A}_{GPS}$; refined version of {\cite[Lemma 5.8]{pelecanos2025mixed}}]
\label{lem:GPS-second-moment}
    Let $\psi_{\mrm{sym}}$ be a mixed state on $\vee^n\mbb{C}^d$. If $\hat{\sigma}$ is the output of $\mc{A}_{GPS}$ when applied to $\psi_{\mrm{sym}}$, then
    \begin{align}
        \mbb{E}[\hat{\sigma} \otimes \hat{\sigma}] &= \frac{n-1}{n}\cdot\frac{d+n}{d+n+1}\cdot (\psi_{\mrm{sym}})_{1,2} + \frac{1}{n}\cdot\frac{d+n}{d+n+1}((\psi_{\mrm{sym}})_{1}\otimes I + I \otimes (\psi_{\mrm{sym}})_{1})\cdot \swap \nonumber\\& -\frac{1}{n(d+n+1)}((\psi_{\mrm{sym}})_{1}\otimes I + I \otimes (\psi_{\mrm{sym}})_{1}) + \frac{1}{n^2}\cdot \frac{d+n}{d+n+1}\cdot \swap - \frac{I^{\otimes 2}}{n^2(d+n+1)}.
    \end{align}
\end{lemma}
\begin{proof}
    Recall that for an outcome $\psi$ from Hayashi's POVM, 
    \begin{equation}
        \hat{\sigma}_\psi = \frac{d+n}{n}\ket{\psi}\bra{\psi} - \frac{1}{n} I.
    \end{equation}
    Thus, 
    \begin{align}
        \mbb{E}[\hat{\sigma} \otimes \hat{\sigma}] &= \mbb{E}_{\psi}[\hat{\sigma}_\psi^{\otimes 2}] = \frac{(d+n)^2}{n^2} \mbb{E}[\ket{\psi}\bra{\psi}^{\otimes 2}] - \frac{d+n}{n^2} (\mbb{E}[\ket{\psi}\bra{\psi}] \otimes I + I \otimes \mbb{E}[\ket{\psi}\bra{\psi}]) + \frac{I^{\otimes 2}}{n^2}.
    \end{align}
    Using \Cref{lem:hayashi-moments}, the statement follows from direct calculation, but we state the full proof for the sake of completeness. 
    \begin{align}
        \mbb{E}[\hat{\sigma^{\otimes 2}}] &= \frac{d+n}{n^2(d+n+1)}(I^{\otimes 2} + \swap) + \frac{d+n}{n(d+n+1)}((\psi_{\mrm{sym}})_1 \otimes I + I \otimes (\psi_{\mrm{sym}})_1)(I^{\otimes 2} + \swap) 
        \nonumber\\&\hspace{1em} + \frac{n-1}{n}\cdot\frac{d+n}{d+n+1} \cdot (\psi_{\mrm{sym}})_{1,2} - \frac{1}{n^2}(2I^{\otimes 2} + n(\psi_{\mrm{sym}})_1 \otimes I + n I \otimes (\psi_{\mrm{sym}})_1) + \frac{I^{\otimes 2}}{n^2}
        \\&= \frac{n-1}{n}\cdot\frac{d+n}{d+n+1} \cdot (\psi_{\mrm{sym}})_{1,2} + \frac{1}{n} \cdot \frac{d+n}{d+n+1} ((\psi_{\mrm{sym}})_1 \otimes I + I \otimes (\psi_{\mrm{sym}})_1) \cdot \swap
        \nonumber\\&\hspace{1em} + ((\psi_{\mrm{sym}})_1 \otimes I + I \otimes (\psi_{\mrm{sym}})_1) \cdot \frac{1}{n}\left(\frac{d+n}{d+n+1} - 1\right) + \frac{1}{n^2} \cdot \frac{d+n}{d+n+1} \cdot \swap + \frac{I^{\otimes 2}}{n^2} \left(\frac{d+n}{d+n+1} - 1\right)
        \\&= \frac{n-1}{n}\cdot\frac{d+n}{d+n+1} \cdot (\psi_{\mrm{sym}})_{1,2} + \frac{1}{n} \cdot \frac{d+n}{d+n+1} ((\psi_{\mrm{sym}})_1 \otimes I + I \otimes (\psi_{\mrm{sym}})_1) \cdot \swap
        \nonumber\\&\hspace{1em} - \frac{1}{n(d+n+1)}((\psi_{\mrm{sym}})_1 \otimes I + I \otimes (\psi_{\mrm{sym}})_1) + \frac{1}{n^2} \cdot \frac{d+n}{d+n+1} \cdot \swap -\frac{I^{\otimes 2}}{n^2(d+n+1)}. \qedhere
    \end{align}
\end{proof}

\subsection{Unbiased mixed state estimation with entangled measurements}
\label{sec:PTSW-moments}

Before presenting the unbiased state estimator of \cite{pelecanos2025mixed}, let us state some necessary definitions. We will make use of the Schur transform and weak Schur sampling, and provide only the minimal details necessary to understand their use in this work. For a detailed background on these operations and their use in quantum learning and testing problems, we refer to \cite{wright2016learn,o2023learning}.

\begin{definition}[Schur transform and weak Schur sampling]
\label{def:schur-transform}
    The \emph{Schur transform} $U_{\mathrm{schur}}^{d,n}$ is a $d^n \times d^n$ unitary that block diagonalizes $\rho^{\otimes n}$, for all $d$-dimensional quantum states $\rho$, with each block indexed by a partition $\lambda \vdash n$. Equivalently, we can write
    \begin{equation}
        U_{\mathrm{schur}}^{d,n} \rho^{\otimes n} (U_{\mathrm{schur}}^{d,n})^\dag = \bigoplus_{\lambda \vdash n} M_\lambda(\rho),
    \end{equation}
    for some positive semi-definite operators $M_\lambda(\rho)$. Let $\Pi_\lambda$ be the projector onto the block associated with $\lambda$. Then, weak Schur sampling on $\rho^{\otimes n}$ is the application of the POVM consisting of operators $\{(U_{\mathrm{schur}}^{d,n})^\dag \Pi_\lambda U_{\mathrm{schur}}^{d,n}\}_{\lambda \vdash n}$ to $\rho^{\otimes n}$.
\end{definition}

For any partition $\lambda$, we define its length $\ell(\lambda)$ to be the number of its non-zero parts. Next, we define the random purification channel.

\begin{definition}[Random purification channel \cite{tang2025conjugate,pelecanos2025mixed}]
\label{def:random-purification-channel}
    There exists a channel $\purify^{d,r}$ that takes as input $n$ copies of any rank-$r$ state $\rho$ and outputs $\mbb{E}_{\ket{\bfrho}}[\ket{\bfrho}\bra{\bfrho}^{\otimes n}]$, where the expectation is taken over all random purifications $\ket{\bfrho} \in \mbb{C}^d \otimes \mbb{C}^r$ of $\rho$.
\end{definition}

 Now, we present the main estimation algorithm of \cite{pelecanos2025mixed}.

\label{sec:unbiased-state-estimation-moments}
\begin{algorithm}[H]
    \begin{algorithmic}[1]
    \caption{Unbiased state estimation via random purification; \cite[Figure 11]{pelecanos2025mixed}}
    \label{alg:quasi-purification-estimation}
    \State \textbf{Input:} $\rho^{\otimes t}$.
    \State Apply weak Schur sampling to $\rho^{\otimes t}$. Let the outcome be $\lambda \vdash t$, with the state collapsing to $\rho|_\lambda$.
    \State Apply the purification channel $\Phi_{\mrm{purify}}^{d,\ell(\lambda)}$ and then the inverse Schur transform $(U_{\mrm{Schur}}^{d,t} \otimes U_{\mrm{Schur}}^{\ell(\lambda),t})^\dag$. Let the resulting state be $\tau_\lambda$.
    \State Apply the Grier-Pashayan-Schaeffer algorithm to learn an estimate $\hat{\sigma}_\lambda$.
    \State Set $\hat{\rho}_\lambda = \tr_2(\hat{\sigma}_\lambda)$. Return $\hat{\rho}_\lambda$.
    \end{algorithmic}
\end{algorithm}

As mentioned previously, this estimator is unbiased, i.e., $\mbb{E}[\hat{\rho}_\lambda] = \rho$. We now state the second moment of this estimator:

\begin{theorem}[Second moment of \Cref{alg:quasi-purification-estimation}; {\cite[Theorem 6.4]{pelecanos2025mixed}}]
\label{thm:quasi-purification-truncated-second-moment}
    Let $\hat{\rho}_\lambda$ be the output of \Cref{alg:quasi-purification-estimation}. Then,
    \begin{equation}
        \mathbb{E}[\hat{\rho}_\lambda \otimes \hat{\rho}_\lambda] = \frac{t-1}{t} \cdot \rho^{\otimes 2} + \frac{1}{t} (\rho \otimes I + I \otimes \rho) \cdot \swap + \frac{\mbb{E}[\ell(\lambda)]}{t^2}\cdot\swap - \mathrm{Lower}_{\rho},
    \end{equation}
    where $\mathrm{Lower}_{\rho} \in \mrm{SoS}(d)$.
\end{theorem}

While explicitly computing the lower-order term in the above theorem seems intractable, we are able to do so \emph{conditioned} on measuring a particular partition $\lambda$ in \Cref{alg:quasi-purification-estimation}.

\begin{lemma}
\label{lem:quasi-purification-conditioned-second-moment}
    Let $\hat{\rho}_\lambda,\tau_\lambda$ be as in \Cref{alg:quasi-purification-estimation}. Denote the register corresponding to $\rho|_\lambda$ by $\mathsf{A}$, and let the purifying register be $\mathsf{B}$. Let $D \triangleq d \cdot \ell(\lambda)$. Then, 
    \begin{align}
        \mbb{E}[\hat{\rho}_\lambda \otimes \hat{\rho}_\lambda | \lambda] &= \frac{t-1}{t}\cdot \frac{D+t}{D+t+1}\cdot(\tau_\lambda)_{\msf{A}_1,\msf{A}_2} + \frac{1}{t}\cdot \frac{D+t}{D+t+1}((\tau_\lambda)_{\msf{A}_1} \otimes I + I \otimes (\tau_\lambda)_{\msf{A}_1}) \cdot \swap
        \nonumber\\& - \frac{\ell(\lambda)}{t(D+t+1)}((\tau_\lambda)_{\msf{A}_1} \otimes I + I \otimes (\tau_\lambda)_{\msf{A}_1}) + \frac{\ell(\lambda)}{t^2}\cdot\frac{D+t}{D+t+1}\swap - \frac{\ell(\lambda)^2}{t^2(D+t+1)}\cdot I^{\otimes 2}.
    \end{align}
\end{lemma}

\begin{proof}
     We will associate $\hat{\sigma}_\lambda^{\otimes 2}$ with registers $\mathsf{A}_1, \mathsf{A}_2, \mathsf{B}_1, \mathsf{B}_2$. As $\tau_\lambda$ is in the symmetric subspace, we use \Cref{lem:GPS-second-moment} and obtain
    \begin{align}
         \mbb{E}[\hat{\sigma}_\lambda^{\otimes 2} | \lambda] &= \frac{t-1}{t}\cdot\frac{D+t}{D+t+1}\cdot (\tau_\lambda)_{1,2} + \frac{1}{t}\cdot\frac{D+t}{D+t+1}(\tau_\lambda)_{1}\otimes I + I\otimes (\tau_\lambda)_{1})\cdot \swap \nonumber\\& -\frac{1}{t(D+t+1)}((\tau_\lambda)_{1}\otimes I + I \otimes (\tau_\lambda)_{1}) + \frac{1}{t^2}\cdot \frac{D+t}{D+t+1}\cdot \swap - \frac{I^{\otimes 2}}{t^2(D+t+1)}.
        \label{eq:quasi-purification-conditioned-second-moment-1}    
    \end{align}
    Further, note that
    \begin{equation}
        \mbb{E}[\hat{\rho}_\lambda^{\otimes 2} | \lambda] =  \mbb{E}[\tr_{\mathsf{B}_1,\mathsf{B}_2}(\hat{\sigma}_\lambda^{\otimes 2}) | \lambda] = \tr_{\mathsf{B}_1,\mathsf{B}_2}(\mbb{E}[\hat{\sigma}_\lambda^{\otimes 2} | \lambda]).
    \end{equation}
    So, we only need to trace out the $\mathsf{B}$ registers from each term in \Cref{eq:quasi-purification-conditioned-second-moment-1}. The partial trace of the first term can be stated trivially. For the second term, using \cite[Proposition 2.5]{pelecanos2025mixed},
    \begin{equation}
        \tr_{\mathsf{B}}((\tau_\lambda)_{\mathsf{A}_1,\mathsf{B}_1} \otimes I_{\mathsf{A}_2, \mathsf{B}_2} \cdot \swap_{\mathsf{A,B}}) = (\tau_\lambda)_{\mathsf{A}_1} \otimes I_{\mathsf{A}_2} \cdot \swap_{\mathsf{A}_1,\mathsf{A}_2},
    \end{equation}
    and similarly for the $(I \otimes (\tau_\lambda)_1)\cdot \swap$ term. For the third term, we have
    \begin{equation}
        \tr_{\mathsf{B}}((\tau_{\lambda})_{\mathsf{A}_1, \mathsf{B}_1} \otimes I_{\mathsf{A}_2, \mathsf{B}_2}) = \ell(\lambda) \cdot (\tau_{\lambda})_{\mathsf{A}_1} \otimes I_{\mathsf{A}_2},
    \end{equation}
    as we trace out the identity on one purifying register. The $(I \otimes (\tau_\lambda)_1)$ term is handled similarly. For the swap term, we write
    \begin{equation}
        \tr_\mathsf{B}(\swap_{\mathsf{A,B}}) = \tr_\mathsf{B}(\swap_{\mathsf{A}} \otimes \swap_{\mathsf{B}}) = \ell(\lambda) \cdot \swap_{\mathsf{A}}.
    \end{equation}
    Finally, the $I^{\otimes 2}$ term gains a $\ell(\lambda)^2$ factor as we trace out two purifying registers.
\end{proof}

We will also use the following helper lemma from \cite{pelecanos2025mixed} when averaging over expressions involving $\tau_\lambda$:
\begin{lemma}[{\cite[Lemma 6.2]{pelecanos2025mixed}}]
\label{lem:partial-trace-tau-lambda}
        For $1 \leq k \leq n$, 
        \begin{equation}
            \mbb{E}_\lambda[(\tau_\lambda)_{\mathsf{A}_1, \dots, \mathsf{A}_k}] = \rho^{\otimes k}. 
        \end{equation}
\end{lemma}

Lastly, we will make use of an upper bound on the expected length of a partition obtained from weak Schur sampling.

\begin{lemma}[see e.g., {\cite[Lemma 6.2]{pelecanos2025debiased}}]
    \label{lem:partition-length-upper-bound}
    Let $\rho \in \mathbb{C}^{d \times d}$ be a quantum state. Let $\lambda \vdash t$ be the partition obtained from weak Schur sampling applied to $\rho^{\otimes t}$. Then,
    \begin{equation}
        \mathbb{E}[\ell(\lambda)] \leq \min\{2\sqrt{t},d\}.
    \end{equation}
\end{lemma}
\section{State certification upper bounds}
\label{sec:main-upper-bound}
In this section, we prove our main upper bound for state certification, which we restate below:
\begin{theorem}[Upper bound from \Cref{thm:intro-state-certification-main} restated]
\label{thm:tcopy-certification-upper}
    Let $d \geq 2, t \geq 1, \eps > 0$. There exists an algorithm for $\eps$-trace norm certification of $d$-dimensional states with confidence $1-\delta$ using fixed, $t$-copy measurements that has copy complexity $\Tilde{\bigo}\left(\max\left\{\frac{d^2}{\sqrt{t}\eps^2}, \frac{d}{\eps^2}\right\} \cdot \log(1/\delta)\right)$.
\end{theorem}

In fact, we will directly prove a more general instance-dependent bound:

\begin{theorem}
    \label{thm:instance-dependent-tcopy-certification}
    Let $d \geq 2, t \geq 1, \eps > 0$, and let $\sigma$ be a $d$-dimensional quantum state. Let $\sigma^*$ be the operator given by removing as many of $\sigma$'s smallest eigenvalues as possible until at most $\eps^2/20$-mass is removed. Then, there exists an algorithm to $\eps$-certify $\sigma$ with respect to the trace distance with confidence $1-\delta$ using fixed $t$-copy measurements that has copy complexity
    \begin{equation}
        \Tilde{\bigo}\left(\left(1 + \frac{\mathrm{rank}(\sigma^*)}{\sqrt{t}}\right) \cdot \frac{\|\sigma^{*}\|_{1/2}}{\eps^2} + \frac{\|\sigma^*\|_{1/3}}{\sqrt{t}\eps^2} + \frac{\rank(\sigma^*)}{\eps^2}\right) \cdot \log(1/\delta).
    \end{equation}
\end{theorem}

Using the above instance-dependent bound, we can now prove \Cref{thm:tcopy-certification-upper}.

\begin{proof}[Proof of \Cref{thm:tcopy-certification-upper}]
    To prove the desired worst-case bound, we simply need to provide worst-case upper bounds for each $\sigma$-dependent quantity in the copy complexity of \Cref{thm:instance-dependent-tcopy-certification}. It is not hard to see that each quantity above is maximized for $\sigma = \mmstate$, giving
    \begin{equation}
        \mathrm{rank}(\sigma^*) \leq d, \quad \|\sigma^*\|_{1/2} \leq \|\mmstate\|_{1/2} = d, \quad \text{and} \quad \|\sigma^*\|_{1/3} \leq \|\mmstate\|_{1/3} = d^2.
    \end{equation}
    Thus the overall copy complexity is at most 
    \begin{equation}
        \Tilde{\bigo}\left( \frac{d}{\eps^2} + \frac{d^2}{\sqrt{t}\eps^2} \right) \cdot \log(1/\delta),
    \end{equation}
    as desired.
\end{proof}

The rest of this section is devoted to proving \Cref{thm:instance-dependent-tcopy-certification}. In \Cref{sec:upper-HS-closeness}, we prove \Cref{thm:closeness-l2-estimation}, a worst-case upper bound for state certification with respect to the Hilbert--Schmidt norm. In \Cref{sec:upper-instance-dependent}, we state a standard reduction from instance-dependent trace-distance certification to worst-case Hilbert--Schmidt certification. Then, we combine this reduction with \Cref{thm:closeness-l2-estimation} to obtain \Cref{thm:instance-dependent-tcopy-certification}.

\subsection{Hilbert--Schmidt state certification}
\label{sec:upper-HS-closeness}

The main result of this section is a worst-case upper bound for Hilbert--Schmidt state certification with fixed $t$-copy measurements. However, we note that this result holds more generally for the problem of closeness testing, and we state and prove this general version below:

\begin{theorem}
\label{thm:closeness-l2-estimation}
    Let $d \geq 2, t \geq 1, \eps > 0$. Given $m = nt$ copies each of $d$-dimensional states $\rho,\sigma$, there exists an algorithm to test whether $\rho = \sigma$ or $\|\rho - \sigma\|_2 \geq \eps$ with confidence $1-\delta$ using fixed $t$-copy measurements with copy complexity 
    \begin{equation}
        m \leq \bigo\left(\max\left\{ 
        \frac{1}{\eps^2}, \frac{d}{\sqrt{t}\eps^2}, \frac{\sqrt{d \cdot \tr(\rho\sigma)}}{\eps^2}, \frac{\sqrt{d}}{\eps}
        \right\} \cdot \log(1/\delta)\right) .
    \end{equation}
\end{theorem}

While it's unclear whether one can remove the term depending on $\tr(\rho\sigma)$ in general, this term can be neglected when one of the states has small operator norm, leading us to the following bound:
\begin{corollary}
\label{cor:balanced-l2-closeness-testing}
    Let $d \geq 2, t \geq 1, \eps > 0$. Let $\rho,\sigma$ be $d$-dimensional states with $\min\{\|\rho\|_{\infty},\|\sigma\|_{\infty}\} \leq \frac{c}{d}$, for some absolute constant $c \geq 1.$ Then, given $m = nt$ copies of $\rho,\sigma$ each, there exists an algorithm to test whether $\rho = \sigma$ or whether $\|\rho - \sigma\|_2 \geq \eps$ with copy complexity
    \begin{equation}
        m \leq \bigo\left(\max\left\{
        \frac{1}{\eps^2}, \frac{d}{\sqrt{t}\eps^2},  \frac{\sqrt{d}}{\eps}
        \right\}\right).
    \end{equation}
\end{corollary}

Note that we can apply the above corollary to the task of mixedness testing by setting $\sigma = \mmstate$, which has operator norm $\frac1d$. In this case, by \Cref{fact:hs-to-trace-norm-testing}, the complexity of mixedness testing with respect to the trace norm would be 
\begin{equation}
    \bigo\Bigl(\max\Bigl\{\frac{d}{\eps^2},\frac{d^2}{\sqrt{t}\eps^2},\frac{d}{\eps}\Bigr\}\Bigr) = \bigo\Bigl(\max\Bigl\{\frac{d}{\eps^2},\frac{d^2}{\sqrt{t}\eps^2}\Bigr\}\Bigr).
\end{equation}
Thus, \Cref{cor:balanced-l2-closeness-testing} allows us to recover \Cref{thm:intro-mixedness-upper}. 

Before moving on to the proof of \Cref{thm:closeness-l2-estimation}, we will first provide a generic reduction from Hilbert--Schmidt squared estimation to the task of state tomography. As a warmup, in \Cref{sec:closeness-testing-uniform-povm}, we will instantiate the above estimator with the uniform POVM to provide an alternate algorithm achieving the tight $\bigo(d^2/\eps^2)$ upper bound for trace-distance closeness testing, first shown by Liu and Acharya~\cite{liu2024role}. Finally, we will instantiate the Hilbert--Schmidt squared estimator with \Cref{alg:quasi-purification-estimation} to prove \Cref{thm:closeness-l2-estimation} in \Cref{sec:closeness-HS-upper-proof}.

\subsubsection{A generic Hilbert--Schmidt estimator}

In this section, we present our algorithm for estimating the Hilbert--Schmidt squared distance using any state estimator, and analyze its variance in a similar manner to our previous algorithm for purity estimation (\Cref{alg:collision-estimation}).

\begin{algorithm}[H]
    \begin{algorithmic}[1]
    \caption{HilbertSchmidt($\mc{A}, \sigma, \rho$), given an algorithm $\mathcal{A}$ that takes as input $\rho^{\otimes t}$ and returns a state $\hat{\rho} \in \mbb{C}^{d \times d}$.}
    \label{alg:closeness-l2-estimation}
    \State \textbf{Input:} $n \cdot t$ copies each of unknown states $\rho, \sigma$.
    \For{$i \in [n]$}
        \State $\hat{\rho}_i \gets \mc{A}(\rho^{\otimes t})$, $\hat{\sigma}_i \gets \mc{A}(\sigma^{\otimes t})$.
        \State $\hat{\Delta}_i \gets \hat{\rho}_i - \hat{\sigma}_i$.
    \EndFor
    \For{$i,j \in [n]$ with $i < j$}
        \State $X_{i,j} \gets \tr(\hat{\Delta}_i\hat{\Delta}_j)$.
    \EndFor
    \Return $\bar{X} \gets \frac{1}{\binom{n}{2}} \sum_{i < j} X_{i,j}$
    \end{algorithmic}
\end{algorithm}

In general, the variance of such algorithms depends on the second moment of the underlying state estimator. We state this dependence explicitly in the following lemma:

\begin{lemma}
\label{lem:hs-estimation-variance-generic}
    Let $\bar{X}$ be the output of $\mathrm{HilbertSchmidt}(\mc{A}, \rho, \sigma)$ (\Cref{alg:closeness-l2-estimation}). Let $\Delta = \rho - \sigma$. Then, $\mbb{E}[\bar{X}] = \tr(\mbb{E}[\hat{\Delta}]^2)$ and
    \begin{equation}
    \label{eq:closeness-variance-generic-1}
          \mbb{V}[\bar{X}] = \bigo(n^{-1}) \cdot (\tr(\mbb{E}[\hat{\Delta}^{\otimes 2}] \cdot \mbb{E}[\hat{\Delta}]^{\otimes 2}) - \mbb{E}[\bar{X}]^2) + \bigo(n^{-2}) \cdot (\tr(\mbb{E}[\hat{\Delta}^{\otimes 2}]^2) - \mbb{E}[\bar{X}]^2).
    \end{equation}
    In particular, when $\mc{A}$ is unbiased, i.e., $\mbb{E}[\hat{\rho}] = \rho$, we have $\mbb{E}[\bar{X}] = \tr(\Delta^2) = \|\rho - \sigma\|_2^2$ and
    \begin{equation}
    \label{eq:closeness-variance-generic-2}
        \mbb{V}[\bar{X}] = \bigo(n^{-1}) \cdot (\tr(\mbb{E}[\hat{\Delta}^{\otimes 2}] \cdot \Delta^{\otimes 2}) - \tr(\Delta^2)^2) + \bigo(n^{-2}) \cdot (\tr(\mbb{E}[\hat{\Delta}^{\otimes 2}]^2) - \tr(\Delta^2)^2).
    \end{equation}    
\end{lemma}
\begin{proof}
    Note that for each pair $i,j$, $\mbb{E}[X_{i,j}] = \tr(\mbb{E}[\hat{\Delta}_i]\mbb{E}[\hat{\Delta}_j]) = \tr(\mbb{E}[\hat{\Delta}]^2)$, as $\hat{\Delta}_i,\hat{\Delta}_j$ are independent. By linearity, we also have $\mbb{E}[\bar{X}] = \tr(\mbb{E}[\hat{\Delta}]^2)$. We will now prove the variance bound, \Cref{eq:closeness-variance-generic-1}.

    In order to bound the second moment of $\bar{X}$, we write 
    \begin{equation}
      \binom{n}{2}^2 \bar{X}^2 = (\sum_{i < j} X_{i,j})^2 = \sum_{i < j, k < l} X_{i,j}X_{k,l}.  
    \end{equation}
    Note that for each term in the above summation, $X_{i,j}$ and $X_{k,l}$ are correlated random variables whenever they have an index or two in common. To bound the expectation of their products, we thus split the above sum into 3 different cases. Our casework mirrors that of the analysis of the classical collision-based uniformity tester in \cite[Section 2.1.2]{canonne2022topics}.

    \paragraph{Case 1:} First, we consider the terms where $(i,j)$ and $(k,l)$ have no indices in common. There are $\binom{n}{2}\binom{n-2}{2} = 6 \binom{n}{4}$ such terms, and each is a product of independent random variables. Thus, we have
    \begin{equation} 
        \mbb{E}[X_{i,j}X_{k,l}] = \mbb{E}[X_{i,j}] \mbb{E}[X_{k,l}] = \mbb{E}[\bar{X}]^2,  \label{eq:generic-variance-case1}.
    \end{equation}
    
    \paragraph{Case 2:} Next, consider the case when there is one common index; without loss of generality, say we wish to bound $\mbb{E}[X_{i,j}X_{i,k}]$ for distinct $i,j,k$. Further, note that there are $6 \binom{n}{3}$ such terms in the expansion of $\bar{X}^2$.
    Now, we have
    \begin{align}
        \mbb{E}[X_{i,j}X_{i,k}] &= \mbb{E}[\tr(\hat{\Delta}_i \hat{\Delta}_j) \tr(\hat{\Delta}_i \hat{\Delta}_k)]
        \\&= \mbb{E}[\tr(\hat{\Delta}_i \otimes \hat{\Delta}_i \cdot \hat{\Delta}_j \otimes \hat{\Delta}_k)]
        \\&= \tr(\mbb{E}[\hat{\Delta}_i \otimes \hat{\Delta}_i] \cdot \mbb{E}[\hat{\Delta}_j] \otimes \mbb{E}[\hat{\Delta}_k])
        \\&= \tr(\mbb{E}[\hat{\Delta} \otimes \hat{\Delta}] \cdot \mbb{E}[\hat{\Delta}] \otimes \mbb{E}[\hat{\Delta}]).
        \label{eq:generic-variance-case2}
    \end{align}

    \paragraph{Case 3:} Finally, we consider the case when the indices are identical, i.e., we wish to bound terms of the form $\mbb{E}[X_{i,j}^2]$ for distinct $i,j$. Note that there are $\binom{n}{2}$ such terms.
    \begin{align}
        \mbb{E}[X_{i,j}^2] &= \mbb{E}[\tr(\hat{\Delta}_i \hat{\Delta}_j) \tr(\hat{\Delta}_i \hat{\Delta}_j)]
        \\&= \mbb{E}[\tr(\hat{\Delta}_i \otimes \hat{\Delta}_i \cdot \hat{\Delta}_j \otimes \hat{\Delta}_j)]
        \\&= \tr(\mbb{E}[\hat{\Delta}_i \otimes \hat{\Delta}_i] \cdot \mbb{E}[\hat{\Delta}_j \otimes \hat{\Delta}_j])
        \\&= \tr(\mbb{E}[\hat{\Delta} \otimes \hat{\Delta}]^2).
        \label{eq:generic-variance-case3}
    \end{align}

    \paragraph{Putting things together:} We can finally bound the overall variance of $\bar{X}$ using \Cref{eq:generic-variance-case1,eq:generic-variance-case2,eq:generic-variance-case3}.
    \begin{align}
        \mbb{V}[\bar{X}] &= \mbb{E}[\bar{X}^2] - \mbb{E}[\bar{X}]^2
        \\&= \frac1{\binom{n}{2}^2} \left(6 \binom{n}{4} \cdot \mbb{E}[\bar{X}]^2 + 6\binom{n}{3} \cdot \tr(\mbb{E}[\hat{\Delta} \otimes \hat{\Delta}] \cdot \mbb{E}[\hat{\Delta}] \otimes \mbb{E}[\hat{\Delta}]) + \binom{n}{2} \cdot \tr(\mbb{E}[\hat{\Delta} \otimes \hat{\Delta}]^2)\right) - \mbb{E}[\bar{X}]^2
        \\&= \frac{6 \binom{n}{3}}{\binom{n}{2}^2} \cdot (\tr(\mbb{E}[\hat{\Delta} \otimes \hat{\Delta}] \cdot \mbb{E}[\hat{\Delta}] \otimes \mbb{E}[\hat{\Delta}]) - \mbb{E}[\bar{X}]^2) + \frac{1}{\binom{n}{2}}(\tr(\mbb{E}[\hat{\Delta} \otimes \hat{\Delta}]^2) - \mbb{E}[\bar{X}]^2)
    \end{align}
    where we used $\binom{n}{2}^2 = 6\binom{n}{3} + 6\binom{n}{4} + \binom{n}{2}$ in the final step, proving \Cref{eq:closeness-variance-generic-1}. \Cref{eq:closeness-variance-generic-2} can be proven in the unbiased case by subsituting $\mbb{E}[\hat{\Delta}] = \Delta$.
\end{proof}

\subsubsection{Warmup: closeness testing with the uniform POVM}
\label{sec:closeness-testing-uniform-povm}

In this section, we will instantiate the framework of \Cref{alg:closeness-l2-estimation} with the uniform POVM, and recover the $\bigo(d^2/\eps^2)$ upper bound for quantum closeness testing with fixed single-copy measurements. Using \Cref{fact:hs-to-trace-norm-testing}, it suffices to show the following theorem:

\begin{theorem}
\label{thm:closeness-testing-uniform-povm}
    Let $d \geq 2, \eps > 0$. Let $\rho,\sigma \in \mbb{C}^{d \times d}$ be quantum states. Then, one can test whether $\rho = \sigma$ or $\|\rho - \sigma\|_2 \geq \eps$ using fixed, unentangled measurements on $\bigo(d/\eps^2)$ copies of $\rho,\sigma$ each. 
\end{theorem}

Our proof will use the following variance bound:

\begin{lemma}[Variance of $\mathrm{HilbertSchmidt}(\mc{A}_{\mrm{unif}}, \sigma, \rho)$]
\label{lem:closeness-variance-uniform-povm}
Let $\bar{X}$ be the output of $\mathrm{HilbertSchmidt}(\mc{A}_{\mrm{unif}}, \sigma, \rho)$. Let $\Delta \triangleq \rho - \sigma$. Then, 
\begin{equation}
    \mbb{V}[\bar{X}] = \bigo\left(\frac{\tr(\Delta^2)}{nd^4} + \frac{1}{n^2d^2}\right).
\end{equation}
\end{lemma}

Let us first prove \Cref{thm:closeness-testing-uniform-povm} using \Cref{lem:closeness-variance-uniform-povm}.

\begin{proof}[Proof of \Cref{thm:closeness-testing-uniform-povm}]
     Let $\bar{X}$ be the output of $\mathrm{HilbertSchmidt}(\mc{A}_{\mrm{unif}}, \sigma,\rho)$. By \Cref{lem:uniform-povm-first-moment}, $\mbb{E}[\hat{\rho}] = \frac{I}{d+1} + \frac{\rho}{d+1}$. Then, using \Cref{lem:hs-estimation-variance-generic},
    \begin{equation}
         \mbb{E}[\bar{X}] = \tr\left(\left(\frac{I}{d+1} + \frac{\rho}{d+1} - \frac{I}{d+1} - \frac{\sigma}{d+1}\right)^2\right) = \frac{\|\rho-\sigma\|_2^2}{(d+1)^2}.
    \end{equation}
    Thus, when $\rho = \sigma$, $\mbb{E}[\bar{X}] = 0$, and when $\|\rho - \sigma\|_2 \geq \eps$, $\mbb{E}[\bar{X}] \geq \eps^2/(d+1)^2$.
     
    Our tester will thus compare $\bar{X}$ to the threshold $\frac{\eps^2}{2(d+1)^2}$. First, consider the case when $\rho = \sigma$. We bound the probability of an incorrect output by Chebyshev's inequality:
    \begin{equation}
       \mathrm{Pr}\left[\bar{X} \geq \frac{\eps^2}{2(d+1)^2}\right] \leq  \bigo\left(\frac{1}{n^2d^2}\right) \cdot \frac{4(d+1)^4}{\eps^4} \leq 1/3,
    \end{equation}
    where we used \Cref{lem:closeness-variance-uniform-povm} with $\Delta = 0$ in the first inequality, and the last inequality holds for $n \geq Cd/\eps^2$, with $C$ a sufficiently large constant.

    On the other hand, when $\|\rho - \sigma\|_2 \geq \eps$, let $\Delta \triangleq \rho - \sigma$. Then, we have

    \begin{align}
        \mathrm{Pr}\left[\bar{X} \leq \frac{\eps^2}{2(d+1)^2}\right] &\leq \mathrm{Pr}\left[\bar{X} \leq \frac{\tr(\Delta^2)}{2(d+1)^2}\right]  \\&\leq \bigo\left(\frac{\tr(\Delta^2)}{nd^4} + \frac{1}{n^2d^2}\right) \cdot \frac{4(d+1)^4}{\tr(\Delta^2)^2} 
        \\&\leq \bigo\left(\frac{1}{n\cdot\tr(\Delta^2)} + \frac{d^2}{n^2\cdot\tr(\Delta^2)^2}\right)
        \\&\leq \bigo\left(\frac{1}{n\eps^2} + \frac{d^2}{n^2\eps^4}\right),
    \end{align}
    where the second inequality used \Cref{lem:closeness-variance-uniform-povm}. The above probability is at most $1/3$ when
    \begin{equation}
        n \geq C \cdot \max \left\{\frac{1}{\eps^2},\frac{d}{\eps^2}\right\} = \frac{Cd}{\eps^2},
    \end{equation}
    for a sufficiently large constant $C$. 

    Thus, it suffices to take $n \leq \bigo(d/\eps^2)$ copies to succeed with probability at least $2/3$ in either case.
\end{proof}

\begin{proof}[Proof of \Cref{lem:closeness-variance-uniform-povm}]
    We will prove this variance bound using \Cref{lem:hs-estimation-variance-generic}. We start by noting that $\mbb{E}[\hat{\Delta}] = \mbb{E}[\hat{\rho}] - \mbb{E}[\hat{\sigma}] = \frac{\Delta}{d+1}$, where we used \Cref{lem:uniform-povm-first-moment}. Now, using \Cref{lem:uniform-povm-first-moment,lem:uniform-povm-second-moment}, let us explicitly state the second moment of $\hat{\Delta}$.

    \begin{align}
        \mathbb{E}[\hat{\Delta}^{\otimes 2}] &= \mathbb{E}[\hat{\rho}^{\otimes 2}] + \mathbb{E}[\hat{\sigma}^{\otimes 2}] - \mathbb{E}[\hat{\rho} \otimes \hat{\sigma}] - \mathbb{E}[\hat{\sigma} \otimes \hat{\rho}]
        \\&= \frac{2\cdot I^{\otimes 2} + 2 \cdot \swap + ((\rho+\sigma) \otimes I + I \otimes (\rho + \sigma) ) \cdot (I^{\otimes 2} + \swap)}{(d+1)(d+2)}\nonumber
        \\&\hspace{1em}- \frac{(I + \rho)\otimes(I + \sigma) + (I + \sigma)\otimes(I + \rho)}{(d+1)^2}
        \\&= \frac{2\cdot \swap + ((\rho+\sigma) \otimes I + I \otimes (\rho + \sigma) ) \cdot \swap}{(d+1)(d+2)} - \frac{\rho \otimes \sigma + \sigma \otimes \rho}{(d+1)^2}\nonumber
        \\&\hspace{1em} - \frac{2 \cdot I^{\otimes 2} + (\rho+\sigma) \otimes I + I \otimes (\rho + \sigma) }{(d+1)^2(d+2)}
        \\&= \frac{2\cdot \swap + ((\rho+\sigma) \otimes I + I \otimes (\rho + \sigma) ) \cdot \swap}{(d+1)(d+2)} + \frac{\Delta^{\otimes 2} - \rho^{\otimes 2} - \sigma^{\otimes 2}}{(d+1)^2} \nonumber
        \\&\hspace{1em}  - \frac{(I + \rho)^{\otimes 2} + (I+\sigma)^{\otimes 2} - \rho^{\otimes 2} - \sigma^{\otimes 2}}{(d+1)^2(d+2)}
        \\&= \frac{2\cdot \swap + ((\rho+\sigma) \otimes I + I \otimes (\rho + \sigma) ) \cdot \swap - \rho^{\otimes 2} - \sigma^{\otimes 2}}{(d+1)(d+2)}\\\nonumber&\hspace{1em} + \frac{\Delta^{\otimes 2}}{(d+1)^2} - \frac{(I+ \rho)^{\otimes 2} + (I + \sigma)^{\otimes 2}}{(d+1)^2(d+2)}.
        \label{eq:closeness-uniform-povm-variance-1}
    \end{align}
    Now, keeping \Cref{lem:hs-estimation-variance-generic} in mind, we will first upper bound $\tr(\mbb{E}[\hat{\Delta}^{\otimes 2}] \cdot \mbb{E}[\hat{\Delta}]^{\otimes 2})$.
    \begin{align}
        \tr(\mbb{E}[\hat{\Delta}^{\otimes 2}] \cdot \mbb{E}[\hat{\Delta}]^{\otimes 2}) &= \frac{2\tr(\Delta^2) + 2\tr(\Delta^2(\rho + \sigma))}{(d+1)^3(d+2)} + \frac{\tr(\Delta^2)^2}{(d+1)^4} - \frac{\tr(\Delta\rho)^2 + \tr(\Delta\sigma)^2}{(d+1)^2}
        \\&\leq \frac{\tr(\Delta^2)^2}{(d+1)^4} + \frac{6\tr(\Delta^2)}{(d+1)^4}
        \\&= \mbb{E}[\bar{X}]^2 + \bigo\left(\frac{\tr(\Delta^2)}{d^4}\right),\label{eq:closeness-uniform-povm-variance-2}
    \end{align}
    where we used the fact that $\Delta$ is traceless in the first equality and that $\rho + \sigma \preceq 2\cdot I$ in the inequality.

    Next, we will aim to upper bound $\tr(\mbb{E}[\hat{\Delta}^{\otimes 2}]^2)$. First, note that for any state $\varrho$, $\tr(\mbb{E}[\hat{\Delta}^{\otimes 2}] \cdot (I + \varrho)^{\otimes 2})$ and $\tr(\mbb{E}[\hat{\Delta}^{\otimes 2}] \varrho^{\otimes 2})$ are both non-negative, and so, using \Cref{eq:closeness-uniform-povm-variance-1}, we have
    \begin{equation}
        \tr(\mbb{E}[\hat{\Delta}^{\otimes 2}]^2) \leq \left\langle \mbb{E}[\hat{\Delta}^{\otimes 2}], \frac{2\cdot \swap + ((\rho + \sigma) \otimes I + I \otimes (\rho + \sigma))\cdot \swap}{(d+1)(d+2)} + \frac{\Delta^{\otimes 2}}{(d+1)^2}\right\rangle.
    \end{equation}
    Now, note that for any state $\varrho$, each term in the RHS of the inner product above has a non-negative inner product with both $(I+\varrho)^{\otimes 2}$ and $\varrho^{\otimes 2}$. Thus, using~\eqref{eq:closeness-uniform-povm-variance-1} again, we have
    \begin{align}
        \tr(\mbb{E}[\hat{\Delta}^{\otimes 2}]^2) &\leq \tr\left(\left(\frac{2\cdot \swap + ((\rho + \sigma) \otimes I + I \otimes (\rho + \sigma))\cdot \swap}{(d+1)(d+2)} + \frac{\Delta^{\otimes 2}}{(d+1)^2}\right)^2\right)
        \\&= \frac{\tr(\Delta^2)^2}{(d+1)^4} + \frac{4\tr(\Delta^2) + 4\tr(\Delta^2(\rho + \sigma))}{(d+1)^3(d+2)} + \frac{\tr((2\cdot \swap + ((\rho + \sigma) \otimes I + I \otimes (\rho + \sigma))\cdot \swap)^2)}{(d+1)^2(d+2)^2}
        \\&\leq \frac{\tr(\Delta^2)^2}{(d+1)^4} + \frac{12\tr(\Delta^2)}{(d+1)^3(d+2)} + \frac{4d^2 + 4d \cdot \tr(\rho + \sigma) + \tr(((\rho+\sigma) \otimes I + I \otimes (\rho + \sigma))^2)}{(d+1)^2(d+2)^2}
        \\&\leq \mathbb{E}[\bar{X}]^2 + \bigo\left(\frac{\tr(\Delta^2)}{d^4} + \frac{1}{d^2}\right), \label{eq:closeness-uniform-povm-variance-3}
    \end{align}
    where we used $\rho + \sigma \preceq 2\cdot I$ throughout. Now, using \Cref{lem:hs-estimation-variance-generic,eq:closeness-uniform-povm-variance-2,eq:closeness-uniform-povm-variance-3}, we get
    \begin{equation}
        \mbb{V}[\bar{X}] = \bigo\left(\frac{\tr(\Delta^2)}{nd^4} + \frac{\tr(\Delta^2)}{n^2d^4} + \frac{1}{n^2d^2}\right) = \bigo\left(\frac{\tr(\Delta^2)}{nd^4} + \frac{1}{n^2d^2}\right),
    \end{equation}
    as claimed.  
\end{proof}

\subsubsection{Proof of \Cref{thm:closeness-l2-estimation}}
\label{sec:closeness-HS-upper-proof}

Now, we prove \Cref{thm:closeness-l2-estimation} by instantiating \Cref{alg:closeness-l2-estimation} with the unbiased estimator in \Cref{alg:quasi-purification-estimation}. First, we state our upper bound on the variance of this algorithm.

\begin{lemma}
\label{lem:tcopy-closeness-testing-variance}
    Let $\bar{X}$ be the output of \Cref{alg:closeness-l2-estimation} instantiated with \Cref{alg:quasi-purification-estimation} when applied to unknown states $\rho, \sigma$. Let $\Delta = \rho - \sigma$. Then,
    \begin{equation}
        \mbb{V}[\bar{X}] = \bigo\left(
            \frac{\tr(\Delta^2)}{nt} + \frac{1}{n^2t^2} + \frac{d^2 \cdot \min\{d^2,t\}}{n^2t^4} + \frac{d \cdot \tr(\rho\sigma)}{n^2t^2} + \frac{d \cdot \tr(\Delta^2)}{n^2t^2}
        \right).
    \end{equation}
\end{lemma}

Before proving the above lemma, we will first use it to prove \Cref{thm:closeness-l2-estimation}.

\begin{proof}[Proof of \Cref{thm:closeness-l2-estimation}]
    We will use \Cref{alg:closeness-l2-estimation} instantiated with the unbiased estimator in \Cref{alg:quasi-purification-estimation} and receive output $\bar{X}$. By \Cref{lem:hs-estimation-variance-generic,}, we have $\mbb{E}[\bar{X}] = \|\rho - \sigma\|_2^2$ and so we will compare $\bar{X}$ to the threshold $\eps^2/2$. By \Cref{lem:tcopy-closeness-testing-variance},  
    \begin{equation}
        \mbb{V}[\bar{X}] = \bigo\left(
            \frac{\tr(\Delta^2)}{nt} + \frac{1}{n^2t^2} + \frac{d^2 \cdot \min\{d^2,t\}}{n^2t^4} + \frac{d \cdot \tr(\rho\sigma)}{n^2t^2} + \frac{d \cdot \tr(\Delta^2)}{n^2t^2}
        \right).        
    \end{equation}
    We will now divide our analysis into two cases, depending on the parameter $t$.

    \paragraph{Case 1 ($t \leq d^2$):} Here, the variance bound simplifies to
    \begin{equation}
        \mbb{V}[\bar{X}] = \bigo\left(
            \frac{\tr(\Delta^2)}{nt} + \frac{1}{n^2t^2} + \frac{d^2}{n^2t^3} + \frac{d \cdot \tr(\rho\sigma)}{n^2t^2} + \frac{d \cdot \tr(\Delta^2)}{n^2t^2}
        \right).
    \end{equation}
    When $\rho = \sigma$, $\tr(\Delta^2) = 0$. Then, using Chebyshev's inequality, we have
    \begin{equation}
        \rho = \sigma \implies \mrm{Pr}[\bar{X} \geq \eps^2/2] \leq \bigo\left(
            \frac{1}{n^2t^2\eps^4} + \frac{d^2}{n^2t^3\eps^4} + \frac{d\cdot \tr(\rho\sigma)}{n^2t^2\eps^4}
        \right).
    \end{equation}
    On the other hand, we have
    \begin{align}
        \|\rho - \sigma\|_2 \geq \eps \implies& \mrm{Pr}[\bar{X} \leq \eps^2/2] \leq \mrm{Pr}[\bar{X} \leq \tr(\Delta^2)/2]
        \\&\leq \bigo\left(
            \frac{1}{nt\cdot\tr(\Delta^2)} + \frac{1}{n^2t^2\cdot\tr(\Delta^2)^2} + \frac{d^2}{n^2t^3\cdot\tr(\Delta^2)^2} + \frac{d \cdot \tr(\rho\sigma)}{n^2t^2\cdot\tr(\Delta^2)^2} + \frac{d}{n^2t^2\cdot\tr(\Delta^2)}
        \right)
        \\&\leq \bigo\left(
            \frac{1}{nt\eps^2} + \frac{1}{n^2t^2\eps^4} + \frac{d^2}{n^2t^3\eps^4} + \frac{d \cdot \tr(\rho\sigma)}{n^2t^2\eps^4} + \frac{d}{n^2t^2\eps^2}
        \right).
    \end{align}
    Thus, in either case, the probability of error is at most
    \begin{equation}
        \bigo\left(
            \frac{1}{nt\eps^2} + \frac{1}{n^2t^2\eps^4} + \frac{d^2}{n^2t^3\eps^4} + \frac{d \cdot \tr(\rho\sigma)}{n^2t^2\eps^4} + \frac{d}{n^2t^2\eps^2}
        \right).
    \end{equation}
    Consequently, to attain success probability at least $2/3$, it suffices to take
    \begin{equation}
        nt \geq C \left(\frac{1}{\eps^2} + \frac{d}{\sqrt{t}\eps^2} + \frac{\sqrt{d \cdot \tr(\rho\sigma)}}{\eps^2} + \frac{\sqrt{d}}{\eps}\right) \label{eq:hs-closeness-thm-1}
    \end{equation}
    for some absolute constant $C > 0$.

    \paragraph{Case 2 ($t \geq d^2$):} Using a similar argument to the previous case, the overall error probability is at most
    \begin{equation}
        \bigo\left(\frac{1}{nt\eps^2} + \frac{1}{n^2t^2\eps^4} + \frac{d^4}{n^2t^4\eps^4} + \frac{d \cdot \tr(\rho\sigma)}{n^2t^2\eps^4} + \frac{d}{n^2t^2\eps^2}\right).
    \end{equation}
    Thus, for an absolute constant $C > 0$, it suffices to take
    \begin{equation}
        nt \geq C \left(
            \frac{1}{\eps^2} + \frac{d^2}{t\eps^2} + \frac{\sqrt{d \cdot \tr(\rho\sigma)}}{\eps^2} + \frac{\sqrt{d}}{\eps}
        \right).
    \end{equation}
    However, here $d^2 \leq t$, so the copy complexity is
    \begin{equation}
        \bigo\left(\frac{1}{\eps^2} + \frac{\sqrt{d \cdot \tr(\rho\sigma)}}{\eps^2} + \frac{\sqrt{d}}{\eps}\right). \label{eq:hs-closeness-thm-2}
    \end{equation}

    Combining both cases, i.e. \Cref{eq:hs-closeness-thm-1,eq:hs-closeness-thm-2}, the overall copy complexity is at most
    \begin{equation}
        nt = \bigo\left(\frac{1}{\eps^2} + \frac{d}{\sqrt{t}\eps^2} + \frac{\sqrt{d \cdot \tr(\rho\sigma)}}{\eps^2} + \frac{\sqrt{d}}{\eps}\right),
    \end{equation}
    as desired. To obtain success probability $1-\delta$, it suffices to perform $\bigo(\log(1/\delta))$ repetitions.
\end{proof}

Finally, we prove the variance bound used in the above proof.

\begin{proof}[Proof of \Cref{lem:tcopy-closeness-testing-variance}]

    First, note that $\hat{\Delta}^{\otimes 2} = (\hat{\rho} - \hat{\sigma})^{\otimes 2}$, so by \Cref{thm:quasi-purification-truncated-second-moment}, we have
    \begin{align}
        \mbb{E}[\hat{\Delta}^{\otimes 2}] &= \mbb{E}[\hat{\rho}^{\otimes 2}] + \mbb{E}[\hat{\sigma}^{\otimes 2}] - \mbb{E}[\hat{\rho} \otimes \hat{\sigma}] - \mbb{E}[\hat{\sigma} \otimes \hat{\rho}]
        \\&= \frac{t-1}{t}(\rho^{\otimes 2} + \sigma^{\otimes 2}) + \frac1t(\rho \otimes I + I \otimes \rho + \sigma \otimes I + I \otimes \sigma) \swap + \frac{\mbb{E}_{\rho}[\ell(\lambda)] + \mbb{E}_{\sigma}[\ell(\lambda)]}{t^2} \swap \nonumber\\&\hspace{1em} - \Lower_\rho - \Lower_\sigma - \rho \otimes \sigma - \sigma \otimes \rho
        \\&= \Delta^{\otimes 2} + \frac1t (\rho \otimes I + I \otimes \rho + \sigma \otimes I + I \otimes \sigma) \swap + \frac{L}{t^2}\swap - \frac{\rho^{\otimes 2} + \sigma^{\otimes 2}}{t} - \Lower_\rho - \Lower_\sigma, \label{eq:hs-closeness-testing-1}
    \end{align}
    where we define $L \triangleq \mbb{E}_{\rho}[\ell(\lambda)] + \mbb{E}_{\sigma}[\ell(\lambda)]$.

    Now, keeping in mind \Cref{lem:hs-estimation-variance-generic}, we first wish to upper bound $\tr(\mbb{E}[\hat{\Delta}^{\otimes 2}] \cdot \Delta^{\otimes 2})$. For any state $\varrho$, $\tr(\Delta^{\otimes 2}\varrho^{\otimes 2})$ and $\tr(\Delta^{\otimes 2}\Lower_\varrho)$ are non-negative (as $\Lower_\varrho \in \mathrm{SoS}(d)$), and so we can write
    \begin{align}
        \tr(\mbb{E}[\hat{\Delta}^{\otimes 2}] \cdot \Delta^{\otimes 2}) &\leq \left\langle \Delta^{\otimes 2} + \frac1t (\rho \otimes I + I \otimes \rho + \sigma \otimes I + I \otimes \sigma) \swap + \frac{L}{t^2}\swap , \Delta^{\otimes 2}\right\rangle
        \\&= \tr(\Delta^2)^2 + \frac{2}{t}\tr(\Delta^2(\rho + \sigma)) + \frac{L}{t^2}\tr(\Delta^2)
        \\&\leq \tr(\Delta^2)^2 + \frac{4}{t}\tr(\Delta^2) + \frac{4}{t^{3/2}}\tr(\Delta^2),
    \end{align}
    where we used $\rho + \sigma \preceq 2I$ and $\mbb{E}_\rho[\ell(\lambda)] \leq 2\sqrt{t}$ for any state $\rho$ (see \Cref{lem:partition-length-upper-bound}). Thus, we have 
    \begin{equation}
        \tr(\mbb{E}[\hat{\Delta}^{\otimes 2}] \Delta^{\otimes 2}) - \tr(\Delta^2)^2 \leq \bigo\left(\frac{\tr(\Delta^2)}{t}\right).
        \label{eq:hs-closeness-testing-2}
    \end{equation}
    Next, to use \Cref{lem:hs-estimation-variance-generic}, we wish to upper bound $\tr(\mbb{E}[\hat{\Delta}^{\otimes 2}]^2)$. First, note that
    \begin{align}
        \langle \mbb{E}[\hat{\Delta}^{\otimes 2}], \mbb{E}[\hat{\Delta}^{\otimes 2}] \rangle &\leq \left \langle \mbb{E}[\hat{\Delta}^{\otimes 2}], \mbb{E}[\hat{\Delta}^{\otimes 2}] + \frac{\rho^{\otimes 2} + \sigma^{\otimes 2}}{t} +\Lower_\rho + \Lower_\sigma \right\rangle
        \\&= \left\langle \mbb{E}[\hat{\Delta}^{\otimes 2}], \Delta^{\otimes 2} + \frac1t (\rho \otimes I + I \otimes \rho + \sigma \otimes I + I \otimes \sigma) \swap + \frac{L}{t^2}\swap \right\rangle.
    \end{align}
    In the inequality, we use the fact that for all states $\varrho,$ the quantities $\tr(\hat{\Delta}^{\otimes 2} \varrho^{\otimes 2})$ and $\tr(\hat{\Delta}^{\otimes 2} \Lower_\varrho)$ are non-negative, and this hold for $\mbb{E}[\hat{\Delta}^{\otimes 2}]$ by linearity of expectation. Further, note that the negative terms can also be neglected in the expansion of the other $\mbb{E}[\hat{\Delta}^{\otimes 2}]$ term as their inner products with $\Delta^{\otimes 2}, (\rho \otimes I + I \otimes \rho)\swap$ and $\swap$ are all non-positive. Thus, we have
    \begin{align}
        \langle \mbb{E}[\hat{\Delta}^{\otimes 2}], \mbb{E}[\hat{\Delta}^{\otimes 2}] \rangle &\leq \tr\left(\left(\Delta^{\otimes 2} + \frac1t (\rho \otimes I + I \otimes \rho + \sigma \otimes I + I \otimes \sigma) \swap + \frac{L}{t^2}\swap\right)^2\right)
        \\&= \tr(\Delta^2)^2 + \frac{1}{t^2} \tr((\rho \otimes I + I \otimes \rho + \sigma \otimes I + I \otimes \sigma)^2) + \frac{d^2L^2}{t^4} \nonumber\\&\hspace{1em} + \frac{4}{t}(\tr(\Delta^2 \rho) + \tr(\Delta^2 \sigma)) + \frac{2L}{t^2}\tr(\Delta^2) + \frac{8dL}{t^3}
        \\&= \tr(\Delta^2)^2 + \frac{2d}{t^2}(\tr(\rho^2) + \tr(\sigma^2) + 2\tr(\rho\sigma)) + \frac{8}{t^2} + \frac{d^2L^2}{t^4} \nonumber\\&\hspace{1em} + \frac{4}{t}\tr(\Delta^2(\rho + \sigma)) + \frac{2L}{t^2}\tr(\Delta^2) + \frac{8dL}{t^3}
        \\&\leq \tr(\Delta^2)^2 + \frac{2d}{t^2}\tr(\Delta^2) + \frac{8d}{t^2}\tr(\rho\sigma) + \frac{8}{t^2} + \frac{d^2L^2}{t^4} + \frac{8}{t}\tr(\Delta^2) + \frac{2L}{t^2} \tr(\Delta^2) + \frac{8dL}{t^3}
        \\&= \tr(\Delta^2)^2 + \tr(\Delta^2)\left(\frac{2d}{t^2} + \frac{8}{t} + \frac{2L}{t^2}\right) + \left(\frac{8}{t^2} + \frac{8dL}{t^3} + \frac{d^2L^2}{t^4}\right) + \frac{8d \cdot \tr(\rho\sigma)}{t^2}
        \\&= \tr(\Delta^2)^2 + \tr(\Delta^2) \cdot \bigo\left(\frac{d}{t^2} + \frac{1}{t} + \frac{1}{t^{3/2}}\right) + \bigo\left(\frac{1}{t^2} + \frac{d^2L^2}{t^4}\right) + \bigo\left(\frac{d \cdot \tr(\rho\sigma)}{t^2}\right)
        \\&= \tr(\Delta^2)^2 + \bigo\left(
        \frac{d \cdot \tr(\Delta^2)}{t^2} + \frac{\tr(\Delta^2)}{t} + \frac{1}{t^2} + \frac{d^2 \cdot \min\{d^2, t\}}{t^4} + \frac{d \cdot \tr(\rho\sigma)}{t^2}
        \right), \label{eq:hs-closeness-testing-3}
    \end{align}
    where we used $\tr(\Delta^2) = \tr(\rho^2) + \tr(\sigma^2) - 2\tr(\rho\sigma)$ and $\rho + \sigma \preceq 2I$ in the third step, $L \leq 4\sqrt{t}$ in the penultimate step, and the stronger bound $L \leq \min\{2d,4\sqrt{t}\}$ in the last step (again, via \Cref{lem:partition-length-upper-bound}). Now, \Cref{lem:hs-estimation-variance-generic} and \Cref{eq:hs-closeness-testing-2,eq:hs-closeness-testing-3} yield the desired bound.
\end{proof}

\subsection{Instance-dependent bound via bucketing}
\label{sec:upper-instance-dependent}
In this section, we prove the instance-dependent upper bound of \Cref{thm:instance-dependent-tcopy-certification}. As done before by \cite{chen2022toward,o2025instance} in the settings of unentangled and fully entangled measurements, we will use bucketing to reduce this problem to worst-case Hilbert--Schmidt certification on lower-dimensional states, and then apply our \Cref{thm:closeness-l2-estimation}. First, we state some necessary notation for bucketing.

\begin{definition}[Bucketing, verbatim from \cite{chen2022toward,o2025instance}]
    Assume $\sigma = \diag(\lambda_1, \dots, \lambda_d) \in \mbb{C}^{d \times d}$ is a quantum state\footnote{The assumption that $\sigma$ is diagonal is without loss of generality, as the tester knows the description of $\sigma$ and can change basis without consuming any additional copies.}, and assume $\lambda_1 \geq \dots \geq \lambda_d \geq 0$.
    \begin{itemize}
        \item Let $\stail = \{d^\prime + 1, \dots, d\}$, where $d^{\prime}$ is the smallest index such that $\sum_{j \in \stail} \lambda_j \leq \eps^2/20.$ Let $\sigma^* = \diag(\lambda_1, \dots, \lambda_{d^\prime}, 0, \dots, 0)$, i.e., the operator obtained by zeroing out the eigenvalues of $\sigma$ in $\stail$.
        \item Divide $\{1, \dots, d^{\prime}\}$ into buckets $\{S_j\}_j$, where $S_j = \{i \in [d^\prime]: 2^{-i-1} \leq \lambda_i  \leq 2^{-i}\}$. Let the set of non-empty buckets be $\mc{J}$; by \cite[Fact 6.6]{chen2022toward}, the number of such buckets is $m \triangleq |\mc{J}| \leq \bigo(\log(d/\eps))$.
        \item For each $j \in \mc{J}$, let $\Pi_j$ be the associated projector onto the indices in $S_j$. Let $\sigma_j = \Pi_j\sigma\Pi_j$, let $\sigma_{j,j^\prime} = (\Pi_j + \Pi_{j^\prime}) \sigma (\Pi_j + \Pi_{j^\prime})$, and similarly define $\rho_j, \rho_{j,j^\prime}$. For any matrix $M$, let $\hat{M} = M/\tr(M)$.
    \end{itemize}
\end{definition}

Having stated the necessary notation, we can now recall the key reduction to low-dimensional Hilbert--Schmidt testing:

\begin{lemma}[Worst-case certification to local Hilbert--Schmidt testing, adapted from {\cite[Lemmas 6.6-6.8]{o2025instance}}]
\label{lem:bucketing-reduction}
    Given the description of $\sigma$ and copies of $\rho$, to test whether $\rho = \sigma$ or $\|\rho - \sigma\|_1 \geq \eps$ with confidence $1-\delta$, it suffices to perform the following tests in order:
    \begin{enumerate}
        \item A pre-processing step that performs fixed, unentangled measurements on $\bigo(\log(1/\delta)/\eps^2)$ copies of $\rho$ and behaves as desired with probability at least $\delta/3$. \label{item:tail-mass}
        \item For each bucket $j \in \mc{J}$: \label{item:diagonal-tests}
        \begin{enumerate}
            \item Perform a pre-processing step that applies fixed, unentangled measurements to $\tilde{\bigo}(\log(1/\delta)/\eps^2)$ copies of $\rho$ and behaves as desired with probability at least $ \delta/6m$. If the test passes, we are confident that $\tr(\rho_j)$ and $\tr(\sigma_j)$ are within constant factors of each other. \label{item:diagonal-preprocessing}
            \item Test whether $\hat{\rho}_j = \hat{\sigma}_j$ or $\|\hat{\rho}_j - \hat{\sigma}_j\|_2 \geq \tilde{\Omega}\left(\frac{\eps}{d_j^{3/2}2^{-j}}\right)$ with confidence $1-\delta/6m$. \label{item:diagonal-main}
        \end{enumerate}
        \item For each pair of distinct buckets $j,j^\prime$, with $d_j \geq d_{j^\prime}$: \label{item:off-diagonal-tests}
        \begin{enumerate}
            \item \label{item:off-diagonal-preprocessing} Perform a pre-processing step that applies fixed, unentangled measurements to $\tilde{\bigo}(\log(1/\delta)/\eps^2)$ copies of $\rho$ and behaves as desired with probability at least $ \delta/6m^2$. If the test passes, we are confident that $\tr(\rho_{j,j^\prime})$ and $\tr(\sigma_{j,j^\prime})$ are within constant factors of each other.
            \item \label{item:off-diagonal-main} Test whether $\hat{\rho}_{j,j^\prime} = \hat{\sigma}_{j,j^\prime}$ or $\|\hat{\rho}_{j,j^\prime} - \hat{\sigma}_{j,j^\prime}\|_2 \geq \tilde{\Omega}\left(\frac{\eps}{\sqrt{d_{j^\prime}}(d_j2^{-j} + d_{j^\prime}2^{-j^\prime})}\right)$ with confidence $1-\delta/6m^2$,
        \end{enumerate}
    \end{enumerate}
    and finally, accept if all tests pass and reject otherwise.
\end{lemma}

In the above lemma, we leave the pre-processing steps in \Cref{item:tail-mass,item:diagonal-preprocessing,item:off-diagonal-preprocessing} unspecified, as these only require fixed single-copy measurements and can be carried out in the same manner as in \cite{chen2022toward,o2025instance}. The only difference in our proof of \Cref{thm:instance-dependent-tcopy-certification} and the upper bounds of \cite{chen2022toward,o2025instance} is the way in which the tests in \Cref{item:diagonal-main,item:off-diagonal-main} are carried out, which we present in the proof below. 

\begin{proof}[Proof of \Cref{thm:instance-dependent-tcopy-certification}]
    To show correctness, we use a union bound and sum over the failure probabilities in each step of \Cref{lem:bucketing-reduction}, which yields a total failure probability at most $\delta$. Now, it suffices to compute the overall copy complexity of the tests in \Cref{lem:bucketing-reduction}. First, recall that the complexity of the tests in \Cref{item:tail-mass,item:diagonal-preprocessing,item:off-diagonal-preprocessing} is at most $\tilde{O}(\log(1/\delta)/\eps^2)$.\footnote{This holds even after summing over all buckets/pairs of buckets, as $m = \bigo(\log(d/\eps))$ and gets suppressed by the $\tilde{\bigo}$.}

    \paragraph{Copy complexity of \Cref{item:diagonal-main}:} As a direct consequence of bucketing, for each $j$, we have $\|\hat{\sigma_j}\|_\infty \leq \frac{2}{d_j}$. Thus, we can use the algorithm of \Cref{cor:balanced-l2-closeness-testing} for each test. Note that we know the $d_j$-dimensional subspace that $\hat\rho_j,\hat\sigma_j$ are supported in, and can thus treat them as $d_j$-dimensional states. So, by \Cref{cor:balanced-l2-closeness-testing}, for the $j$th bucket, it suffices to take
    \begin{equation}
        N_j = \tilde{O}\left(\max\left\{\frac{d_j^3 2^{-2j}}{\eps^2}, \frac{d_j^4 2^{-2j}}{\sqrt{t}\eps^2}, \frac{d_j^2 2^{-j}}{\eps}\right\}\right)
    \end{equation}
    copies of $\hat{\rho}_j$ to succeed with $1-O(1)$ probability. To obtain these copies, we can measure copies of $\rho$ with the two-outcome measurement $\{\Pi_j, I - \Pi_j\}$, and only retain states that yield the former outcome. This can be done with $1-O(1)$ probability using
    \begin{equation}
        \bigo\left(\frac{N_j}{\tr(\rho_j)}\right) = \bigo\left(\frac{N_j}{d_j 2^{-j}}\right) = \tilde{O}\left(\max\left\{\frac{d_j^2 2^{-j}}{\eps^2}, \frac{d_j^3 2^{-j}}{\sqrt{t}\eps^2}, \frac{d_j}{\eps}\right\}\right)
    \end{equation}
    copies of $\rho$. The first equality above used the fact that $\tr(\sigma_j)$ and $\tr(\rho_j)$ are within constant factors of each other. Now, the overall success probability of the test associated with bucket $j$ can be boosted to $1-\delta/6m$ using $\bigo(\log(m/\delta))$ repetitions. Next, we will sum over the complexities associated with each bucket. For the first term, 
    \begin{equation}
        \sum_{j \in \mc{J}} d_j^2 2^{-j} \leq (\sum_{j \in \mc{J}} d_j 2^{-j/2})^2 \leq 2\|\sigma^*\|_{1/2}.
    \end{equation}
    Next, for the second term,
    \begin{equation}
        \sum_{j \in \mc{J}} d_j^3 2^{-j} \leq (\sum_{j \in \mc{J}} d_j 2^{-j/3})^3 \leq 2\|\sigma^*\|_{1/3}.
    \end{equation}
    Finally, $\sum_j d_j = \rank(\sigma^*)$. Thus, the overall complexity of the tests in \Cref{item:diagonal-main} is at most
    \begin{equation}
        \tilde{\bigo}\left(\max\left\{
            \frac{\|\sigma^*\|_{1/2}}{\eps^2}, \frac{\|\sigma^*\|_{1/3}}{\sqrt{t}\eps^2}, \frac{\rank(\sigma^*)}{\eps}
        \right\}\right) \cdot \log(1/\delta). \label{eq:instance-dependent-upper-diagonal}
    \end{equation}

    \paragraph{Copy complexity of \Cref{item:off-diagonal-main}:} To apply \Cref{cor:balanced-l2-closeness-testing} in this case, we need to show that $\|\hat{\sigma}_{j,j^\prime}\|_{\infty} \leq C \cdot (d_j + d_{j^\prime})^{-1}$; however, this need not be true in general. We will instead show a weaker upper bound on $\tr(\hat{\sigma}_{j,j^\prime} \hat{\rho}_{j,j^\prime})$ and then apply the more general \Cref{thm:closeness-l2-estimation}. We have the following operator norm bound:
    \begin{equation}
        \|\hat{\sigma}_{j,j^\prime}\|_{\infty} \leq \frac{\max\{2^{-j},2^{-j^\prime}\}}{\tr(\sigma_{j,j^\prime})} \leq \frac{2 \cdot \max\{2^{-j},2^{-j^\prime}\}}{d_j2^{-j} + d_j2^{-j^\prime}} \leq \max\left\{\frac{2}{d_j},\frac{2}{d_{j^\prime}}\right\} = \frac{2}{d_{j^\prime}}.
    \end{equation}
    Thus, we can use \Cref{thm:closeness-l2-estimation} with $\tr(\hat{\sigma}_{j,j^\prime} \hat{\rho}_{j,j^\prime}) \leq \frac{2}{d_{j^\prime}}$, which yields the following upper bound on the number of copies of $\hat{\rho}_{j,j^\prime}$ for the test in \Cref{item:off-diagonal-main}:
    \begin{equation}
        N_{j,j^\prime} \leq \Tilde{\bigo}\left(\max\left\{
            \frac{d_{j^\prime}p_{j,j^\prime}^2}{\eps^2}, \frac{d_jd_{j^\prime}p_{j,j^\prime}^2}{\sqrt{t}\eps^2}, \frac{\sqrt{d_j d_{j^\prime}} p_{j,j^\prime}^2}{\eps^2}, \frac{\sqrt{d_j d_{j^\prime}} p_{j,j^\prime}}{\eps}
        \right\}\right),
    \end{equation}
    where we bound the dimension by $d_j + d_{j^\prime} \leq 2d_j$ and define $p_{j,j^\prime} = d_j2^{-j} + d_{j^\prime}2^{-j^\prime}$. We can obtain these copies by performing measurements $\{\Pi_j + \Pi_{j^\prime}, I - \Pi_j - \Pi_{j^\prime}\}$ on $\bigo(N_{j,j^\prime}/\tr(\rho_{j,j^\prime})) = \bigo(N_{j,j^\prime}/p_{j,j^\prime})$ copies of $\rho$. Thus, the copy complexity for the pair $j,j^\prime$ is at most 
    \begin{equation}
        \Tilde{\bigo}\left(\max\left\{
            \frac{d_{j^\prime}p_{j,j^\prime}}{\eps^2}, \frac{d_jd_{j^\prime}p_{j,j^\prime}}{\sqrt{t}\eps^2}, \frac{\sqrt{d_j d_{j^\prime}} p_{j,j^\prime}}{\eps^2}, \frac{\sqrt{d_j d_{j^\prime}}}{\eps}
        \right\}\right).
    \end{equation}
    We can obtain the desired confidence by performing $\tilde{\bigo}(\log(m/\delta))$ repetitions as before. We will now sum over all pairs $j,j^\prime$. For the first term,
    \begin{equation}
        \sum_{j,j^\prime \in \mc{J}} d_{j^\prime}p_{j,j^\prime} = \sum_{j,j^\prime} (d_{j^\prime}d_j 2^{-j} + d_{j^\prime}^2 2^{-j^\prime}) \leq 2m \sum_j d_j^2 2^{-j} \leq 4m\|\sigma^*\|_{1/2}. 
    \end{equation}
    Next, for the second term,
    \begin{equation}
        \sum_{j,j^\prime \in \mc{J}} d_j d_{j^\prime}p_{j,j^\prime} = \sum_{j,j^\prime} (d_{j^\prime}d_j^2 2^{-j} + d_jd_{j^\prime}^2 2^{-j^\prime}) \leq m\sum_j d_j^3 2^{-j} + \sum_j d_j \sum_{j^\prime} d_{j^\prime}^2 2^{-j^\prime} \leq 2m\|\sigma^*\|_{1/3} + \rank(\sigma^*) \cdot \|\sigma^*\|_{1/2}.
    \end{equation}
    For the third term, we have 
    \begin{equation}
        \sum_{j,j^\prime} \sqrt{d_j d_{j^\prime}} p_{j,j^\prime} \leq \sum_{j,j^\prime} d_j p_{j,j^\prime} \leq \sum_{j,j^\prime} d_j^2 2^{-j} + \sum_j d_j \sum_{j^\prime} d_{j^\prime} 2^{-j^\prime} \leq 2m\|\sigma^*\|_{1/2} + \rank(\sigma^*).
    \end{equation}
    Finally, $\sum_{j,j^\prime} \sqrt{d_j d_{j^\prime}} \leq \sum_{j,j^\prime} d_j \leq m \cdot \rank(\sigma^*)$. Recalling that $m = \bigo(\log(d/\eps))$ and summing over all terms, the overall copy complexity of the tests in \Cref{item:off-diagonal-main} is thus
    \begin{align}
        \tilde{\bigo}\left(\max\left\{
            \frac{\|\sigma^*\|_{1/2}}{\eps^2}, \frac{\|\sigma^*\|_{1/3} + \rank(\sigma^*) \|\sigma^*\|_{1/2}}{\sqrt{t}\eps^2}, \frac{\|\sigma^*\|_{1/2} + \rank(\sigma^*)}{\eps^2}, \frac{\rank(\sigma^*)}{\eps}
        \right\}\cdot \log(1/\delta)\right)
        \\=
        \tilde{\bigo}\left(\max\left\{
            \frac{\|\sigma^*\|_{1/2}}{\eps^2}, \frac{\|\sigma^*\|_{1/3}}{\sqrt{t}\eps^2}, \frac{\rank(\sigma^*) \|\sigma^*\|_{1/2}}{\sqrt{t}\eps^2}, \frac{\rank(\sigma^*)}{\eps^2}
        \right\} \cdot \log(1/\delta)\right)
        \label{eq:instance-dependent-upper-offdiagonal}
    \end{align}
    \paragraph{Putting things together:} Summing over the copy complexities of the pre-processing steps and those in \Cref{eq:instance-dependent-upper-diagonal,eq:instance-dependent-upper-offdiagonal}, we obtain the stated upper bound.
\end{proof}

\section{Mixedness testing and purity estimation upper bounds}
\label{sec:mixedness-and-purity-upper}
In this section, we prove our upper bound for multiplicative-error purity estimation (\Cref{thm:intro-purity-upper}) and provide an alternate proof of our trace-norm mixedness testing upper bound (\Cref{thm:intro-mixedness-upper}). We first present a general algorithm for purity estimation that repeatedly applies a quantum state estimator and then performs some classical post-processing. This algorithm can be used for mixedness testing by comparing the estimated purity to an appropriate threshold. We present the overall algorithm and analyze its first two moments below. 

\begin{algorithm}[H]
    \begin{algorithmic}[1]
    \caption{Purity($\mc{A}$), given an algorithm $\mathcal{A}$ that takes as input $\rho^{\otimes t}$ and returns a state $\hat{\rho} \in \mbb{C}^{d \times d}$.}
    \label{alg:collision-estimation}
    \State \textbf{Input:} $n \cdot t$ copies of an unknown state $\rho$
    \For{$i \in [n]$}
        \State $\hat{\rho}_i \gets \mc{A}(\rho^{\otimes t})$
    \EndFor
    \For{$i,j \in [n]$ with $i < j$}
        \State $X_{i,j} \gets \tr(\hat{\rho}_i\hat{\rho}_j)$
    \EndFor
    \Return $\bar{X} \gets \frac{1}{\binom{n}{2}} \sum_{i < j} X_{i,j}$
    \end{algorithmic}
\end{algorithm}

The following lemma characterizes the first and second moments of \Cref{alg:collision-estimation} in terms of the first and second moments of the underlying state estimation algorithm.
\begin{lemma}
 \label{lem:collision-variance-generic}
    Let $\bar{X}$ be the output of $\mathrm{Collision}(\mc{A})$ (\Cref{alg:collision-estimation}). Then, $\mbb{E}[\bar{X}] = \tr(\mbb{E}[\hat{\rho}]^2)$ and
    \begin{equation}
    \label{eq:collision-variance-generic-1}
          \mbb{V}[\bar{X}] = \bigo(n^{-1}) \cdot (\tr(\mbb{E}[\hat{\rho}^{\otimes 2}] \cdot \mbb{E}[\hat{\rho}]^{\otimes 2}) - \mbb{E}[\bar{X}]^2) + \bigo(n^{-2}) \cdot (\tr(\mbb{E}[\hat{\rho}^{\otimes 2}]^2) - \mbb{E}[\bar{X}]^2).
    \end{equation}
    In particular, when $\mc{A}$ is unbiased, i.e., $\mbb{E}[\hat{\rho}] = \rho$, we have $\mbb{E}[\bar{X}] = \tr(\rho^2)$ and
    \begin{equation}
    \label{eq:collision-variance-generic-2}
        \mbb{V}[\bar{X}] = \bigo(n^{-1}) \cdot (\tr(\mbb{E}[\hat{\rho}^{\otimes 2}] \cdot \rho^{\otimes 2}) - \tr(\rho^2)^2) + \bigo(n^{-2}) \cdot (\tr(\mbb{E}[\hat{\rho}^{\otimes 2}]^2) - \tr(\rho^2)^2).
    \end{equation}
\end{lemma}

\begin{proof} 
    We omit the proof as it is essentially identical to that of \Cref{lem:hs-estimation-variance-generic}.
\end{proof}

We will now apply the above algorithm to mixedness testing. In \Cref{sec:mixedness-upper-t-copy}, we instantiate the algorithm with \Cref{alg:quasi-purification-estimation} and prove \Cref{thm:intro-mixedness-upper} by explicitly working out the variance. Then, we use this variance bound to prove \Cref{thm:intro-purity-upper} in \Cref{sec:purity-estimation-t-copy}. 

\subsection{Mixedness testing with $t$-copy measurements}
\label{sec:mixedness-upper-t-copy}

In this section, we will prove the following weaker version of \Cref{thm:intro-mixedness-upper}:
\begin{theorem}
 \label{thm:dtr-mixedness-upper-t-copy}
     Let $d \geq 2, \eps > 0$. There exists an algorithm to test whether $\rho = \mmstate$ or $\|\rho - \mmstate\|_1 \geq \eps$ using fixed, $t$-copy measurements on $\bigo\left(\max\left\{\frac{d^2}{\sqrt{t}\eps^2}, \frac{\sqrt{t}}{\eps^2}\right\}\right)$ copies. 
\end{theorem}

\begin{remark}
\label{rmk:fewer-copies-suffice}
     When $t \leq d^2$, \Cref{thm:dtr-mixedness-upper-t-copy} yields an $\bigo(d^2/\sqrt{t}\eps^2)$ bound. However, for $t \geq d^2$, the bound is $O(\sqrt{t}/\eps^2)$, which grows with $t$. We believe this is an artifact of our analysis, and the same algorithm should achieve the $\bigo(\max\{d^2/\sqrt{t}\eps^2, d/\eps^2\})$ rate stated in \Cref{thm:intro-mixedness-upper}. However, a simple modification immediately corrects this issue. In the regime of $t \geq d^2$, instead of performing $t$-copy measurements, the tester can perform these measurements on only $d^2$ copies at a time, obtaining an $\bigo(d/\eps^2)$ copy complexity. This precisely recovers the tradeoff of \Cref{thm:intro-mixedness-upper} as desired. Thus, it suffices to prove \Cref{thm:dtr-mixedness-upper-t-copy}. 
\end{remark}

Our proof of \Cref{thm:dtr-mixedness-upper-t-copy} will make use of the following bounds on quantities appearing in the variance of our estimator:

\begin{lemma}
\label{lem:quasi-purification-variance-case-2}
    Let $\hat{\rho}_\lambda$ be as in in \Cref{alg:quasi-purification-estimation}. Let $\Delta = \rho - \mmstate$. Then,
    \begin{equation}
        \mbb{E}[\tr(\hat{\rho}_\lambda^{\otimes 2} \rho^{\otimes 2})] \leq \tr(\rho^2)^2 + \bigo\left(\frac{\tr(\Delta^2)}{dt} + \frac{ \mbb{E}[\ell(\lambda)]}{t^2}\tr(\Delta^2) + \frac{\tr(\Delta^2)^{3/2}}{t}\right).
    \end{equation}
\end{lemma}

\begin{lemma}
\label{lem:quasi-purification-variance-case-3}
    Let $\hat{\rho}_\lambda$ be as in in \Cref{alg:quasi-purification-estimation}. Let $\Delta = \rho - \mmstate$. Then,
    \begin{align}
        \tr(\mbb{E}[\hat{\rho}_\lambda^{\otimes 2}]^2) &\leq \tr(\rho^2)^2 + \bigo\left(\frac{1}{d^2t} + \frac{1}{t^2} + \frac{d \cdot \mbb{E}[\ell(\lambda)]}{t^3} + \frac{\mbb{E}[\ell(\lambda)]}{dt^2} + \frac{\mbb{E}[\ell(\lambda)]^2 \cdot d^2}{t^4}\right) \nonumber\\& \hspace{1em}+ \bigo\left(\frac{1}{dt} + \frac{d}{t^2} + \frac{\mbb{E}[\ell(\lambda)]}{t^2}\right)\tr(\Delta^2) + \bigo\left(\frac{\tr(\Delta^2)^{3/2}}{t}\right).
    \end{align}
\end{lemma}

We will prove the above bounds on the variance quantities in \Cref{sec:appendix-variance-bounds}. Now, using these bounds, we prove \Cref{thm:dtr-mixedness-upper-t-copy}.

\begin{proof}[Proof of \Cref{thm:dtr-mixedness-upper-t-copy}]
    Our tester uses \Cref{alg:collision-estimation} and instantiates it with the unbiased state estimator of \Cref{alg:quasi-purification-estimation}. We will first bound the complexity of Hilbert--Schmidt testing with this algorithm, and then convert it to trace distance. Now, our purity estimator is unbiased, i.e., $\mbb{E}[\bar{X}] = \tr(\rho^2)$. Let $\hat{\rho}_\lambda$ denote the output of \Cref{alg:quasi-purification-estimation}. Using \Cref{lem:collision-variance-generic}, we have
    \begin{equation}
        \mbb{V}[\bar{X}] = \bigo\left(\frac{\tr(\mbb{E}[\hat{\rho}_\lambda^{\otimes 2}] \rho^{\otimes 2}) - \tr(\rho^2)^2}{n} + \frac{\tr(\mbb{E}[\hat{\rho}_\lambda^{\otimes 2}]^2) - \tr(\rho^2)^2}{n^2}\right).
    \end{equation}
    Our bounds on these quantities in \Cref{lem:quasi-purification-variance-case-2,lem:quasi-purification-variance-case-3} are in terms of $\mbb{E}[\ell(\lambda)]$. As we have $\mbb{E}[\ell(\lambda)] \leq \min\{2\sqrt{t},d\}$ (recall \Cref{lem:partition-length-upper-bound}), we split into two cases based on the value of $t$:
    \paragraph{Case 1:} When $t \leq \bigo(d^2)$, we obtain the following bound on the variance from \Cref{lem:quasi-purification-variance-case-2,lem:quasi-purification-variance-case-3}:
    \begin{equation}
        \mbb{V}[\bar{X}] \leq \bigo\left(\frac{d^2}{n^2t^3} + \frac{\tr(\Delta^2)^{3/2}}{nt} + \frac{\tr(\Delta^2)}{nt^{3/2}} + \frac{d \cdot \tr(\Delta^2)}{n^2t^2}\right).
    \end{equation}
    For Hilbert--Schmidt testing, we want to distinguish between the cases $\rho = \mmstate$ and $\|\rho - \mmstate\|_2 = \|\Delta\|_2 \geq \eps$. In the first case, we have
    \begin{equation}
        \mbb{V}[\bar{X}] \leq \bigo\left(\frac{d^2}{n^2t^3}\right) \implies \mrm{Pr}[\bar{X} \geq \eps^2/2] \leq \bigo\left(\frac{d^2}{n^2t^3\eps^4}\right).
    \end{equation}
    On the other hand, when $\|\Delta\|_2 \geq \eps$, we have
    \begin{align}
        \mrm{Pr}[\bar{X} \leq \eps^2/2] &\leq \mrm{Pr}[\bar{X} \leq \tr(\Delta^2)/2] \leq  \bigo\left(\frac{d^2}{n^2t^3\tr(\Delta^2)^2} + \frac{1}{nt\cdot\tr(\Delta^2)^{1/2}} + \frac{1}{nt^{3/2}\tr(\Delta^2)} + \frac{d}{n^2t^2\tr(\Delta^2)}\right)
        \\&\leq \bigo\left(\frac{d^2}{n^2t^3\eps^4} + \frac{1}{nt\eps} + \frac{1}{nt^{3/2}\eps^2} + \frac{d}{n^2t^2\eps^2}\right),
    \end{align}
    as $\tr(\Delta^2) \geq \eps^2$. Thus, in either case, the probability of error is at most
    \begin{equation}
        \bigo\left(\frac{d^2}{n^2t^3\eps^4} + \frac{1}{nt\eps} + \frac{1}{nt^{3/2}\eps^2} + \frac{d}{n^2t^2\eps^2}\right).
    \end{equation}
    Consequently, it suffices to take
    \begin{equation}
        nt = \bigo\left(\frac{d}{\sqrt{t}\eps^2} + \frac{1}{\eps} + \frac{1}{\sqrt{t}\eps^2} + \frac{\sqrt{d}}{\eps}\right) = \bigo\left(\frac{d}{\sqrt{t}\eps^2} + \frac{\sqrt{d}}{\eps}\right)
    \end{equation}
    to test whether $\rho = \mmstate$ or $\|\rho - \mmstate\|_2 \geq \eps$ with probability at least $2/3$. By \Cref{fact:hs-to-trace-norm-testing}, the copy complexity for trace-distance testing is thus
    \begin{equation}
        \bigo\left(\frac{d^2}{\sqrt{t}\eps^2} + \frac{d}{\eps}\right) = \bigo\left(\frac{d^2}{\sqrt{t}\eps^2}\right),
    \end{equation}
    as $t \leq \bigo(d^2)$.

    \paragraph{Case 2:} When $t \geq \Omega(d^2)$, \Cref{lem:quasi-purification-variance-case-2,lem:quasi-purification-variance-case-3} now imply the following bound on the variance:
    \begin{equation}
        \mbb{V}[\bar{X}] \leq \bigo\left(\frac{1}{n^2d^2t} + \frac{\tr(\Delta^2)^{3/2}}{nt} + \frac{\tr(\Delta^2)}{ndt}\right).
    \end{equation}
    Using a similar argument as in the previous case, the probability of error for Hilbert--Schmidt testing is at most
    \begin{equation}
        \bigo\left(\frac{1}{n^2d^2t\eps^4} + \frac{1}{nt\eps} + \frac{1}{ndt\eps^2}\right).
    \end{equation}
    Consequently, it suffices to take
    \begin{equation}
        nt = \bigo\left(\frac{\sqrt{t}}{d\eps^2} + \frac{1}{\eps} + \frac{1}{d\eps^2}\right) = \bigo\left(\frac{\sqrt{t}}{d\eps^2}\right),
    \end{equation}
    when $t = \Omega(d^2)$. By \Cref{fact:hs-to-trace-norm-testing}, the copy complexity for trace-distance testing is thus $\bigo(\frac{\sqrt{t}}{\eps^2})$.

    Our bounds in the two cases trade off precisely at the threshold $t = d^2$, implying the theorem statement.
\end{proof}

\subsection{Purity estimation with $t$-copy measurements}
\label{sec:purity-estimation-t-copy}

Using our variance bounds from \Cref{sec:mixedness-upper-t-copy}, we can also prove our purity estimation upper bound (\Cref{thm:intro-purity-upper}).

\begin{proof}[Proof of \Cref{thm:intro-purity-upper}]
    We will again use \Cref{alg:collision-estimation} instantiated with \Cref{alg:quasi-purification-estimation}. This is already an unbiased estimator for the purity, so it suffices to bound its variance. We will use \Cref{lem:collision-variance-generic,lem:quasi-purification-variance-case-2,lem:quasi-purification-variance-case-3}, and again split into two cases based on the value of $t$.

    \paragraph{Case 1 $(t \leq d^2)$:} 
    In this case, the variance is upper bounded by
    \begin{equation}
        \mbb{V}[\bar{X}] \leq \bigo\left(\frac{d^2}{n^2t^3} + \frac{\tr(\Delta^2)^{3/2}}{nt} + \frac{\tr(\Delta^2)}{nt^{3/2}} + \frac{d \cdot\tr(\Delta^2)}{n^2t^2} \right)
        \leq \bigo\left(\frac{d^2}{n^2t^3} + \frac{\tr(\rho^2)^{3/2}}{nt} + \frac{\tr(\rho^2)}{nt^{3/2}} + \frac{d \cdot\tr(\rho^2)}{n^2t^2} \right),
    \end{equation}
    where we use the fact that $\tr(\Delta^2) \leq \tr(\rho^2)$. Thus, by Chebyshev's inequality,
    \begin{align}
        \mathrm{Pr}[|\bar{X} - \mbb{E}[\bar{X}]| > \eps \cdot \tr(\rho^2)] &\leq \frac{\mbb{V}[\bar{X}]}{\eps^2 \cdot \tr(\rho^2)^2}
        \\& \leq \bigo\left(\frac{d^2}{n^2t^3\eps^2 \tr(\rho^2)^2} + \frac{1}{nt\eps^2 \tr(\rho^2)^{1/2}} + \frac{1}{nt^{3/2}\eps^2 \tr(\rho^2)} + \frac{d} {n^2t^2\eps^2 \tr(\rho^2)} \right)
        \\&\leq \bigo\left(\frac{d^4}{n^2t^3\eps^2} + \frac{\sqrt{d}}{nt\eps^2} + \frac{d}{nt^{3/2}\eps^2} + \frac{d^2}{n^2t^2\eps^2}\right),
    \end{align}
    where in the last step we use $\tr(\rho^2) \geq 1/d$ for all states $\rho$.
    Thus, for the error probability to be at most $1/3$, it suffices to take
    \begin{equation}
        nt \geq C \cdot \max\left\{\frac{d^2}{\sqrt{t}\eps},\frac{\sqrt{d}}{\eps^2},\frac{d}{\sqrt{t}\eps^2},\frac{d}{\eps}\right\},
    \end{equation}
    where $C > 0$ is a sufficiently large constant. This proves the theorem in the regime $t \leq d^2$.

    \paragraph{Case 2 $(t \geq d^2)$:}
    Note that for $t \geq d^2$, we have $\max\left\{\frac{d^2}{\sqrt{t}\eps},\frac{\sqrt{d}}{\eps^2},\frac{d}{\sqrt{t}\eps^2},\frac{d}{\eps}\right\} = \max\left\{\frac{\sqrt{d}}{\eps^2},\frac{d}{\eps}\right\}$, and so it suffices to show that the latter number of copies is sufficient. Now in this case, the variance is at most
    \begin{equation}
        \mbb{V}[\bar{X}] \leq \bigo\left(\frac{1}{n^2d^2t} + \frac{\tr(\rho^2)^{3/2}}{nt} + \frac{\tr(\rho^2)}{ndt}\right),
    \end{equation}
    where we again use $\tr(\Delta^2) \leq \tr(\rho^2)$. Again, by Chebyshev's inequality,
    \begin{align}
        \mathrm{Pr}[|\bar{X} - \mbb{E}[\bar{X}]| > \eps \cdot \tr(\rho^2)] &\leq \frac{\mbb{V}[\bar{X}]}{\eps \cdot \tr(\rho^2)^2}
        \leq \bigo\left(\frac{1}{n^2d^2t\eps^2\tr(\rho^2)^2} + \frac{1}{nt\eps^2\tr(\rho^2)^{1/2}} + \frac{1}{ndt\eps^2\tr(\rho^2)}\right)
        \\&\leq \bigo\left(\frac{1}{n^2t\eps^2} + \frac{\sqrt{d}}{nt\eps^2} + \frac{1}{nt\eps^2}\right)
        \\&= \bigo\left(\frac{1}{n^2t\eps^2} + \frac{\sqrt{d}}{nt\eps^2}\right),
    \end{align}
    where the penultimate step again used $\tr(\rho^2) \geq 1/d$. Thus, it suffices to take 
    \begin{equation}
        nt \geq C \cdot \max\left\{\frac{\sqrt{t}}{\eps}, \frac{\sqrt{d}}{\eps^2}\right\},
    \end{equation}
    where $C > 0$ is again a sufficiently large constant. The same argument as in \Cref{rmk:fewer-copies-suffice} also applies here, and we can assume the tester never measures more than $d^2$ copies at once. Thus, we can indeed succeed by taking
    \begin{equation}
        nt \geq C \cdot \max\left\{\frac{d}{\eps},\frac{\sqrt{d}}{\eps^2}\right\}.
    \end{equation}
\end{proof}

\section{Lower bounds}
\label{sec:lower}
We will now prove our lower bounds from \Cref{thm:intro-state-certification-main,thm:intro-lower-shared,thm:intro-mixedness-upper}. Without loss of generality, we will assume that testers have access to the unknown state $\rho$ in batches of $t$ copies at once, and sequentially perform POVMS $\mc{M}_1, \dots, \mc{M}_n$ on $\rho^{\otimes t}$. The total number of copies consumed is thus $nt$. We will only consider non-adaptive POVM schedules, where the choice of POVM $\mc{M}_i$ cannot depend on outcomes of POVMs $\mc{M}_1, \dots \mc{M}_{i-1}$. Note that non-adaptive POVM schedules can include POVMs that are chosen at random. In general, this randomness may be shared across the POVM choices. However, we will also consider POVM schedules where measurements are drawn using only \emph{private} randomness. First, we state our lower bound in the more general setting, when the measurement schedules can have shared randomness.

\begin{theorem}[Non-adaptive lower bound] 
\label{thm:lower-high-pres-shared}
    Let $t,d, d^* \in \mbb{N}, d > d^*, 0 < \eps \leq \bigo(1/t)$, where $d^*$ is an absolute constant. The copy complexity of $\eps$-certifying $d$-dimensional states with respect to the trace norm using non-adaptive $t$-copy measurements is at least
    \begin{equation}
        nt \ge \Omega\left(\min\left\{\frac{d^{3/2}}{t^{1/4}\eps^2}, \frac{1}{t^2\eps^3}\right\}\right).
    \end{equation}
\end{theorem}

For $\epsilon \leq \bigo\left(\frac{1}{d^{3/2}t^{7/4}}\right)$, the lower bound is $\Omega\left(\frac{d^{3/2}}{t^{1/4}\eps^2}\right)$; together with the general $\Omega(d/\eps^2)$ lower bound of \cite{o2015quantum}, this recovers the statement of \Cref{thm:intro-lower-shared}. Next, we state our lower bound against testers whose measurements cannot share randomness.

\begin{theorem}[Private randomness lower bound] 
\label{thm:lower-high-pres-private}
    Let $t,d, d^* \in \mbb{N}, d > 1, 0 < \eps \leq \bigo(1/t)$, where $d^*$ is an absolute constant. The copy complexity of $\eps$-certifying $d$-dimensional states with respect to the trace norm using $t$-copy measurements that are drawn with private randomness is at least
    \begin{equation}
        nt \ge \Omega\left(\min\left\{\frac{d^{2}}{\sqrt{t}\eps^2},\frac{1}{t^2\eps^3}\right\}\right).
    \end{equation}
\end{theorem}

Here, for $\epsilon \leq \bigo\left(\frac{1}{d^{2}t^{3/2}}\right)$, the lower bound above is $\Omega\left(\frac{d^{2}}{\sqrt{t}\eps^2}\right)$; again, with the general $\Omega(d/\eps^2)$ lower bound of \cite{o2015quantum}, we recover the lower bound of \Cref{thm:intro-state-certification-main,thm:intro-mixedness-upper}. Before moving to the proofs of the above bounds, we provide the necessary technical background in \Cref{sec:lower-bound-prelims}. The main analysis is carried out in \Cref{sec:lower-bound-analysis}, and we conclude the proofs in \Cref{sec:lower-bound-proofs}.

\subsection{Lower bound framework and useful lemmas}
\label{sec:lower-bound-prelims}

For $i \in [n]$, let the $i$th POVM $\mc{M}_i$ consist of operators $\{M^{(i)}_{x}\}_x$, where each operator $M^{(i)}_{x} \in \mbb{C}^{d^t \times d^t}$. Without loss of generality, we can use the standard assumption that each operator is rank-1 (see e.g., \cite[Lemma 4.8]{chen2022exponential}). Thus, let $M_x^{(i)} \triangleq \ket{\psi_x^{(i)}}\bra{\psi_x^{(i)}}$. Let the distribution over outcomes of the POVM $\mc{M}_i$ acting on $\rho^{\otimes t}$ be $p^{(i)}_\rho$, with $p^{(i)}_{\rho}(x) = \tr(M_x^{(i)} \rho^{\otimes t}) = \bra{\psi_x^{(i)}}\rho^{\otimes t}\ket{\psi_x^{(i)}}$. Let the overall distribution be $P_\rho^{(n)} \triangleq \bigotimes_{i \in [n]} p^{(i)}_\rho$.

We will also assume that all POVMs respect the Schur-basis structure of $\rho^{\otimes t}$. In particular, recall from \Cref{def:schur-transform} that a unitary $U_{\mathrm{schur}}^{d,t}$ block-diagonalizes $\rho^{\otimes t}$ into blocks indexed by partitions $\lambda \vdash t$. Following the notation of \Cref{def:schur-transform}, let $\{\Pi_\lambda\}_{\lambda \vdash t}$ be the projectors onto these blocks in the Schur basis, and let $\{P_\lambda = (U_{\mathrm{schur}}^{d,t})^\dag \Pi_\lambda U_{\mathrm{schur}}^{d,t}\}_{\lambda \vdash t}$ be the associated projectors in the standard basis. We will then consider POVMs of the following form:

\begin{definition}[Schur-refined POVMs]
    \label{def:schur-refined-povms}
    Let $\mathcal{M}$ be a POVM with measurement operators $M_x \in \mathbbm{C}^{d^t \otimes d^t}$. We define its Schur-refinement to be the POVM with outcomes $\{x,\lambda\}_{\lambda \vdash t}$ and associated measurement operators $M_{x,\lambda} \triangleq P_\lambda M_x P_\lambda$; we also say a POVM is Schur-refined if all its POVM operators can be written in this form.
\end{definition}

One can verify that the above is a valid POVM. Note that this refinement is without loss of generality when acting on $\rho^{\otimes t}$. One can always recover statistics of the original POVM by discarding the observed partition $\lambda$, as the total probability of observing outcome $x$ remains unchanged:
\begin{equation}
\label{eq:schur-refinement-generality}
    \sum_{\lambda} \tr(P_\lambda M_x P_\lambda \rho^{\otimes t}) = \tr(M_x \sum_\lambda P_\lambda \rho^{\otimes t} P_\lambda) = \tr(M_x \rho^{\otimes t}),
\end{equation}
where the last equality holds by \Cref{def:schur-transform}. Despite its generality, this refinement will be crucial to our later analysis; specifically, for the bound in \Cref{lem:reduced-channel-trace}. We also note that \Cref{lem:reduced-channel-trace} and its applications are the only components of our paper that use this refinement. 

Now, to prove our lower bounds, we will consider the specific case of mixedness testing. As is standard, we will prove the lower bound for the task of distinguishing between the maximally mixed state and a state drawn from an almost-$\eps$ perturbation of the maximally mixed state. We define such ensembles next.

\begin{definition}[Almost-$\eps$ perturbations]
    An ensemble of quantum states $D$ is an almost-$\eps$ perturbation of the maximally mixed state if 
    \begin{equation}
        \mathrm{Pr}_{\rho \sim D}[\| \rho - \mmstate\|_1 \geq \eps] \geq \frac12.
    \end{equation}
    Let $\mc{D}_\eps$ be the set of all almost-$\eps$ perturbations of $\mmstate$. 
\end{definition}

Our hard instance will be given by the following almost-$\eps$ perturbation, which first appeared in \cite{liu2024role}.

\begin{definition}
\label{def:rademacher-hard-instance}
    Let $d^2/2 \leq \ell \leq d^2-1$. For $\bfz \sim \{-1,+1\}^\ell$ and an orthornormal basis $V_1, \dots, V_{d^2}$ of $\mathbb{C}^{d \times d}$ with $V_{d^2} = \mathbbm{1}/\sqrt{d}$, let $\Delta_{\bfz} \triangleq \frac{c\eps}{\sqrt{d\ell}} \sum_{i = 1}^\ell z_i V_i$, for an appropriate constant $c > 0$. Let $a_{\bfz} \triangleq \min\{1, \frac{1}{d\|\Delta_{\bfz}\|_{\op}}\}$. Let $\bar{\Delta}_{\bfz} = a_{\bfz} \Delta_{\bfz}$. Then, we define
    \begin{equation}
        \rho_{\bfz} = \mmstate + \bar{\Delta}_{\bfz}.
    \end{equation}
    It will also be useful to define $\mc{V} = [\mathrm{vec}(V_1), \dots, \mathrm{vec}(V_\ell)]$; note that this matrix only includes the first $\ell$ perturbation operators.
\end{definition}

It was shown in \cite[Corollary 4.4]{liu2024role} that for any valid choice of $V_1, \dots V_{d^2}$, the ensemble from \Cref{def:rademacher-hard-instance} is indeed an almost-$\eps$ perturbation of the maximally mixed state.\footnote{Note that \cite{liu2024role} showed this only for $\eps$ smaller than an absolute constant. As our lower bounds only hold in the further restricted high-precision regime, we will not explicitly account for this constraint.} They also showed the following useful lemma, allowing us to replace bounds on mgfs of functions of the normalized perturbations with those of the unnormalized ones:

\begin{lemma}[Adapted from {\cite[Claim A.9]{liu2024role}}]
\label{lem:hard-instance-ignoring-normalization}
    Consider any function $f: \{-1,+1\}^{\ell} \times \{-1,+1\}^{\ell} \mapsto \mbb{R}$. For $z \in \{-1,+1\}^\ell$, let $a_z$ be as defined in \Cref{def:rademacher-hard-instance}. Then,
    \begin{equation}
        \mbb{E}_{\bfz,\bfz^\prime} \exp(a_{\bfz},a_{\bfz^\prime}f(\bfz,\bfz^\prime)) \leq \mbb{E}_{\bfz,\bfz^\prime} \exp(f(\bfz,\bfz^\prime)) + 4\exp(-d).
    \end{equation}
\end{lemma}

To show the hardness of distinguishing the maximally mixed state from the almost-$\eps$ perturbation defined above, we will show that the outcome distributions are statistically indistinguishable unless a large number of $t$-copy measurements are made. Specifically, using standard arguments (see e.g., \cite{acharya2020inference,liu2024role}), we have the following fact relating the copy complexity to the minmax and maxmin $\chi^2$-divergences between distributions over POVM outcomes.

\begin{fact}
\label{fact:copies-divergence}
    Let $n$ be large enough to succeed at mixedness testing with probability at least $2/3$ using measurements with shared randomness. Then, 
    \begin{equation}
        \min_{D \in \mc{D}_\eps} \max_{\mc{M}_1, \dots, \mc{M}_n} \divchi(\mbb{E}_{\rho \sim D}[P^{(n)}_\rho] \| P^{(n)}_{\mmstate}) \geq \frac12.
    \end{equation}
    Further, if $n$ is large enough to succeed even with privately random measurements, then
    \begin{equation}
        \max_{\mc{M}_1, \dots, \mc{M}_n} \min_{D \in \mc{D}_\eps} \divchi(\mbb{E}_{\rho \sim D}[P^{(n)}_\rho] \| P^{(n)}_{\mmstate}) \geq \frac12.
    \end{equation}
\end{fact}

To bound the above divergences, we will use the Ingster--Suslina method, which provides a convenient upper bound on such quantities.
\begin{lemma}
\label{lem:ingster-suslina}
    Let $\{p_{\bftheta}^{(i)}\}_{i \in [n]}$ be a set of distribution over $[d]$ parameterized by a random parameter $\bftheta$ and $\{q^{(i)}\}_{i \in [n]}$ be a set of fixed distributions. For $x \in [d]$, define the likelihoood ratio deviation
    \begin{equation}
        \delta_{\bftheta}^{(i)}(x) = \frac{p^{(i)}_{\bftheta}(x)}{q^{(i)}(x)} - 1.
    \end{equation}
    Then,
    \begin{equation}
        \divchi\left(\mbb{E}_{\bftheta} \bigotimes_{i \in [n]} p_{\bftheta}^{(i)} \bigg\| \bigotimes_{i \in [n]} q^{(i)}\right) \leq \mbb{E}_{\bftheta,\bftheta^\prime}\exp\left(\sum_{i \in [n]} \phi_i(\bftheta,\bftheta^\prime)\right) - 1,
    \end{equation}
    where $\bftheta^\prime$ is an independent copy of $\bftheta$, and
    \begin{equation}
        \phi_i(\bftheta,\bftheta^\prime) = E_{x \sim q^{(i)}(x)}[\delta_{\bftheta}^{(i)}(x) \delta_{\bftheta^\prime}^{(i)}(x)].
    \end{equation}
\end{lemma}

To compute the inner product of likelihood deviations above, it will be helpful to define the L\"uders channel associated with a POVM.

\begin{definition}[L\"uders channel]
Let $\mc{M}$ be a POVM with measurement operators $\{M_x\}_x$, where each operator is a rank-$1$ operator with dimension $d \times d$. Then, the L\"uders channel associated with $\mc{M}$ acts on quantum states $\rho \in \mbb{C}^{d \times d}$ as follows:
    \begin{equation}
        \mc{H}_{\mc{M}}(\rho) = \sum_{x} \tr(M_x \rho) \hat{M}_x.
    \end{equation}
\end{definition}

We will also state some useful properties of such channels:

\begin{lemma}[Properties of L\"uders channel; {\cite[Lemma 3.4]{liu2024role}}]
\label{lem:luders-properties}
    Let $\mc{H}_{\mc{M}}$ be the L\"uders associated with some POVM $\mc{M}$ over $d$-dimensional states. Then, 
    \begin{enumerate}
        \item $\mc{H}_{\mc{M}}$ has Hermitian eigenvectors, say,  $V_1, \dots, V_{d^2}$ with associated eigenvalues $0 \leq \lambda_1, \dots, \lambda_{d^2} \leq 1$. 
        \item The eigenvectors of $\mc{H}_{\mc{M}}$ form an orthonormal basis of $\mathbb{C}^{d \times d}$.
        \item $\mc{H}_{\mc{M}}$ is unital, i.e., $\mc{H}_{\mc{M}}(\mathbbm{1}) = \mathbbm{1}$. Thus, one of the eigenvectors is $\mathbbm{1}/\sqrt{d}$, with eigenvalue 1. WLOG, assume $V_{d^2} = \mathbbm{1}/\sqrt{d}, \lambda_{d^2} = 1$. Consequently, all other eigenvectors are traceless.
        \item $\tr(\mc{H}_{\mc{M}}) = \sum_{i = 1}^{d^2} \lambda_i = d$.
    \end{enumerate}
\end{lemma}

Let us now define a class of induced channels that will arise naturally in our analysis:

\begin{definition}[Single-register induced channel]
\label{def:reduced-channels}
Given $t \in \mbb{N}$, as well as a quantum channel $\mc{H} : \mbb{C}^{d^t \times d^t} \mapsto \mbb{C}^{d^t \times d^t}$, we define for all $j,k \in [t]$ the channel $\Tilde{\mc{H}}_{j,k} : \mbb{C}^{d \times d} \mapsto \mbb{C}^{d \times d}$ as follows:
\begin{equation}
    \Tilde{\mc{H}}_{j,k} (M) = \tr_{[t] \setminus \{j\}} (\mc{H}(M_k \otimes \mmstate^{\otimes [t] \setminus \{k\}})).
\end{equation}
We also define the average of these channels, $\Tilde{\mc{H}} \triangleq \frac{1}{t^2} \sum_{j,k \in [t]} \Tilde{\mc{H}}_{j,k}$.
\end{definition}

Next, we state and prove some useful properties of such channels when they are induced by a L\"uders channel on the $d^t$-dimensional space.

\begin{lemma}
\label{lem:reduced-channel-basic-properties}
    Let $\mc{H}$ be a L\"uders channel associated with a measurement $\mathcal{M}$ over $d^t$-dimensional states. Let $\tilde{\mc{H}}$ be the induced channel defined in \Cref{def:reduced-channels}. Then, $\tilde{\mc{H}}$ satisfies items 1-3 of \Cref{lem:luders-properties}.
\end{lemma} 

\begin{proof}
We split the proof by component.
\paragraph{Hermitian spectral decomposition:} It is not hard to see that $\tilde{\mc{H}}_{j,k}^\dag = \tilde{\mc{H}}_{k,j}$. As $\tilde{\mc{H}}$ sums over all pairs $j,k$, we have $\tilde{\mc{H}} = \tilde{\mc{H}}^\dag$, i.e., the reduced channel is Hermitian. As the L\"uders channel $\mc{H}$ is Hermiticity preserving, and so are the other operations in $\tilde{\mc{H}}_{j,k}$, we immediately see that $\tilde{\mc{H}}_{j,k}$ is Hermiticity preserving; the same property is obtained for $\tilde{\mc{H}}$ by linearity. As a consequence of these two properties, the eigenvectors of $\tilde{\mc{H}}$ are Hermitian, and form an orthonormal basis of the space of $d \times d$ Hermitian operators. However, as this space has the same dimension as $\mbb{C}^{d \times d}$, they also form an orthonormal basis of $\mbb{C}^{d \times d}$.

\paragraph{Unitality:} Recall that $\mc{H}$ itself is unital, and so by definition, the maps $\tilde{\mc{H}}_{j,k}$ are also unital. $\tilde{\mc{H}}$ is thus unital by linearity. 

\paragraph{Eigenvalue upper bound:} We wish to show that all eigenvalues of $\tilde{\mc{H}}$ are at most $1$. More generally, we will show that this upper bound holds for all unital CPTP maps. For the sake of contradiction, assume the channel has some eigenvalue $\lambda > 1$, and let $N \in \mbb{C}^{d \times d}$ be the corresponding non-zero traceless eigenvector. Let $c > 0$ be the largest possible real such that $\sigma = \mmstate + cN$ is psd; such a $c$ must exist as $N$ being traceless and non-zero must have a negative entry along its diagonal. Then, $\tilde{\mc{H}}(\sigma) = \mmstate + c\lambda N$. As $c\lambda > c$, $\tilde{\mc{H}}(\sigma)$ is not psd, implying that $\tilde{\mc{H}}$ is not completely positive, giving us our desired contradiction.

\end{proof}

Next, we bound the trace of such induced channels associated with Schur-refined POVMs.

\begin{lemma}
\label{lem:reduced-channel-trace}
    Let $\mc{M}$ and  $\tilde{\mc{H}}$ be as in \Cref{lem:reduced-channel-basic-properties}.
    When $\mathcal{M}$ is Schur-refined (as in \Cref{def:schur-refined-povms}), $\tilde{\mc{H}}$ satisfies 
    \begin{equation}
        \tr(\tilde{\mc{H}}) \leq 1 + \frac{2d}{\sqrt{t}}.
    \end{equation} 
\end{lemma}

\begin{proof}
    As we have a Schur-refined POVM, we can write the measurement operators in the form
    \begin{equation}
        M_{x,\lambda} = P_\lambda M_x P_\lambda.
    \end{equation}
    For ease of notation, we will instead index the measurement operators by $y = (x,\lambda)$. For outcome $y$ and for all $k \in [t]$, we define
    \begin{equation}
        \rho_y^{(k)} \triangleq \frac{\tr_{[t]\setminus \{k\}} (M_y)}{\tr(M_y)}.
    \end{equation}
    Note that each $\rho_y^{(k)}$ is a density matrix. Further, define the averaged operator over $k$ to be
    \begin{equation}
        \overline{\rho}_y \triangleq \frac1t \sum_{k = 1}^t \rho_y^{(k)}.
    \end{equation}
    Then, for all $A \in \mathbb{C}^{d \times d}$, and all $j,k \in [t]$, we can write
    \begin{equation}
        \tilde{\mc{H}}_{j,k}(A) = \frac{1}{d^{t-1}} \sum_y \tr(M_y) \tr(\rho_{j}^{(k)} A) \rho_y^{(k)}.
    \end{equation}
    Summing over $j,k$ and averaging, we get
    \begin{equation}
        \tilde{\mc{H}}(A) = \frac{1}{d^{t-1}} \sum_y \tr(M_y) \tr(\overline{\rho}_y  A) \cdot \overline{\rho}_y
    \end{equation}
    
    One can verify that superoperators of the form $A \mapsto \tr(AB) C$ have trace $\tr(BC)$. Thus,
    \begin{equation}
        \tr(\tilde{\mc{H}}) = \frac{1}{d^{t-1}}\sum_y \tr(M_y) \tr(\overline{\rho}_y^2).
    \end{equation}
    
    Note that when applying the POVM $\{M_y\}$ to $\mmstate^{\otimes t}$, outcome $y$ is observed with probability $\tr(M_y)/d^t$. Consequently, we can write
    
    \begin{equation}
        \tr(\tilde{\mc{H}}) = d \cdot \mathbb{E}_{y} \tr(\overline{\rho}_y^2),
    \end{equation}
    where the expectation is over the outcomes obtained from applying the POVM to $\mmstate^{\otimes t}$.
    
    Next, we use the following claim, whose proof we defer to \Cref{app:claim}.
    
    \begin{claim}
    \label{claim:operator-norm}
        Let $y = (x,\lambda)$, and let $\overline{\rho}_y$ be as defined above. Then, $\|\overline{\rho}_y\|_\infty \leq \frac{\lambda_1}{t}$, where $\lambda_1$ is the largest part of the partition $\lambda$.
    \end{claim}
    
    For each outcome $y$, let $\lambda(y)$ denote the corresponding partition, and $\lambda_1(y)$ be its largest part. Using the above claim, we have
    \begin{equation}
        \tr(\tilde{\mc{H}}) \leq d \cdot \mathbb{E}_y \tr(\overline{\rho}_y) \|\overline{\rho}_y\|_\infty \leq \frac{d}{t} \cdot \mathbb{E}_y \lambda_1(y).
    \end{equation}
    
    Note that the probability associated with a particular partition $\lambda$ is given by
    
    \begin{equation}
        \sum_{y : \lambda(y) = \lambda} \tr(M_y \mmstate^{\otimes t}) = \sum_{x} \tr(P_\lambda M_x P_\lambda \mmstate^{\otimes t}) = \tr(P_\lambda \mmstate^{\otimes t}),
    \end{equation}
    as the operators $M_x$ form a POVM and thus sum to identity. However, this is the same probability as if we performed weak Schur sampling (\Cref{def:schur-transform}) on $\mmstate^{\otimes t}$. Thus, we have
    \begin{equation}
        \tr(\mc{\tilde{H}}) \leq \frac{d}{t} \cdot \mathbb{E}_\lambda \lambda_1,
    \end{equation}
    where the expectation is over the outcome of weak Schur sampling applied to $t$ copies of the maximally mixed state. By~\cite[Theorem 1.4]{o2017efficient}, we have $\mathbb{E} \lambda_1 \leq \frac{t}{d} + 2\sqrt{t}$, immediately yielding the lemma.
\end{proof}

The final necessary property of such induced channels is an even sharper upper bound on the eigenvalues associated with traceless operators.

\begin{lemma}
    \label{lem:reduced-channel-operator}
    Let $\tilde{\mc{H}}$ be as in \Cref{lem:reduced-channel-basic-properties}, and let its eigenvalues be $\lambda_1 \leq \dots \leq \lambda_{d^2-1} \leq \lambda_{d^2} = 1$. Then,
    \begin{equation}
        \lambda_{d^2-1} \leq \frac1t.
    \end{equation}
\end{lemma}

\begin{proof}
    Recall that $\lambda_{d^2} = 1$ was the eigenvalue associated with the identity operator. Thus, to prove the lemma, we can equivalently show that for each \emph{traceless} Hermitian $A \in \mathbb{C}^{d \times d}$, we have
    \begin{equation}
        \tr(A \mathcal{\Tilde{H}}(A)) \leq \frac1t \tr(A^2).
    \end{equation}
    Let $\{M_x\}_x$ be the measurement operators of the underlying POVM, and let $\widehat{M}_x = M_x / \tr(M_x)$. For outcome $x$ and for all $k \in [t]$, we define
    \begin{equation}
        \rho_x^{(k)} \triangleq \frac{\tr_{[t] \setminus \{k\}}(M_x)}{\tr(M_x)}, \quad \overline{\rho}_x \triangleq \frac1t \sum_{k = 1}^t \rho_x^{(k)}.
    \end{equation}
    By definition of $\Tilde{\mc{H}}$, one can write
    \begin{equation}
        \tr(A \Tilde{\mc{H}}(A)) = \frac{1}{d^{t-1}} \sum_{x} \tr(M_x) \tr(\overline{\rho}_x A)^2.
    \end{equation}
    We will rewrite the above expression in a more convenient form. First, recalling the symmetric sum, we define
    \begin{equation}
        K_A \triangleq \symsum A \otimes I^{\otimes t-1} = \sum_{k = 1}^t A_k \otimes I^{\otimes [t] \setminus \{k\}}.
    \end{equation}
    Then, we have
    \begin{equation}
        \tr(\overline{\rho}_x A) = \frac1t \tr(\widehat{M}_x K_A).
    \end{equation}
    Consequently,
    \begin{equation}
        \tr(A \Tilde{\mc{H}}(A)) = \frac{1}{t^2 d^{t-1}} \sum_x \tr(M_x) \tr(\widehat{M}_x K_A)^2 \leq \frac{1}{t^2 d^{t-1}} \sum_x \tr(M_x) \tr(\widehat{M}_x K_A^2),
    \end{equation}
    where the inequality used the identity $\tr(\rho B)^2 \leq \tr(\rho B^2)$ for all density matrices $\rho$ and Hermitian matrices $B$. Thus,
    \begin{equation}
    \label{eq:operator-norm-bound}
        \tr(A \Tilde{\mc{H}}(A)) \leq \frac{1}{t^2 d^{t-1}} \sum_x \tr(M_x K_A^2) = \frac{1}{t^2 d^{t-1}} \tr(K_A^2),
    \end{equation}
    as $\{M_x\}_x$ form a collection of POVM operators. To compute $\tr(K_A^2)$, we expand
    \begin{align}
        K_A^2 &= (\sum_{k = 1}^t A_k \otimes I^{\otimes [t]\setminus \{k\}})^2
        \\&= \sum_{k = 1}^t A_k^2 \otimes I^{\otimes [t] \setminus \{k\}} + \sum_{j \neq k \in [t]} A_j \otimes A_k \otimes I^{\otimes [t]\setminus \{j,k\}}.
    \end{align}
    Recall that $A$ is traceless, and so,
    \begin{equation}
        \tr(K_A^2) = \sum_{k = 1}^t \tr(A^2) \tr(I)^{t - 1} = t d^{t-1} \tr(A^2).
    \end{equation}
    Substituting the above into \Cref{eq:operator-norm-bound}, we have
    \begin{equation}
        \tr(A \tilde{\mc{H}}(A)) \leq \frac{1}{t} \tr(A^2),
    \end{equation}
    as we set out to prove.
\end{proof}

Lastly, we will use the following standard fact about Rademacher random variables to bound the moment generating functions arising from our use of \Cref{lem:ingster-suslina}.

\begin{lemma}
\label{lem:rademacher-quadratic-mgf}
    Let $\bfz,\bfz^\prime \sim \{-1,+1\}^\ell$. Then, for $M \in \mbb{C}^{\ell \times \ell}$,
    \begin{equation}
        \mbb{E}_{\bfz,\bfz^\prime}[
        \exp(\lambda \bfz^\top M \bfz^\prime)
        ] \leq \exp(\bigo(\lambda^2\|M\|_2^2)),
    \end{equation}
    whenever $\lambda = \bigo(1/\|M\|_{\infty})$.
\end{lemma}

\subsection{Bounding the $\chi^2$-divergence}
\label{sec:lower-bound-analysis}

Given \Cref{fact:copies-divergence}, we will bound the $\chi^2$-divergence between the distributions over measurement outcomes. First, we use \Cref{lem:ingster-suslina} to relate such divergences to inner products involving the L\"uders channel associated with the measurements.

\begin{lemma}
\label{lem:tcopy-rademacher-chi2}
    Given POVM schedule $\mc{M}_1, \dots, \mc{M}_n$, let the associated L\"uders channels be $\mc{H}_1, \dots, \mc{H}_n$. Let the average L\"uders channel be $\mc{H} = \frac1n \sum_{i} \mc{H}_i$. Let $\Delta_{\bfz}^{(t)} = \rho_{\bfz}^{\otimes t} - \mmstate^{\otimes t}$. Then,
    \begin{equation}
         \divchi(\mbb{E}_{\bfz} P^{(n)}_{\rho_{\bfz}} \| P_{\mmstate}^{(n)}) \leq \mbb{E}_{\bfz,\bfz^\prime}\exp\left(nd^t \tr(\Delta_{\bfz}^{(t)} \mc{H}(\Delta_{\bfz^\prime}^{(t)}))\right) - 1.
    \end{equation}
\end{lemma}

\begin{proof}
    Note that for the $i$th measurement, the likelihood ratio deviation is given by
    \begin{equation}
        \delta_z^{(i)}(x) = \frac{p_{\rho_z}(x) - p_{\mmstate}(x)}{p_{\mmstate}(x)} = \frac{\tr(M_x^{(i)} (\rho_z^{\otimes t} - \mmstate^{\otimes t})}{\tr(\mmstate^{\otimes t} M_x^{(i)})} = \frac{d^t \bra{\psi_x^{(i)}}\Delta_z^{(t)}\ket{\psi_x^{(i)}}}{\braket{\psi_x^{(i)}|\psi_x^{(i)}}}.
    \end{equation}
    Then,
    \begin{align}
        \phi_i(\bfz,\bfz^\prime) &= \sum_x p_{\mmstate}(x) \delta_{\bfz}^{(i)}(x) \delta_{\bfz^\prime}^{(i)}(x) 
        \\&=  \sum_x d^t\frac{\bra{\psi_x^{(i)}}\Delta_{\bfz}^{(t)}\ket{\psi_x^{(i)}} \bra{\psi_x^{(i)}}\Delta_{\bfz^\prime}^{(t)}\ket{\psi_x^{(i)}}}{\braket{\psi_x^{(i)}|\psi_x^{(i)}}}
        \\&= \sum_x d^t\frac{\tr(\Delta_{\bfz}^{(t)}\ket{\psi_x^{(i)}} \bra{\psi_x^{(i)}}\Delta_{\bfz^\prime}^{(t)}\ket{\psi_x^{(i)}}\bra{\psi_x^{(i)}})}{\braket{\psi_x^{(i)}|\psi_x^{(i)}}}
        \\&=  d^t \sum_x\tr(\Delta_{\bfz}^{(t)} \hat{M}_x) \tr(\Delta_{\bfz^\prime}^{(t)} M_x)
        \\&= d^t \tr(\Delta_{\bfz}^{(t)} \mc{H}_i(\Delta_{\bfz^\prime}^{(t)})).
    \end{align}
    Now, by \Cref{lem:ingster-suslina},
    \begin{align}
        \divchi(\mbb{E}_{\bfz} P^{(n)}_{\rho_{\bfz}} \| P_{\mmstate}^{(n)}) &\leq \mbb{E}_{\bfz,\bfz^\prime}\exp\left(d^t \sum_{i \in [n]} \tr(\Delta_{\bfz}^{(t)} \mc{H}_i(\Delta_{\bfz^\prime}^{(t)}))\right) - 1
        \\&= \mbb{E}_{\bfz,\bfz^\prime}\exp\left(nd^t \tr(\Delta_{\bfz}^{(t)} \mc{H}(\Delta_{\bfz^\prime}^{(t)}))\right) - 1,
    \end{align}
    as desired.
\end{proof}

We note that the inner product above is a polynomial of degree $2t$ in the entries of the vectors $z,z^\prime$. It is unclear how to bound the MGF of such polynomials in general, and we instead simplify this expression by splitting it up into its linear and higher-order components. First, note that we can write 

\begin{align}
    \Delta_{z}^{(t)} &= (\mmstate + \bar{\Delta}_z)^{\otimes t} - \mmstate^{\otimes t}
    \\&= \sum_{k = 1}^t \symsum \bar{\Delta}_z^{\otimes k} \otimes \mmstate^{\otimes t - k}.
    \\&= \underbrace{\symsum \bar{\Delta}_z \otimes \mmstate^{\otimes t - 1}}_{\triangleq L_z} + \underbrace{\sum_{k = 2}^t \symsum \bar{\Delta}_z^{\otimes k} \otimes \mmstate^{\otimes t - k}}_{\triangleq H_z},
\end{align}
where $L_z$ consists of the linear perturbations and $H_z$ consists of the higher-order perturbations. Given the above, we can rewrite the innermost trace in the statement of \Cref{lem:tcopy-rademacher-chi2} as follows:
\begin{align}
    \tr(\Delta_{\bfz}^{(t)} \mc{H}(\Delta_{\bfz^\prime}^{(t)})) &= \tr((L_{\bfz} + H_{\bfz}) \mc{H}(L_{\bfz^\prime} + H_{\bfz^\prime})) 
    \\&= \tr(L_{\bfz} \mc{H}(L_{\bfz^\prime})) + \tr(L_{\bfz} \mc{H}(H_{\bfz^\prime})) + \tr(H_{\bfz} \mc{H}(L_{\bfz^\prime})) + \tr(H_{\bfz} \mc{H}(H_{\bfz^\prime})). 
\end{align}

To upper bound the MGF appearing in \Cref{lem:tcopy-rademacher-chi2}, our strategy will be to prove uniform upper bounds on the terms involving $H_z$, and then bound the MGF of the linear perturbations. We achieve the former in the following lemma:

\begin{lemma}
    \label{lem:rademacher-higher-order-terms}
    When $t\eps \leq \bigo(1)$, we have
    \begin{equation}
        \tr(L_{\bfz} \mc{H}(H_{\bfz^\prime})) +  \tr(H_{\bfz} \mc{H}(L_{\bfz^\prime})) + \tr(H_{\bfz} \mc{H}(H_{\bfz^\prime}))\leq \bigo\left(\frac{t^3\eps^3}{d^t}\right).
    \end{equation}
\end{lemma}

\begin{proof}
    From \Cref{lem:luders-properties} , $\|H\|_{\infty} \leq 1,$ and so 
    \begin{align}
        \tr(L_{\bfz} \mc{H}(H_{\bfz^\prime})) & \leq \|L_{\bfz}\|_2 \|H_{\bfz^\prime}\|_2
    \end{align}
    We first bound the norm for the linear perturbations:
    \begin{align}
        \|L_{\bfz}\|_2 &= \|\symsum \bar{\Delta}_{\bfz} \otimes \mmstate^{\otimes t - 1}\|_2
        \\&\leq \frac{t}{d^{(t-1)/2}} \cdot \|\bar{\Delta}_{\bfz}\|_2
        \\&= \frac{t}{d^{(t-1)/2}} \cdot \frac{a_{\bfz} c\eps}{\sqrt{d}} \leq \bigo\left(\frac{t\eps}{d^{t/2}}\right),
    \end{align}
    where we used the triangle inequality and the fact that $\|\mmstate\|_2 = 1/\sqrt{d}$ in the second step, and that $a_{\bfz} \leq 1$ in the final step. Similarly, for the higher-order term, we have
    \begin{align}
        \|H_{\bfz}\|_2 &= \|\sum_{k = 2}^{t} \symsum \bar{\Delta}_{\bfz}^{\otimes k} \otimes \mmstate^{\otimes t-k} \|_2
        \\&\leq \sum_{k = 2}^t \binom{t}{k} \|\bar{\Delta}_z\|_2^k \cdot \frac{1}{d^{(t-k)/2}}
        \\&\leq \sum_{k = 2}^t \frac{t^k}{d^{(t-k)/2}} \cdot \frac{a_{\bfz}^k c^k \eps^k}{d^{k/2}}
        \\&\leq \sum_{k = 2}^t \frac{t^k\epsilon^k}{d^{t/2}}
        \\&\leq \bigo\left(\frac{t^2\eps^2}{d^{t/2}}\right),
    \end{align}
    where the last inequality holds under the assumption $t\eps \leq \bigo(1)$. Now,
    \begin{equation}
        \tr(L_{\bfz} \mc{H}(H_{\bfz^\prime})) \leq \|L_{\bfz}\|_2 \|H_{\bfz^\prime}\|_2 \leq \bigo\left(\frac{t^3\eps^3}{d^t}\right).
    \end{equation}
    Similarly, we have $\tr(H_{\bfz} \mc{H}(L_{\bfz^\prime})) \leq \bigo(t^3\eps^3/d^t)$. Finally, 
    \begin{equation}
        \tr(H_{\bfz} \mc{H}(H_{\bfz^\prime})) \leq \|H_{\bfz}\|_2 \|H_{\bfz^\prime}\|_2 \leq \bigo\left(\frac{t^4\eps^4}{d^t}\right).
    \end{equation}
    As $t\eps \leq \bigo(1)$, we have 
    \begin{equation}
        tr(L_{\bfz} \mc{H}(H_{\bfz^\prime})) +  \tr(H_{\bfz} \mc{H}(L_{\bfz^\prime})) + \tr(H_{\bfz} \mc{H}(H_{\bfz^\prime}))\leq \bigo\left(\frac{t^3\eps^3}{d^t}\right),
    \end{equation}
    as desired.
\end{proof}

Next, we bound the MGF associated with the linear perturbations:
\begin{lemma}
    \label{lem:rademacher-linear-terms}
    For $1 \leq j,k \leq t$, let $\Tilde{\mc{H}}_{j,k}(M) = \tr_{[t] \setminus \{j\}}(\mc{H}(M_k \otimes \mmstate^{\otimes [t] \setminus \{k\}}))$. Further, let $\Tilde{\mc{H}} = \frac{1}{t^2} \sum_{j,k} \Tilde{\mc{H}}_{j,k}$. Let $S_{\Tilde{\mc{H}}}$ be the Liouville representation of $\Tilde{\mc{H}}$.Then, 
    \begin{equation}
        \mbb{E}_{\bfz, \bfz^\prime}[\exp(nd^t \tr(L_{\bfz} \mc{H}(L_{\bfz^\prime})))] \leq \exp\left(\bigo\left(
            \frac{n^2t^4\eps^4}{\ell^2} \cdot \|\mc{V}^\dag S_{\Tilde{\mc{H}}} \mc{V}\|_2^2
        \right)\right) + 4\exp(-d),
    \end{equation}
    whenever $n = \bigo\left(\frac{\ell}{t^2 \eps^2\|\mc{V}^\dag S_{\Tilde{\mc{H}}} \mc{V}\|_{\infty}}\right)$.
\end{lemma}

\begin{proof}
    We will first rewrite $\tr(L_{\bfz} \mc{H}(L_{\bfz^\prime}))$ in terms of $\Tilde{\mc{H}}$.
    \begin{align}
        \langle L_{\bfz}, \mc{H}(L_{\bfz^\prime}) \rangle &= \sum_{j,k = 1}^t \left\langle \bar{\Delta}_{\bfz}^{(j)} \otimes \frac{\mathbbm{1}^{\otimes [t] \setminus \{j\}}}{d^{t-1}}, \mc{H}(\bar{\Delta}_{\bfz^\prime}^{(k)} \otimes \mmstate^{[t] \setminus \{k\}}) \right\rangle
        \\&= \frac{a_{\bfz} a_{\bfz^\prime}}{d^{t-1}} \sum_{j,k = 1}^t \tr(\Delta_{\bfz} \cdot \tr_{[t] \setminus \{j\}}(\Delta_{\bfz^\prime} \otimes \mmstate^{[t] \setminus \{k\}}))
        \\&= \frac{a_{\bfz}a_{\bfz^\prime}t^2}{d^{t-1}} \tr(\Delta_{\bfz} \Tilde{\mc{H}}(\Delta_{\bfz^\prime})).
    \end{align}
    To remove the normalization factors, we use \Cref{lem:hard-instance-ignoring-normalization}, giving us
    \begin{align}
        \mbb{E}_{\bfz, \bfz^\prime}[\exp(nd^t \tr(L_{\bfz} \mc{H}(L_{\bfz^\prime})))] &= \mbb{E}_{\bfz, \bfz^\prime}[\exp(a_{\bfz} a_{\bfz^\prime} \cdot ndt^2 \cdot \tr(\Delta_{\bfz} \Tilde{\mc{H}}(\Delta_{\bfz^\prime})))] \\&\leq \mbb{E}_{\bfz, \bfz^\prime}[\exp(ndt^2 \cdot \tr(\Delta_{\bfz} \Tilde{\mc{H}}(\Delta_{\bfz^\prime})))] + \frac{4}{e^d}.
    \end{align}
    To complete the proof, we wish to rewrite $\tr(\Delta_{\bfz} \Tilde{\mc{H}}(\Delta_{\bfz^\prime}))$ in the form of \Cref{lem:rademacher-quadratic-mgf}. Using the Liouville representation, we can write
    \begin{equation}
        \tr(\Delta_{\bfz} \Tilde{\mc{H}}(\Delta_{\bfz^\prime})) = \mathrm{vec}(\Delta_{\bfz})^\dag S_{\Tilde{\mc{H}}} \mathrm{vec}(\Delta_{\bfz^\prime}).
    \end{equation}
    Now, recall that
    \begin{equation}
        \Delta_z = \frac{c\eps}{\sqrt{d\ell}} \cdot \sum_{i = 1}^\ell V_i z_i.
    \end{equation}
    We can thus write
    \begin{equation}
        \mathrm{vec}(\Delta_z) = \frac{c\eps}{\sqrt{d\ell}} \sum_{i = 1}^\ell \mathrm{vec}(V_i) z_i = \frac{c\eps}{\sqrt{d\ell}} \mc{V}z.
    \end{equation}
    Finally, we rewrite 
    \begin{equation}
        ndt^2 \cdot \tr(\Delta_{\bfz} \Tilde{\mc{H}}(\Delta_{\bfz^\prime})) = \frac{c^2 nt^2\eps^2}{\ell} \bfz^\top \mc{V}^\dag S_{\Tilde{\mc{H}}} \mc{V} \bfz^\prime.
    \end{equation}
    We can now invoke \Cref{lem:rademacher-quadratic-mgf} to upper bound $\mbb{E}_{\bfz, \bfz^\prime}[\exp(ndt^2 \cdot \tr(\Delta_{\bfz} \Tilde{\mc{H}}(\Delta_{\bfz^\prime})))]$ and get the desired statement.
\end{proof}

\subsection{Putting things together}
\label{sec:lower-bound-proofs}
Given \Cref{lem:rademacher-linear-terms,lem:rademacher-higher-order-terms}, we can now prove \Cref{thm:lower-high-pres-private,thm:lower-high-pres-shared}.
\begin{proof}[Proof of \Cref{thm:lower-high-pres-shared}]
    By \Cref{fact:copies-divergence}, it suffices to show that $\divchi$ is small for any appropriate ensemble. We will consider the ensemble in \Cref{def:rademacher-hard-instance} for some fixed orthonormal basis $V_1, \dots, V_{d^2}$.
    Then, by \Cref{lem:tcopy-rademacher-chi2}, we have
    \begin{align}
         \divchi(\mbb{E}_{\bfz} P^{(n)}_{\rho_{\bfz}} \| P_{\mmstate}^{(n)}) &\leq \mbb{E}_{\bfz,\bfz^\prime}\exp\left(nd^t \tr(\Delta_{\bfz}^{(t)} \mc{H}(\Delta_{\bfz^\prime}^{(t)}))\right) - 1.
         \\&\leq \mbb{E}_{\bfz,\bfz^\prime}\exp\left(nd^t \tr(L_{\bfz} \mc{H}(L_{\bfz^\prime})) + \bigo(nt^3\eps^3)\right) - 1
         \\&\leq \exp\left(\bigo\left(
            \frac{n^2t^4\eps^4}{\ell^2} \cdot \|\mc{V}^\dag S_{\Tilde{\mc{H}}} \mc{V}\|_2^2 + nt^3\eps^3
        \right)\right) - 1 + 4\exp( \bigo(nt^3\eps^3) -d) \label{eq:shared-proof-1},
    \end{align}
    where we used \Cref{lem:rademacher-higher-order-terms} in the second step, which holds for $\eps \leq \bigo(1/t)$; this is satisfied under the hypothesis of the theorem. We used \Cref{lem:rademacher-linear-terms} in the last step, assuming
    \begin{equation}
        n \leq \bigo\left(\frac{\ell}{t^2 \eps^2\|\mc{V}^\dag S_{\Tilde{\mc{H}}} \mc{V}\|_{\infty}}\right). \label{eq:shared-proof-2}
    \end{equation}
    Now, the $\chi^2$-divergence is $o(1)$ unless $n$ is large enough for \Cref{eq:shared-proof-1} to be $\Omega(1)$ or for the assumed \Cref{eq:shared-proof-2} to be false. The former condition implies 
    \begin{equation}
        n = \Omega\left(\min\left\{\frac{\ell}{t^2\eps^2 \|\mc{V}^\dag S_{\Tilde{\mc{H}}} \mc{V}\|_2}, \frac{1}{t^3 \eps^3}, \frac{d}{t^3\eps^3} \right\}\right) = \Omega\left(\min\left\{\frac{\ell}{t^2\eps^2 \|\mc{V}^\dag S_{\Tilde{\mc{H}}} \mc{V}\|_2}, \frac{1}{t^3 \eps^3} \right\}\right),
    \end{equation}
    where we use the fact that for $d$ sufficiently large, i.e., at least some unspecified constant $d^*$, $4\exp(\bigo(nt^3\eps^3) - d) \leq .01$ unless $n \geq \Omega(\frac{d}{t^3\eps^3})$. On the other hand, the latter condition implies
    \begin{equation}
        n = \Omega\left(\frac{\ell}{t^2\eps^2 \|\mc{V}^\dag S_{\Tilde{\mc{H}}} \mc{V}\|_\infty}\right) \geq \Omega\left(\frac{\ell}{t^2\eps^2 \|\mc{V}^\dag S_{\Tilde{\mc{H}}} \mc{V}\|_2}\right).
    \end{equation}
    Together, both conditions yield the lower bound 
    \begin{equation}
        n = \Omega\left(\min\left\{\frac{\ell}{t^2\eps^2 \|\mc{V}^\dag S_{\Tilde{\mc{H}}} \mc{V}\|_2}, \frac{1}{t^3 \eps^3} \right\}\right). \label{eq:shared-proof-3}
    \end{equation}
    To proceed, we will upper bound the $\|\cdot\|_2$-norm above. Let $\lambda_1 \leq \dots \leq \lambda_{d^2} = 1$ be the eigenvalues of $\Tilde{\mc{H}}$. As $\mc{V}$ is an isometry consisting of columns orthogonal to $\mathrm{vec}(I)$, we have
    \begin{equation}
        \|\mc{V}^\dag S_{\Tilde{\mc{H}}} \mc{V}\|_2^2 \leq \sum_{i = 1}^{d^2-1} \lambda_i^2 \leq \frac{1}{t} \sum_{i = 1}^{d^2-1} \lambda_i = \frac{\tr(\Tilde{\mc{H}}) - 1}{t},
    \end{equation}
    where the second inequality used \Cref{lem:reduced-channel-operator}. Now, recall that, from the discussion surrounding \Cref{eq:schur-refinement-generality}, we can assume the underlying POVM associated with $\mc{\Tilde{H}}$ is Schur-refined. Thus, we can apply \Cref{lem:reduced-channel-trace}, obtaining
    \begin{equation}
        \|\mc{V}^\dag S_{\Tilde{\mc{H}}} \mc{V}\|_2 \leq \sqrt{\frac{2d}{t^{3/2}}}.
    \end{equation}
    Now, with \Cref{eq:shared-proof-3}, and setting $\ell = d^2 - 1$ (see \Cref{def:rademacher-hard-instance}), we have
    \begin{equation}
        nt = \Omega\left(\min\left\{\frac{d^{3/2}}{t^{1/4}\eps^2}, \frac{1}{t^2\eps^3}\right\}\right),
    \end{equation}
    as desired.
\end{proof}

Finally, we prove \Cref{thm:lower-high-pres-private}.
\begin{proof}[Proof of \Cref{thm:lower-high-pres-private}]
    Using \Cref{eq:shared-proof-3} from the proof of \Cref{thm:lower-high-pres-shared}, we again have
    \begin{equation}
        nt \geq \Omega\left(\min\left\{\frac{\ell}{t\eps^2 \|\mc{V}^\dag S_{\Tilde{\mc{H}}} \mc{V}\|_2}, \frac{1}{t^2 \eps^3} \right\}\right). \label{eq:private-proof-1}
    \end{equation}
    Note that due to \Cref{fact:copies-divergence}, in the case of private measurements, we can bound the max-min divergence, i.e., we can adversarially pick the basis $V_1, \dots, V_{d^2}$ depending on the individual measurements performed by the tester. To obtain the strongest lower bound possible, we wish to minimize $\|\mc{V}^\dag S_{\Tilde{\mc{H}}} \mc{V}\|_2$, and so we will pick $\mc{V}$ to consist of the $\ell$ smallest eigenvectors of $S_{\Tilde{\mc{H}}}$. Let the corresponding eigenvalues of $S_{\Tilde{\mc{H}}}$ be $\lambda_1, \dots, \lambda_\ell$. Recall that the Liouville representation of a channel has the same spectrum as the channel itself. Further, note that for each $1 \leq i \leq \ell$,
    \begin{equation}
        \lambda_i \leq \frac{\sum_{j = \ell+1}^{d^2-1} \lambda_j}{d^2 - 1 - \ell} \leq \frac{\tr(\tilde{\mc{H}}) - 1}{d^2- 1 - \ell};
    \end{equation}
    note that the summation intentionally omits $\lambda_{d^2} =1$. Now, recall that we can assume the measurements are Schur-refined without any loss of generality; consequently, we can apply \Cref{lem:reduced-channel-trace} above and obtain 
    \begin{equation}
        \lambda_i \leq \frac{4d}{\sqrt{t}(d^2-\ell)}.
    \end{equation} 
    Now, for our adversarial choice of $\mc{V}$,
    \begin{equation}
        \|\mc{V}^\dag S_{\Tilde{\mc{H}}} \mc{V}\|_2 = \sqrt{\sum_{i = 1}^\ell \lambda_i^2} \leq \sqrt{\ell} \cdot \lambda_\ell \leq \frac{4d}{\sqrt{t}} \cdot \frac{\sqrt{l}}{d^2-\ell} \leq \frac{4\sqrt{2}}{\sqrt{t}},
    \end{equation}
    where we set $\ell = d^2/2$ in the last step (see \Cref{def:rademacher-hard-instance}). Now, for this choice of $\ell$ and using $\|\mc{V}^\dag S_{\Tilde{\mc{H}}} \mc{V}\|_2 \leq 4\sqrt{2}/\sqrt{t}$, \Cref{eq:private-proof-1} yields
    \begin{equation}
        nt \geq \Omega\left(\min\left\{\frac{d^2}{\sqrt{t}\eps^2}, \frac{1}{t^2 \eps^3} \right\}\right),
    \end{equation}
    as desired.
\end{proof}

\appendix
\section{Batching the tester from \cite{buadescu2019quantum}}
\label{sec:batching-BOW-tester}

In this section, we present the folklore $t$-copy tester obtained by batching the tester of \cite{buadescu2019quantum}.

\begin{proposition}
\label{prop:tcopy-hs-closeness}
    Let $d \geq 2, t \geq 1, \eps > 0$. There is an algorithm using fixed $t$-copy measurements which, given $m$ copies each of states $\rho,\sigma \in \mbb{C}^{d \times d}$, can distinguish between $\|\rho - \sigma\|_2 \leq \eps/\sqrt{2}$ and $\|\rho - \sigma\|_2 \geq \eps$ with high probability when $m = \bigo(\frac{1}{\eps^2} + \frac{1}{t\eps^4})$.
\end{proposition}

By \Cref{fact:hs-to-trace-norm-testing}, we have the following corollary:

\begin{corollary}
    Let $d \geq 2, t \geq 1, \eps > 0$. The copy complexity of $\eps$-state certification with respect to the trace norm using fixed $t$-copy measurements is at most $\bigo(\max\{\frac{d}{\eps^2}, \frac{d^2}{t\eps^4}\})$.
\end{corollary}

\begin{proof}[Proof of Proposition \ref{prop:tcopy-hs-closeness}]
    \cite[Proposition 5.6]{buadescu2019quantum} construct an unbiased estimator $\hat{z}$ for $\|\rho - \sigma\|_2^2$, that given $t$ copies of each state, satisfies
    \begin{equation}
        \mbb{V}[\hat{z}] = \bigo \left(\frac{1}{t^2} + \frac{\|\rho - \sigma\|_2^2}{t}\right).
    \end{equation}
    Take $m$ copies of each state, divide them into $n = m/t$ batches, and construct independent estimators $\hat{z}_1, \dots, \hat{z}_n$ using \cite[Proposition 5.6]{buadescu2019quantum}. Let $\hat{z} = \frac1n \sum_{i = 1}^n \hat{z}_i$. Then, $\mbb{E}[\hat{z}] = \|\rho - \sigma\|_2^2$ and 
    \begin{equation}
        \mbb{V}[\hat{z}] = \bigo \left(\frac{1}{nt^2} + \frac{\|\rho - \sigma\|_2^2}{nt}\right).
    \end{equation}
    We test for the two cases by checking whether $\hat{z} \leq 3\eps^2/4$ or not. Let $n = \frac{c_1}{t\eps^2} + \frac{c_2}{t^2\eps^4}$ for absolute constants $c_1,c_2 > 0$. In the first case, $\mbb{E}[\hat{z}] = \|\rho-\sigma\|_2^2 \leq \eps^2/2$, and thus
    \begin{equation}
        \mbb{V}[\hat{z}] \leq \bigo \left(
            \frac{\eps^4}{c_2} + \frac{\eps^4}{c_1}
        \right) \implies \mathrm{Pr}[\hat{z} \geq 3\eps^2/4] \leq \frac{16}{\eps^4} \cdot \bigo \left(
            \frac{\eps^4}{c_2} + \frac{\eps^4}{c_1}
        \right) \leq 1/3,
    \end{equation}
    where the last inequality holds for sufficiently large $c_1,c_2$. In the second case, let $\mu \triangleq \mbb{E}[\hat{z}] = \|\rho-\sigma\|_2^2 \geq \eps^2$. Then, $n \geq \frac{c_1}{t\mu}$ and $n \geq \frac{c_2}{t^2\mu^2}$. In this case,
    \begin{equation}
        \mbb{V}[\hat{z}] \leq \bigo \left(
            \frac{\mu^2}{c_2} + \frac{\mu^2}{c_1}
        \right) \implies \mathrm{Pr}[\hat{z} \leq 3\eps^2/4] \leq \mathrm{Pr}[\hat{z} \leq 3\mu/4] \leq \frac{16}{\mu^2} \cdot \bigo \left(
            \frac{\mu^2}{c_2} + \frac{\mu^2}{c_1}
        \right) \leq 1/3,
    \end{equation}
    where the last inequality again holds for sufficiently large $c_1,c_2$. Thus, in either case, it suffices to take $m = nt \leq \bigo\left(\frac{1}{\eps^2} + \frac{1}{t\eps^4}\right)$ copies, as claimed.
\end{proof}

\section{Variance bounds for mixedness testing}
\label{sec:appendix-variance-bounds}

We now prove \Cref{lem:quasi-purification-variance-case-2,lem:quasi-purification-variance-case-3}.

\begin{proof}[Proof of \Cref{lem:quasi-purification-variance-case-2}]
    As the exact second moment of $\hat{\rho}_\lambda$ does not have a nice closed-form expression, we will first bound this quantity conditioned on measuring a particular partition $\lambda$ in \Cref{alg:quasi-purification-estimation}. For ease of notation, let $\ell \triangleq \ell(\lambda)$. Using \Cref{lem:quasi-purification-conditioned-second-moment}, we get
    \begin{align}
        &\mbb{E}[\tr(\hat{\rho}_\lambda^{\otimes 2} \rho^{\otimes 2}) | \lambda] = \tr(\mbb{E}[\hat{\rho}_\lambda^{\otimes 2} | \lambda] \cdot \rho^{\otimes 2})
        \\&= \frac{(t-1)(D+t)}{t(D+t+1)} \tr((\tau_\lambda)_{\msf{A}_1,\msf{A}_2} \rho^{\otimes 2}) + \frac{D+t}{t(D+t+1)}\cdot 2\tr(\rho^2(\tau_\lambda)_{\msf{A}_1})\nonumber\\&\hspace{1em}-\frac{\ell}{t(D+t+1)}\cdot 2\tr(\rho(\tau_\lambda)_{\msf{A}_1}) + \frac{\ell(D+t)}{t^2(D+t+1)}\cdot \tr(\rho^2) - \frac{\ell^2}{t^2(D+t+1)}
        \\&= \frac{(t-1)(D+t)}{t(D+t+1)} \tr((\tau_\lambda)_{\msf{A}_1,\msf{A}_2} \rho^{\otimes 2}) + \frac{D+t}{t(D+t+1)}\cdot 2\tr\left(\left(\frac{I}{d^2} + \frac{2\Delta}{d} + \Delta^2\right)(\tau_\lambda)_{\msf{A}_1}\right)\nonumber\\&\hspace{1em}-\frac{\ell}{t(D+t+1)}\cdot 2\tr\left(\left(\frac{I}{d} + \Delta\right)(\tau_\lambda)_{\msf{A}_1}\right) + \frac{\ell(D+t)}{t^2(D+t+1)}\cdot \left(\frac{1}{d} + \tr(\Delta^2)\right) - \frac{\ell^2}{t^2(D+t+1)}
        \\&= \frac{(t-1)(D+t)}{t(D+t+1)} \tr((\tau_\lambda)_{\msf{A}_1,\msf{A}_2} \rho^{\otimes 2}) + \frac{D+t}{t(D+t+1)}\cdot 2\tr\left(\left(\frac{I}{d^2} + \frac{2\Delta}{d} + \Delta^2\right)(\tau_\lambda)_{\msf{A}_1}\right)\nonumber\\&\hspace{1em}+
        \frac{\ell}{t^2(D+t+1)}\cdot \left(
         -\frac{2t}{d} - 2t\cdot\tr(\Delta(\tau_\lambda)_{\msf{A}_1}) + \ell + \frac{t}{d} + (D+t)\tr(\Delta^2) - \ell
        \right)
        \\&= \frac{(t-1)(D+t)}{t(D+t+1)} \tr((\tau_\lambda)_{\msf{A}_1,\msf{A}_2} \rho^{\otimes 2}) +\frac{2t}{t(D+t+1)}\left(\frac{1}{d^2} + \frac{2\tr(\Delta(\tau_\lambda)_{\msf{A}_1})}{d} + \tr(\Delta^2(\tau_\lambda)_{\msf{A}_1})\right) \nonumber\\&\hspace{1em} +\frac{\ell}{t^2(D+t+1)}\cdot \left((D+t)\tr(\Delta^2) - \frac{t}{d} - 2t\cdot\tr(\Delta(\tau_\lambda)_{\msf{A}_1}) + \frac{2t}{d} + 4t\cdot \tr(\Delta(\tau_\lambda)_{\msf{A}_1}) +  2dt\cdot\tr(\Delta^2(\tau_\lambda)_{\msf{A}_1})\right)
        \\&= \frac{(t-1)(D+t)}{t(D+t+1)} \tr((\tau_\lambda)_{\msf{A}_1,\msf{A}_2} \rho^{\otimes 2}) +\frac{2t}{t(D+t+1)}\left(\frac{1}{d^2} + \frac{2\tr(\Delta(\tau_\lambda)_{\msf{A}_1})}{d} + \tr(\Delta^2(\tau_\lambda)_{\msf{A}_1})\right) \nonumber\\&\hspace{1em}+ \frac{\ell}{t^2(D+t+1)}\cdot \left((D+t)\tr(\Delta^2) + \frac{t}{d} + 2t\cdot\tr(\Delta(\tau_\lambda)_{\msf{A}_1}) +  2dt\cdot\tr(\Delta^2(\tau_\lambda)_{\msf{A}_1})\right)
        \\&= \frac{(t-1)(D+t)}{t(D+t+1)} \tr((\tau_\lambda)_{\msf{A}_1,\msf{A}_2} \rho^{\otimes 2}) + \frac{1}{d^2}\cdot \frac{D+2t}{t(D+t+1)} + \frac{2\tr(\Delta(\tau_\lambda)_{\msf{A}_1})}{d} \cdot \frac{D+2t}{t(D+t+1)}\nonumber\\&\hspace{1em}+ \tr(\Delta^2(\tau_\lambda)_{\msf{A}_1})\cdot \frac{2(D+t)}{t(D+t+1)} + \frac{\ell(D+t)}{t^2(D+t+1)}\tr(\Delta^2)
        \\&= \frac{(t-1)(D+t)}{t(D+t+1)} \left(\frac{1}{d^2} + \tr((\tau_\lambda)_{\msf{A}_1,\msf{A}_2}(\Delta \otimes d^{-1}I + d^{-1}I\otimes \Delta + \Delta^{\otimes 2}))\right) + \frac{1}{d^2}\cdot \frac{D+2t}{t(D+t+1)}\nonumber\\&\hspace{1em}+ \frac{2\tr(\Delta(\tau_\lambda)_{\msf{A}_1})}{d} \cdot \frac{D+2t}{t(D+t+1)} + \tr(\Delta^2(\tau_\lambda)_{\msf{A}_1})\cdot \frac{2(D+t)}{t(D+t+1)} + \frac{\ell(D+t)}{t^2(D+t+1)}\tr(\Delta^2)
        \\&= \frac{1}{d^2}\cdot \frac{(t-1)(D+t) + D+2t}{t(D+t+1)} + \frac{2\tr(\Delta(\tau_\lambda)_{\msf{A}_1})}{d} \cdot \frac{(t-1)(D+t) + D+2t}{t(D+t+1)} \nonumber\\&\hspace{1em}+\frac{(t-1)(D+t)}{t(D+t+1)}\tr((\tau_\lambda)_{\msf{A}_1,\msf{A}_2}\Delta^{\otimes 2}) + \tr(\Delta^2(\tau_\lambda)_{\msf{A}_1})\cdot \frac{2(D+t)}{t(D+t+1)} + \frac{\ell(D+t)}{t^2(D+t+1)}\tr(\Delta^2)
        \\&= \frac{1}{d^2} + \frac{2\tr(\Delta(\tau_\lambda)_{\msf{A}_1})}{d} + \frac{t-1}{t}\left(1 - \frac{1}{D+t+1}\right)\tr((\tau_\lambda)_{\msf{A}_1,\msf{A}_2}\Delta^{\otimes 2}) \nonumber\\&\hspace{1em}+\tr(\Delta^2(\tau_\lambda)_{\msf{A}_1})\cdot \frac{2(D+t)}{t(D+t+1)} + \frac{\ell(D+t)}{t^2(D+t+1)}\tr(\Delta^2). \label{eq:quasi-purification-collision-case2-1}
    \end{align}
    We used \Cref{lem:quasi-purification-conditioned-second-moment} in the second equality, and in the third equality we substituted $\rho = I/d+\Delta$. The next four steps involved carefully rearranging the terms therein, followed by again substituting in $\rho = I/d+\Delta$. In the third-to-last step, we used that $\tr(M_{1,2}(I \otimes N)) = \tr((M)_2 N)$ and that $\tau_\lambda$ is in the symmetric subspace. 

    Note that $D = d\cdot \ell(\lambda)$ is dependent on $\lambda$. So, to take the expectation of the above terms with respect to $\lambda$, we will first attempt to obtain an upper bound that does not have this $D$-dependence. First, consider the matrix $M \triangleq (\Delta \otimes I + I \otimes \Delta)^2$. As $\Delta,I$ are Hermitian, $M$ is positive semidefinite. Thus, we have
    \begin{align}
        0 \leq \tr( (\tau_\lambda)_{\msf{A}_1,\msf{A}_2} M) &= 2\tr((\tau_\lambda)_{\msf{A}_1,\msf{A}_2} \Delta^{\otimes 2}) + \tr((\tau_\lambda)_{\msf{A}_1,\msf{A}_2} \cdot \Delta^2 \otimes I) + \tr((\tau_\lambda)_{\msf{A}_1,\msf{A}_2} \cdot I \otimes \Delta^2) 
        \\&= 2\tr((\tau_\lambda)_{\msf{A}_1,\msf{A}_2} \Delta^{\otimes 2}) + 2 \tr((\tau_\lambda)_{\msf{A}_1}\Delta^2).
    \end{align}

    Now, adding $\frac{t-1}{2t(D+t+1)}\tr( (\tau_\lambda)_{\msf{A}_1,\msf{A}_2} M)$ to the RHS of \Cref{eq:quasi-purification-collision-case2-1}, we get
    \begin{align}
        \mbb{E}[\tr(\hat{\rho}_\lambda^{\otimes 2} \rho^{\otimes 2}) | \lambda] &\leq \frac{1}{d^2} + \frac{2\tr(\Delta(\tau_\lambda)_{\msf{A}_1})}{d} + \frac{t-1}{t}\tr((\tau_\lambda)_{\msf{A}_1,\msf{A}_2}\Delta^{\otimes 2}) + \frac{2D+3t-1}{t(D+t+1)} \tr(\Delta^2(\tau_\lambda)_{\msf{A}_1}) + \frac{\ell}{t^2}\tr(\Delta^2)
        \\&\leq \frac{1}{d^2} + \frac{2\tr(\Delta(\tau_\lambda)_{\msf{A}_1})}{d} + \frac{t-1}{t}\tr((\tau_\lambda)_{\msf{A}_1,\msf{A}_2}\Delta^{\otimes 2}) +\frac{3}{t}\tr(\Delta^2(\tau_\lambda)_{\msf{A}_1}) + \frac{\ell}{t^2}\tr(\Delta^2),
    \end{align} 
    where we also used that $\tr(\Delta^2(\tau_\lambda)_{A_1})$ and $\tr(\Delta^2)$ are non-negative, as $\Delta^2$ and $\tau_\lambda$ are both psd. Finally, averaging over $\lambda$, we have
    \begin{align}
        \mbb{E}[\tr(\hat{\rho}_\lambda^{\otimes 2} \rho^{\otimes 2})] &\leq \frac{1}{d^2} + \frac{2\tr(\Delta\mbb{E}_\lambda[(\tau_\lambda)_{\msf{A}_1}])}{d} + \frac{t-1}{t}\tr(\mbb{E}_\lambda[(\tau_\lambda)_{\msf{A}_1,\msf{A}_2}]\Delta^{\otimes 2}) + \frac3t\tr(\Delta^2\mbb{E}_\lambda[(\tau_\lambda)_{\msf{A}_1}]) + \frac{\mbb{E}_\lambda[\ell(\lambda)]}{t^2}\tr(\Delta^2)
        \\&= \frac{1}{d^2} + \frac{2\tr(\Delta \rho)}{d} + \frac{t-1}{t}\tr(\rho^{\otimes 2}\Delta^{\otimes 2}) + \frac3t\tr(\Delta^2\rho) + \frac{\mbb{E}_\lambda[\ell(\lambda)]}{t^2}\tr(\Delta^2)
        \\&\leq \tr(\rho^2)^2 + \frac{3}{dt}\tr(\Delta^2) + \frac3t\tr(\Delta^3) + \frac{\mbb{E}_\lambda[\ell(\lambda)]}{t^2}\tr(\Delta^2)
        \\&\leq \tr(\rho^2)^2 + \frac{3}{dt}\tr(\Delta^2) + \frac3t\tr(\Delta^2)^{3/2} + \frac{\mbb{E}_\lambda[\ell(\lambda)]}{t^2}\tr(\Delta^2),
    \end{align}
    where the second line used \Cref{lem:partial-trace-tau-lambda} and that $\tr(\rho^{\otimes 2} \Delta^{\otimes 2}) = \tr(\Delta^2)^2$, and the last step used $\tr(\Delta^3) = \|\Delta\|_3^3 \leq \|\Delta\|_2^3 = \tr(\Delta^2)^{3/2}$.
\end{proof}

\begin{proof}[Proof of \Cref{lem:quasi-purification-variance-case-3}]
    By \Cref{thm:quasi-purification-truncated-second-moment}, we have
    \begin{align}
        \tr(\mbb{E}[\hat{\rho}_\lambda^{\otimes 2}]^2) &= \left\langle \mbb{E}[\hat{\rho}_\lambda^{\otimes 2}], \frac{t-1}{t}\rho^{\otimes 2} + \frac1t(\rho \otimes I + I \otimes \rho)\cdot \swap + \frac{\mbb{E}[\ell(\lambda)]}{t^2}\swap - \mathrm{Lower}_{\rho}\right\rangle.
        \\&\leq \left\langle \mbb{E}[\hat{\rho}_\lambda^{\otimes 2}], \frac{t-1}{t}\rho^{\otimes 2} + \frac1t(\rho \otimes I + I \otimes \rho)\cdot \swap + \frac{\mbb{E}[\ell(\lambda)]}{t^2}\swap \right\rangle,
        \label{eq:quasi-purification-mixedness-variance-case3-1}
    \end{align}
    where we used that $\tr(\hat{\rho}_\lambda^{\otimes 2} \mrm{Lower}_\rho) \geq 0$ as $\mrm{Lower}_\rho \in \mrm{SoS}(d)$.
    We will individually compute each term in the inner product above. By \Cref{lem:quasi-purification-variance-case-2}, we have
    \begin{align}
        \left\langle \mbb{E}[\hat{\rho}_\lambda^{\otimes 2}], \rho^{\otimes 2}\right\rangle \leq \tr(\rho^2)^2 + \bigo\left(\frac{\tr(\Delta^2)}{dt} + \frac{ \mbb{E}[\ell(\lambda)]}{t^2}\tr(\Delta^2) + \frac{\tr(\Delta^2)^{3/2}}{t}\right).
        \label{eq:quasi-purification-mixedness-variance-case3-2}
    \end{align}
    For ease of notation, let $M \triangleq (\rho \otimes I + I \otimes \rho)\cdot \swap$. Now, by \Cref{thm:quasi-purification-truncated-second-moment} and by noting that $\tr(\mathrm{Lower}_\rho \cdot M) \geq 0$,
    \begin{align}
    \left\langle \mbb{E}[\hat{\rho}_\lambda^{\otimes 2}], M \right\rangle &\leq \left\langle \frac{t-1}{t}\rho^{\otimes 2} + \frac1t \cdot M + \frac{\mbb{E}[\ell(\lambda)]}{t^2}\swap, M \right\rangle
    \\&= \frac{t-1}{t} \cdot 2\tr(\rho^3) + \frac{1}{t}(2d\tr(\rho^2) + 2) + \frac{\mbb{E}[\ell(\lambda)]}{t^2} \cdot 2d
    \\&= \frac{2t-2}{t}\left(\frac{1}{d^2} + \frac{3\tr(\Delta^2)}{d} + \tr(\Delta^3)\right) + \frac{4}{t} + \frac{2d \cdot \tr(\Delta^2)}{t} + \frac{2d \cdot \mbb{E}[\ell(\lambda)]}{t^2}
    \\&\leq \left(\frac{1}{d^2} + \frac{2\tr(\Delta^2)}{d} + \tr(\Delta^2)^2\right) + \left(\frac{1}{d^2} + \frac{4}{t} + \frac{2d \cdot \mbb{E}[\ell(\lambda)]}{t^2}\right) + \tr(\Delta^2)\left(\frac{4}{d} + \frac{2d}{t}\right) + 2\tr(\Delta^2)^{3/2}
    \\&= \tr(\rho^2)^2 + \bigo\left(\frac{1}{d^2} + \frac{1}{t} + \frac{d \cdot \mbb{E}[\ell(\lambda)]}{t^2} + \frac{\tr(\Delta^2)}{d} + \frac{d\cdot\tr(\Delta^2)}{t} + \tr(\Delta^2)^{3/2}\right).
        \label{eq:quasi-purification-mixedness-variance-case3-3}
    \end{align}
    For the third term, we have
    \begin{align}
    \left\langle \mbb{E}[\hat{\rho}_\lambda^{\otimes 2}], \swap \right\rangle &\leq \left\langle \frac{t-1}{t}\rho^{\otimes 2} + \frac1t \cdot M + \frac{\mbb{E}[\ell(\lambda)]}{t^2}\swap, \swap \right\rangle
    \\&= \frac{t-1}{t} \tr(\rho^2) + \frac{2d}{t} + \frac{d^2\mbb{E}[\ell(\lambda)]}{t^2}
    \\&\leq \frac{1}{d} + \tr(\Delta^2) + \frac{2d}{t} + \frac{d^2\mbb{E}[\ell(\lambda)]}{t^2}.
        \label{eq:quasi-purification-mixedness-variance-case3-4}
    \end{align}
    The statement of the lemma is finally obtained by substituting \Cref{eq:quasi-purification-mixedness-variance-case3-2,eq:quasi-purification-mixedness-variance-case3-3,eq:quasi-purification-mixedness-variance-case3-4} into \Cref{eq:quasi-purification-mixedness-variance-case3-1}.
\end{proof}

\section{Proof of \Cref{claim:operator-norm}}
\label{app:claim}

Before proving the claim, we state some necessary representation-theoretic notation. Each partition $\lambda \vdash t$ corresponds to a \emph{Young diagram}, i.e., a sequence of rows of boxes with the $i$th row containing $\lambda_i$ boxes. The shape of the Young diagram refers to the corresponding partition. A Young tableau is a Young diagram with the boxes filled with entries from $[t]$; for our purposes, each element of $[t]$ occurs exactly once. Now, given such a tableau $T$, we define its row subgroup $R_T \subseteq S_t$ to be the group of permutations that leave the rows invariant, and similarly define the column subgroup $C_t \subseteq S_t$. For a permutation $\pi \in S_t$, we will represent its unitary action on the tensor factors of $(\mathbb{C}^d)^{\otimes t}$ by $U_\pi$. 

Next, we define the operators 
\begin{equation}
    a_T \triangleq \sum_{\pi \in R_T} U_\pi, \quad b_T \triangleq \sum_{\pi \in C_T} \mathrm{sgn}(\pi) U_\pi, \quad c_T \triangleq a_T b_T,
\end{equation}
where $c_T$ is the representation of the Young symmetrizer corresponding to $T$. It is a standard fact that the projectors $P_\lambda$ can be decomposed as a linear combination of all corresponding Young symmetrizers. In particular, for some scalar $\kappa_\lambda$, we can write 
\begin{equation}
\label{eq:isotypic-symmetrizer-decomposition}
    P_\lambda = \kappa_\lambda \sum_{T : \mathrm{shape}(T) = \lambda} c_T;
\end{equation} 
see e.g. \cite[Proposition 9.3.14 and Theorem 9.3.15]{goodman2009symmetry}. 

Let $e_1, \dots, e_d$ be the computational basis vectors of $\mathbb{C}^d$. For any unit vector $u \in \mathbb{C}^d$, we define the operator $N_u \in \mathbb{C}^{d^t \times d^t}$:
\begin{equation}
    N_u \triangleq \sum_{j = 1}^t \ket{u}\bra{u}_j \otimes I^{\otimes [t] \setminus \{j\}}.
\end{equation}
Let $U$ be the unitary such that $u = U e_1$. Then, note that we can write $N_u = U^{\otimes t} N_{e_1} U^{\dag \otimes t}$. Further, for $m \in \{0, \dots, t\}$, let $Q_m$ be the projector onto the span of $t$-fold computational basis tensors containing exactly $m$ copies of $e_1$. Consequently, we have 
\begin{equation}
\label{eq:occupation-operator}
    N_{e_1} = \sum_{m = 0}^t m Q_m \quad \textnormal{and} \quad \sum_m Q_m = I.
\end{equation}
With all notation in place, we can finally prove the necessary claim.

\begin{proof}[Proof of \Cref{claim:operator-norm}]
    We will aim to show that, for every unit vector $u \in \mathbb{C}^d$,
    \begin{equation}
    \label{eq:claim-key}
        P_\lambda N_u P_\lambda \preceq \lambda_1 P_\lambda. 
    \end{equation}
    This would immediately imply the desired operator norm bound, as, for all unit vectors $u$, we would have
    \begin{align}
        \bra{u} \overline{\rho}_y \ket{u} &= \frac1t \sum_{j = 1}^t \tr(\ket{u}\bra{u}_j \otimes I^{\otimes [t] \setminus \{j\}} \cdot M_y) \cdot \tr(M_y)^{-1}
        \\&= \frac{\tr(N_u M_y)}{t \cdot \tr(M_y)} = \frac{\tr(P_\lambda N_u P_\lambda M_y)}{t \cdot \tr(M_y)} \leq \frac{\lambda_1}{t} \frac{\tr(P_\lambda M_y)}{\tr(M_y)} = \frac{\lambda_1}{t},
    \end{align}
    where the third and final equalities used $M_y = P_\lambda M_y P_\lambda$.

    Now, to prove \Cref{eq:claim-key}, it suffices to show that
    \begin{equation}
    \label{eq:claim-key-2}
        P_\lambda N_{e_1} P_\lambda \preceq \lambda_1 P_\lambda,
    \end{equation}
    as the bound extends to all unit vectors by unitary covariance.

    To achieve this, we show that for all $m > \lambda_1$, it holds that $P_\lambda Q_m = 0$. Equivalently, $P_\lambda \ket{i_1, \dots, i_t} = 0$ whenever the multiset $\{i_1, \dots, i_t\}$ contains more than $\lambda_1$ occurrences of the index $1$. 
    
    To see this, fix a tableau $T$ of shape $\lambda$. The Young diagram has exactly $\lambda_1$ columns; as $m > \lambda_1$, two positions $r,s \in [t]$ carrying index $1$ in $ \ket{i_1, \dots, i_t}$ must lie in the same column of $T$. Consequently, the transposition $\tau = (r \ s)$ leaves the columns of $T$ invariant, i.e., $\tau \in C_T$. Moreover, we have 
    \begin{equation}
        U_\tau \ket{i_1, \dots, i_t} = \ket{i_1, \dots, i_t}.
    \end{equation}
    As $b_T$ antisymmetrizes over the column subgroup, we also have
    \begin{equation}
        b_T U_\tau = - b_T.
    \end{equation}
    Together, we obtain
    \begin{equation}
        b_T \ket{i_1, \dots, i_T} = b_T U_\tau \ket{i_1, \dots, i_T} = -b_T \ket{i_1, \dots, i_T} \implies b_T \ket{i_1, \dots, i_T} = 0.
    \end{equation}
    As $c_T = a_T b_T$, we also obtain $c_T \ket{i_1, \dots, i_T} = 0$.
    Using \Cref{eq:isotypic-symmetrizer-decomposition}, we can sum over all tableaux with shape $\lambda$ to obtain
    \begin{equation}
        P_\lambda \ket{i_1, \dots, i_t} = 0,
    \end{equation}
    implying that $P_\lambda Q_m = 0$ whenever $m > \lambda_1$. Finally, we use both items of \Cref{eq:occupation-operator} to write
    \begin{align}
        P_\lambda N_{e_1} P_\lambda &= \sum_{m = 0}^t m P_\lambda Q_m P_\lambda
        \\& \preceq \lambda_1 \sum_{m = 0}^{\lambda_1} P_\lambda Q_m P_\lambda
        \\&= \lambda_1 P_\lambda \sum_{m = 0}^t Q_m P_\lambda = \lambda_1 P_\lambda,
    \end{align}
    concluding the proof.    
\end{proof}

\bibliographystyle{alpha}
\bibliography{biblio}

@article{chen2026quantum,
  title={Quantum channel tomography: optimal bounds and a Heisenberg-to-classical phase transition},
  author={Chen, Kean and Girardi, Filippo and Oufkir, Aadil and Yu, Nengkun and Zhang, Zhicheng},
  journal={arXiv preprint arXiv:2604.17369},
  year={2026}
}

@article{chen2026local,
  title={Local test for unitarily invariant properties of bipartite quantum states},
  author={Chen, Kean and Wang, Qisheng and Zhang, Zhicheng},
  journal={IEEE Transactions on Information Theory},
  year={2026},
  publisher={IEEE}
}

@article{chen2024optimalHighPrecisionShadow,
  title={Optimal high-precision shadow estimation},
  author={Chen, Sitan and Li, Jerry and Liu, Allen},
  journal={arXiv preprint arXiv:2407.13874},
  year={2024}
}

@article{anshu2024survey,
  title={A survey on the complexity of learning quantum states},
  author={Anshu, Anurag and Arunachalam, Srinivasan},
  journal={Nature Reviews Physics},
  volume={6},
  number={1},
  pages={59--69},
  year={2024},
  publisher={Nature Publishing Group UK London}
}

@article{noller2025infinite,
  title={An infinite hierarchy of multi-copy quantum learning tasks},
  author={N{\"o}ller, Jan and Tran, Viet T and Gachechiladze, Mariami and Kueng, Richard},
  journal={arXiv preprint arXiv:2510.08070},
  year={2025}
}

@article{ye2025exponential,
  title={Exponential Advantage from One More Replica in Estimating Nonlinear Properties of Quantum States},
  author={Ye, Qi and Liu, Zhenhuan and Deng, Dong-Ling},
  journal={arXiv preprint arXiv:2509.24000},
  year={2025}
}

@inproceedings{arunachalam2025generalized,
  title={Generalized Inner Product Estimation with Limited Quantum Communication},
  author={Arunachalam, Srinivasan and Schatzki, Louis},
  booktitle={42nd International Symposium on Theoretical Aspects of Computer Science (STACS 2025)},
  pages={11--1},
  year={2025},
  organization={Schloss Dagstuhl--Leibniz-Zentrum f{\"u}r Informatik}
}

@inproceedings{king2025triply,
  title={Triply efficient shadow tomography},
  author={King, Robbie and Gosset, David and Kothari, Robin and Babbush, Ryan},
  booktitle={Proceedings of the 2025 Annual ACM-SIAM Symposium on Discrete Algorithms (SODA)},
  pages={914--946},
  year={2025},
  organization={SIAM}
}

@inproceedings{aaronson2018shadow,
  title={Shadow tomography of quantum states},
  author={Aaronson, Scott},
  booktitle={Proceedings of the 50th annual ACM SIGACT symposium on theory of computing},
  pages={325--338},
  year={2018}
}

@article{coladangelo2026power,
  title={The Power of Two Bases: Robust and copy-optimal certification of nearly all quantum states with few-qubit measurements},
  author={Coladangelo, Andrea and Li, Jerry and Slote, Joseph and Wu, Ellen},
  journal={arXiv preprint arXiv:2602.11616},
  year={2026}
}

@article{mele2025optimal,
  title={Optimal learning of quantum channels in diamond distance},
  author={Mele, Antonio Anna and Bittel, Lennart},
  journal={arXiv preprint arXiv:2512.10214},
  year={2025}
}

@article{girardi2025randomPurificationSimplified,
  title={Random purification channel made simple},
  author={Girardi, Filippo and Mele, Francesco Anna and Lami, Ludovico},
  journal={arXiv preprint arXiv:2511.23451},
  year={2025}
}

@article{mele2025randomPurificationGaussianBosons,
  title={Random purification channel for passive Gaussian bosons},
  author={Mele, Francesco Anna and Girardi, Filippo and Chen, Senrui and Fanizza, Marco and Lami, Ludovico},
  journal={arXiv preprint arXiv:2512.16878},
  year={2025}
}

@article{walter2025randomPurificationGeneric,
  title={A random purification channel for arbitrary symmetries with applications to fermions and bosons},
  author={Walter, Michael and Witteveen, Freek},
  journal={arXiv preprint arXiv:2512.15690},
  year={2025}
}

@article{girardi2025randomStinespring,
  title={Random Stinespring superchannel: converting channel queries into dilation isometry queries},
  author={Girardi, Filippo and Mele, Francesco Anna and Zhao, Haimeng and Fanizza, Marco and Lami, Ludovico},
  journal={arXiv preprint arXiv:2512.20599},
  year={2025}
}

@article{yoshida2025randomDilation,
  title={Random dilation superchannel},
  author={Yoshida, Satoshi and Niwa, Ryotaro and Murao, Mio},
  journal={arXiv preprint arXiv:2512.21260},
  year={2025}
}

@article{acharya2025pauliSingleQubit,
  title={Pauli Measurements Are Near-Optimal for Single-Qubit Tomography},
  author={Acharya, Jayadev and Dharmavarapu, Abhilash and Liu, Yuhan and Yu, Nengkun},
  journal={arXiv preprint arXiv:2507.22001},
  year={2025}
}

@inproceedings{acharya2025pauliSingleCopy,
  title={Pauli measurements are not optimal for single-copy tomography},
  author={Acharya, Jayadev and Dharmavarapu, Abhilash and Liu, Yuhan and Yu, Nengkun},
  booktitle={Proceedings of the 57th Annual ACM Symposium on Theory of Computing},
  pages={718--729},
  year={2025}
}

@inproceedings{o2017efficient,
  title={Efficient quantum tomography II},
  author={O'Donnell, Ryan and Wright, John},
  booktitle={Proceedings of the 49th Annual ACM SIGACT Symposium on Theory of Computing},
  pages={962--974},
  year={2017}
}

@article{de2025non,
  title={Non-iid hypothesis testing: from classical to quantum},
  author={De Palma, Giacomo and Fanizza, Marco and Mowry, Connor and O'Donnell, Ryan},
  journal={arXiv preprint arXiv:2510.06147},
  year={2025}
}

@inproceedings{aliakbarpour2025adversarially,
  author={Aliakbarpour, Maryam and Braverman, Vladimir and Chia, Nai-Hui and Liu, Yuhan},
  booktitle={2025 IEEE 66th Annual Symposium on Foundations of Computer Science (FOCS)}, 
  title={Adversarially Robust Quantum State Learning and Testing}, 
  year={2025},
  volume={},
  number={},
  pages={186-204},
  doi={10.1109/FOCS63196.2025.00013}
}

@article{grewal2026pauli,
  title={Pauli Measurements Are Near-Optimal for Pure State Tomography},
  author={Grewal, Sabee and Gupta, Meghal and He, William and Sen, Aniruddha and Singhal, Mihir},
  journal={arXiv preprint arXiv:2601.04444},
  year={2026}
}

@article{acharya2020estimating,
  title={Estimating quantum entropy},
  author={Acharya, Jayadev and Issa, Ibrahim and Shende, Nirmal V and Wagner, Aaron B},
  journal={IEEE Journal on Selected Areas in Information Theory},
  volume={1},
  number={2},
  pages={454--468},
  year={2020},
  publisher={IEEE}
}

@article{harrow2013church,
  title={The church of the symmetric subspace},
  author={Harrow, Aram W},
  journal={arXiv preprint arXiv:1308.6595},
  year={2013}
}

@article{liu2024quantum,
  title={Quantum state testing with restricted measurements},
  author={Liu, Yuhan and Acharya, Jayadev},
  journal={arXiv preprint arXiv:2408.17439},
  year={2024}
}

@article{acharya2020inference,
  title={Inference under information constraints I: Lower bounds from chi-square contraction},
  author={Acharya, Jayadev and Canonne, Cl{\'e}ment L and Tyagi, Himanshu},
  journal={IEEE Transactions on Information Theory},
  volume={66},
  number={12},
  pages={7835--7855},
  year={2020},
  publisher={IEEE}
}

@article{grier2024sample,
  title={Sample-optimal classical shadows for pure states},
  author={Grier, Daniel and Pashayan, Hakop and Schaeffer, Luke},
  journal={Quantum},
  volume={8},
  pages={1373},
  year={2024},
  publisher={Verein zur F{\"o}rderung des Open Access Publizierens in den Quantenwissenschaften}
}

@article{hayashi1998asymptotic,
  title={Asymptotic estimation theory for a finite-dimensional pure state model},
  author={Hayashi, Masahito},
  journal={Journal of Physics A: Mathematical and General},
  volume={31},
  number={20},
  pages={4633},
  year={1998},
  publisher={IOP Publishing}
}

@article{tang2025conjugate,
  title={Conjugate queries can help},
  author={Tang, Ewin and Wright, John and Zhandry, Mark},
  journal={arXiv preprint arXiv:2510.07622},
  year={2025}
}

@article{pelecanos2025mixed,
  title={Mixed state tomography reduces to pure state tomography},
  author={Pelecanos, Angelos and Spilecki, Jack and Tang, Ewin and Wright, John},
  journal={arXiv preprint arXiv:2511.15806},
  year={2025}
}

@article{pelecanos2025debiased,
  title={The debiased Keyl's algorithm: a new unbiased estimator for full state tomography},
  author={Pelecanos, Angelos and Spilecki, Jack and Wright, John},
  journal={arXiv preprint arXiv:2510.07788},
  year={2025}
}

@inproceedings{liu2024role,
  title={The role of randomness in quantum state certification with unentangled measurements},
  author={Liu, Yuhan and Acharya, Jayadev},
  booktitle={The Thirty Seventh Annual Conference on Learning Theory},
  pages={3523--3555},
  year={2024},
  organization={PMLR}
}

@phdthesis{wright2016learn,
  title={How to learn a quantum state},
  author={Wright, John},
  year={2016},
  school={Carnegie Mellon University}
}

@inproceedings{o2025instance,
  title={Instance-optimal quantum state certification with entangled measurements},
  author={O'Donnell, Ryan and Wadhwa, Chirag},
  booktitle={Proceedings of the 58th Annual ACM Symposium on Theory of Computing},
  pages={398--409},
  year={2026}
}

@article{canonne2022topics,
  title={Topics and techniques in distribution testing: A biased but representative sample},
  author={Canonne, Cl{\'e}ment},
  journal={Foundations and Trends{\textregistered} in Communications and Information Theory},
  volume={19},
  number={6},
  pages={1032--1198},
  year={2022},
  publisher={Now Publishers, Inc.}
}

@article{flammia2024quantum,
  title={Quantum chi-squared tomography and mutual information testing},
  author={Flammia, Steven T and O'Donnell, Ryan},
  journal={Quantum},
  volume={8},
  pages={1381},
  year={2024},
  publisher={Verein zur F{\"o}rderung des Open Access Publizierens in den Quantenwissenschaften}
}

@inproceedings{chen2022toward,
  title={Toward instance-optimal state certification with incoherent measurements},
  author={Chen, Sitan and Li, Jerry and O’Donnell, Ryan},
  booktitle={Conference on Learning Theory},
  pages={2541--2596},
  year={2022},
  organization={PMLR}
}

@inproceedings{chen2022tightStateCertification,
  title={Tight bounds for quantum state certification with incoherent measurements},
  author={Chen, Sitan and Li, Jerry and Huang, Brice and Liu, Allen},
  booktitle={2022 IEEE 63rd Annual Symposium on Foundations of Computer Science (FOCS)},
  pages={1205--1213},
  year={2022},
  organization={IEEE}
}

@inproceedings{anshu2022distributed,
  title={Distributed quantum inner product estimation},
  author={Anshu, Anurag and Landau, Zeph and Liu, Yunchao},
  booktitle={Proceedings of the 54th Annual ACM SIGACT Symposium on Theory of Computing},
  pages={44--51},
  year={2022}
}

@article{gong2024sample,
  title={On the sample complexity of purity and inner product estimation},
  author={Gong, Weiyuan and Haferkamp, Jonas and Ye, Qi and Zhang, Zhihan},
  journal={arXiv preprint arXiv:2410.12712},
  year={2024}
}

@article{kueng2017low,
  title={Low rank matrix recovery from rank one measurements},
  author={Kueng, Richard and Rauhut, Holger and Terstiege, Ulrich},
  journal={Applied and Computational Harmonic Analysis},
  volume={42},
  number={1},
  pages={88--116},
  year={2017},
  publisher={Elsevier}
}

@inproceedings{haah2016sample,
  title={Sample-optimal tomography of quantum states},
  author={Haah, Jeongwan and Harrow, Aram W and Ji, Zhengfeng and Wu, Xiaodi and Yu, Nengkun},
  booktitle={Proceedings of the forty-eighth annual ACM symposium on Theory of Computing},
  pages={913--925},
  year={2016}
}

@inproceedings{o2016efficient,
  title={Efficient quantum tomography},
  author={O'Donnell, Ryan and Wright, John},
  booktitle={Proceedings of the forty-eighth annual ACM symposium on Theory of Computing},
  pages={899--912},
  year={2016}
}

@article{montanaro2016survey,
  title={A Survey of Quantum Property Testing},
  author={Montanaro, Ashley and de Wolf, Ronald},
  journal={Theory of Computing},
  pages={1--81},
  year={2016},
  publisher={Theory of Computing Exchange}
}

@inproceedings{buadescu2019quantum,
  title={Quantum state certification},
  author={B{\u{a}}descu, Costin and O'Donnell, Ryan and Wright, John},
  booktitle={Proceedings of the 51st Annual ACM SIGACT Symposium on Theory of Computing},
  pages={503--514},
  year={2019}
}

@inproceedings{o2015quantum,
  title={Quantum spectrum testing},
  author={O'Donnell, Ryan and Wright, John},
  booktitle={Proceedings of the forty-seventh annual ACM symposium on Theory of computing},
  pages={529--538},
  year={2015}
}

@inproceedings{chen2022exponential,
  title={Exponential separations between learning with and without quantum memory},
  author={Chen, Sitan and Cotler, Jordan and Huang, Hsin-Yuan and Li, Jerry},
  booktitle={2021 IEEE 62nd Annual Symposium on Foundations of Computer Science (FOCS)},
  pages={574--585},
  year={2022},
  organization={IEEE}
}

@article{huang2021information,
  title={Information-theoretic bounds on quantum advantage in machine learning},
  author={Huang, Hsin-Yuan and Kueng, Richard and Preskill, John},
  journal={Physical Review Letters},
  volume={126},
  number={19},
  pages={190505},
  year={2021},
  publisher={APS}
}

@article{chen2021hierarchy,
  title={A hierarchy for replica quantum advantage},
  author={Chen, Sitan and Cotler, Jordan and Huang, Hsin-Yuan and Li, Jerry},
  journal={arXiv preprint arXiv:2111.05874},
  year={2021}
}

@inproceedings{bubeck2020entanglement,
  title={Entanglement is necessary for optimal quantum property testing},
  author={Bubeck, Sebastien and Chen, Sitan and Li, Jerry},
  booktitle={2020 IEEE 61st Annual Symposium on Foundations of Computer Science (FOCS)},
  pages={692--703},
  year={2020},
  organization={IEEE}
}

@inproceedings{chen2023does,
  title={When does adaptivity help for quantum state learning?},
  author={Chen, Sitan and Huang, Brice and Li, Jerry and Liu, Allen and Sellke, Mark},
  booktitle={2023 IEEE 64th Annual Symposium on Foundations of Computer Science (FOCS)},
  pages={391--404},
  year={2023},
  organization={IEEE}
}

@inproceedings{chen2024optimalTradeoffsShadowTomography,
  title={Optimal tradeoffs for estimating pauli observables},
  author={Chen, Sitan and Gong, Weiyuan and Ye, Qi},
  booktitle={2024 IEEE 65th Annual Symposium on Foundations of Computer Science (FOCS)},
  pages={1086--1105},
  year={2024},
  organization={IEEE}
}

@inproceedings{chen2024optimalTradeoffsTomography,
  title={An optimal tradeoff between entanglement and copy complexity for state tomography},
  author={Chen, Sitan and Li, Jerry and Liu, Allen},
  booktitle={Proceedings of the 56th Annual ACM Symposium on Theory of Computing},
  pages={1331--1342},
  year={2024}
}

@inproceedings{diakonikolas2016new,
  title={A new approach for testing properties of discrete distributions},
  author={Diakonikolas, Ilias and Kane, Daniel M},
  booktitle={2016 IEEE 57th Annual Symposium on Foundations of Computer Science (FOCS)},
  pages={685--694},
  year={2016},
  organization={IEEE}
}

@inproceedings{huang2024certifying,
  title={Certifying almost all quantum states with few single-qubit measurements},
  author={Huang, Hsin-Yuan and Preskill, John and Soleimanifar, Mehdi},
  booktitle={Symposium on Foundations of Computer Science (FOCS)},
  pages={1202--1206},
  year={2024},
  organization={IEEE}
}

@article{gupta2025few,
  title={Few Single-Qubit Measurements Suffice
to Certify Any Quantum State},
  author={Gupta, Meghal and He, William and O'Donnell, Ryan},
  journal={arXiv preprint arXiv:2506.11355},
  year={2025}
}

@inproceedings{pelecanos2025beating,
  title={Beating full state tomography for unentangled spectrum estimation},
  author={Pelecanos, Angelos and Tan, Xinyu and Tang, Ewin and Wright, John},
  booktitle={Proceedings of the 2026 Annual ACM-SIAM Symposium on Discrete Algorithms (SODA)},
  pages={3313--3363},
  year={2026},
  organization={SIAM}
}

@article{o2023learning,
  title={Learning and testing quantum states via probabilistic combinatorics and representation theory},
  author={O’Donnell, Ryan and Wright, John},
  journal={Current Developments in Mathematics},
  volume={2021},
  number={1},
  pages={43--94},
  year={2023},
  publisher={International Press of Boston}
}

@article{doosti2026distributed,
  title={Distributed Quantum Property Testing with Communication Constraints},
  author={Doosti, Mina and Sweke, Ryan and Wadhwa, Chirag},
  journal={arXiv preprint arXiv:2604.05962},
  year={2026}
}

@book{goodman2009symmetry,
  title={Symmetry, representations, and invariants},
  author={Goodman, Roe and Wallach, Nolan R and others},
  volume={255},
  year={2009},
  publisher={Springer}
}

@inproceedings{yu2021sample,
  title={Sample efficient identity testing and independence testing of quantum states},
  author={Yu, Nengkun},
  booktitle={12th Innovations in Theoretical Computer Science Conference (ITCS 2021)},
  pages={11--1},
  year={2021},
  organization={Schloss Dagstuhl--Leibniz-Zentrum f{\"u}r Informatik}
}

\end{document}